\documentclass[11pt]{article}
\usepackage{graphicx} 
\usepackage{comment}
\usepackage[margin=1in]{geometry}
\usepackage{amsthm,amsmath,amssymb}
\usepackage{natbib}
\usepackage{rotating}
\usepackage{xcolor}
\usepackage{booktabs}  
\usepackage{array} 
\usepackage{tikz}
\usepackage{bbm}
\usepackage{setspace} 
\usepackage{url}
\usepackage{enumitem}
\usepackage{subcaption}
\usepackage{caption}
\usepackage{bm}
\setlist[enumerate]{leftmargin=*}

\usetikzlibrary{shapes.geometric, arrows}

\tikzstyle{startstop} = [rectangle, rounded corners, minimum width=3cm, minimum height=1cm, text centered, draw=black, fill=red!30]
\tikzstyle{process}   = [rectangle, minimum width=3cm, minimum height=1cm, text centered, draw=black, fill=orange!30]
\tikzstyle{decision}  = [diamond, aspect=2, text centered, draw=black, fill=green!30]
\tikzstyle{arrow}     = [thick,->,>=stealth]

\allowdisplaybreaks

\newtheorem{lemma}{Lemma}

\newtheorem{proposition}{Proposition}
\newtheorem{definition}{Definition}
\newtheorem{assumption}{Assumption}
\newtheorem{remark}{Remark}

\newcommand{\anon}{1}

\newif\ifblind
\blindfalse   

\newcommand{\blindval}[3]{%
  \ifblind
    \expandafter\newcommand\csname #1\endcsname{#3}%
  \else
    \expandafter\newcommand\csname #1\endcsname{#2}%
  \fi
}
\blindval{redaps}{REDAPS }{X }

\date{}

\begin{document}

\if1\anon
{\title{ Randomization inference for stepped-wedge designs with noncompliance with application to a palliative care pragmatic trial}
  \author{Jeffrey Zhang \\ Data Science Institute, University of Chicago \\
    Zhe Chen \\ Center for Clinical Trials Innovation, Department of Biostatistics, \\ Epidemiology and Informatics, University of Pennsylvania \\ 
    Xinyuan Chen \\ Department of Mathematics and Statistics, Mississippi State University \\ 
    Katherine R. Courtright \\ Department of Medicine, University of Pennsylvania \\ 
    Scott Halpern \\ Department of Medicine, University of Pennsylvania \\ 
    Michael O. Harhay \\ Center for Clinical Trials Innovation, Department of Biostatistics, \\ Epidemiology and Informatics, University of Pennsylvania \\
    Dylan S. Small \\ Department of Statistics and Data Science, University of Pennsylvania \\ 
    Fan Li\\ Department of Biostatistics, Yale School of Public Health\\}
  \maketitle
} \fi

\if0\anon
{
  \bigskip
  \bigskip
  \bigskip
  \begin{center}
    {\LARGE Randomization inference for stepped-wedge designs with noncompliance with application to a palliative care pragmatic trial}
\end{center}
  \medskip
} \fi

\begin{abstract}
    While palliative care is increasingly commonly delivered to hospitalized patients with serious illnesses, few studies have estimated its causal effects. \citet{redaps_design} adopted a stepped-wedge cluster-randomized design to assess the effect of palliative care on a patient-centered outcome. The randomized intervention was a nudge to administer palliative care but did not guarantee receipt of palliative care, resulting in noncompliance. A subsequent analysis using methods suited for standard trial designs produced statistically anomalous results, as an intention-to-treat analysis found no effect while an instrumental variable analysis did \citep{Courtright2024}. This highlights the need for a more principled approach to address noncompliance in stepped-wedge designs. We provide a formal causal inference framework for the stepped-wedge design with noncompliance by introducing a relevant causal estimand and corresponding estimators and inferential procedures. Through numerical studies, we compare an array of estimators and provide practical guidance in choosing an analysis method. Finally, we apply our recommended methods to reanalyze the palliative care pragmatic trial, producing point estimates suggesting a larger effect than the original analysis, but intervals that did not reach statistical significance.
\end{abstract}

\vfill

\section{Introduction}\label{sec: X trial}
\subsection{The Randomized Evaluation of Default Access to Palliative Services \\ (\redaps) trial}
Palliative care, a medical approach that prioritizes quality of life, is widely advocated and increasingly available for patients suffering from serious illness. Until recently, there have been few large-scale randomized trials that evaluated the effectiveness of palliative care. In an attempt to fill this gap, \citet{redaps_design} conducted a pragmatic, cluster-randomized stepped-wedged trial to determine whether ordering palliative care by default could improve outcomes among seriously ill patients. The stepped-wedge cluster-randomized design is a one-directional crossover design where clusters are randomized to transition to the intervention condition at different time periods, and remain there until the end of the study. The \redaps trial was conducted from March 2016 to November 2018 and consisted of around 24,000 patients aged 65 or older with advanced diseases from 11 hospitals. The primary outcome of interest was hospital length of stay, measured in log hours. When a hospital transitioned to intervention, all of its patients were ordered palliative care by default, though the patient's clinician had the ability to override the default order, in which case palliative care was not administered. Several baseline covariates were also measured. Based on communication with domain experts, some particularly important baseline covariates included age, Elixhauser score, eligible diagnosis, ICU vs. ward, days between repeated enrollments, and admission source. For details regarding the protocol, institutional review board approval, and informed consent, see \citet{Courtright2024}.

One of the primary analyses conducted by \citet{Courtright2024} was on the intention-to-treat (ITT) basis, which compared the outcomes between the randomized and the default palliative care group and the control group. A mixed-effects linear regression did not find a significant change in hospital length of stay. The point estimate was $-0.0053$, with a 95\% confidence interval of $[-0.0351, 0.0253]$. Moreover, receipt of the randomized intervention to default palliative care was encouraged but did not guarantee receipt of palliative care; in other words, there was noncompliance at the individual level. As is common in randomized trials subject to noncompliance, \citet{Courtright2024} also conducted an instrumental variable analysis, treating randomization to default order as an instrument and whether palliative care was administered as the primary treatment of interest. Since there were no established analysis methods specifically designed for noncompliance in stepped-wedge designs, \citet{Courtright2024} used two-stage least squares methods with time period as a fixed effect and cluster-robust variance estimation. Surprisingly, in contrast to the intention-to-treat analysis, the two-stage least squares analyses respectively found statistically significant effects $-0.096, 95\% \text{ CI } [-0.175,-0.016]$ of the default order intervention.

To understand why this result is peculiar, we briefly review the standard instrumental variable setting. In the canonical instrumental variable context without covariates, where $Y$ denotes outcome, $Z$ is the instrument, and $D$ the treatment, the Wald estimand corresponds to 
\begin{align*}
    \frac{\tau_Y}{\tau_D} \equiv \frac{ \text{average effect of $Z$ on $Y$}}{\text{average effect of $Z$ on $D$}}.
\end{align*}
It would be undesirable for a statistical test for $H_{0,\text{ITT}}: \tau_Y = 0$ to fail to reject while a test for $H_{0,\text{Wald}}: \tau_Y/\tau_D = 0$ rejects, since whenever $\tau_D \neq 0$ the latter null holds if and only if the former does---the two nulls share the numerator estimand $\tau_Y$ and are thus logically equivalent \citep{greevy2004randomization}. In \cite{Courtright2024}, the statistical null hypotheses are not explicitly stated, and the ITT and noncompliance analyses rely on different test statistics and procedures with no enforced coherence between them. While a disagreement between two valid tests can in principle arise from sampling variability or differing efficiency, the more fundamental concern here is the absence of a procedure that respects this logical equivalence between the two nulls. We therefore aim to ground the ITT and Wald null hypotheses in a common causal inference framework, and to develop a principled testing procedure under which failure to reject $H_{0,\text{ITT}}$ precludes rejection of $H_{0,\text{Wald}}$, thereby guaranteeing logical coherence between the two tests by construction.

This peculiar result from \citet{Courtright2024} suggests a need for more in-depth inquiry and new methods for stepped-wedge trial designs when there is noncompliance. It is clear that using standard instrumental variable methods and applying them to stepped-wedge designs with noncompliance can lead to potentially misleading results. Besides noncompliance, the \redaps trial presented additional statistical challenges. First, the \redaps trial employed a one-cluster-per-sequence staggered adoption design, where exactly one cluster crossed over to the intervention at each time period. With only a single cluster crossing over at a time and a limited number of clusters, constructing valid variance estimators is statistically challenging. Second, a cluster-period size is said to be informative when the cluster-period-specific potential outcomes or treatment effects are associated with the number of individuals in that cluster-period. Informative cluster sizes arise naturally in pragmatic trials for several reasons. Larger hospitals may differ systematically from smaller ones in their experience, staffing, or quality of care, and patients enrolled at high- versus low-volume sites may differ in baseline risk and case mix; either mechanism can induce a dependence between cluster-period size and the magnitude of the treatment effect. In the \redaps trial, the cluster-period sizes are highly heterogeneous (Figure \ref{fig: cluster period counts}), ranging from a few dozen to several hundred patients, making informative cluster size a genuine concern. This matters because, as discussed in \citet{Kahan2024}, conventional mixed-model approaches may then target an ambiguous estimand that depends on unknown correlation parameters---even in a parallel-arm cluster-randomized trial---so that the quantity being estimated no longer corresponds to a clearly interpretable average treatment effect. This motivates our estimand-aligned approach, in which the target effect ratio accounting for noncompliance is defined explicitly and a priori, independent of any working model.
\begin{figure}
    \centering
    \includegraphics[width=0.48\linewidth]{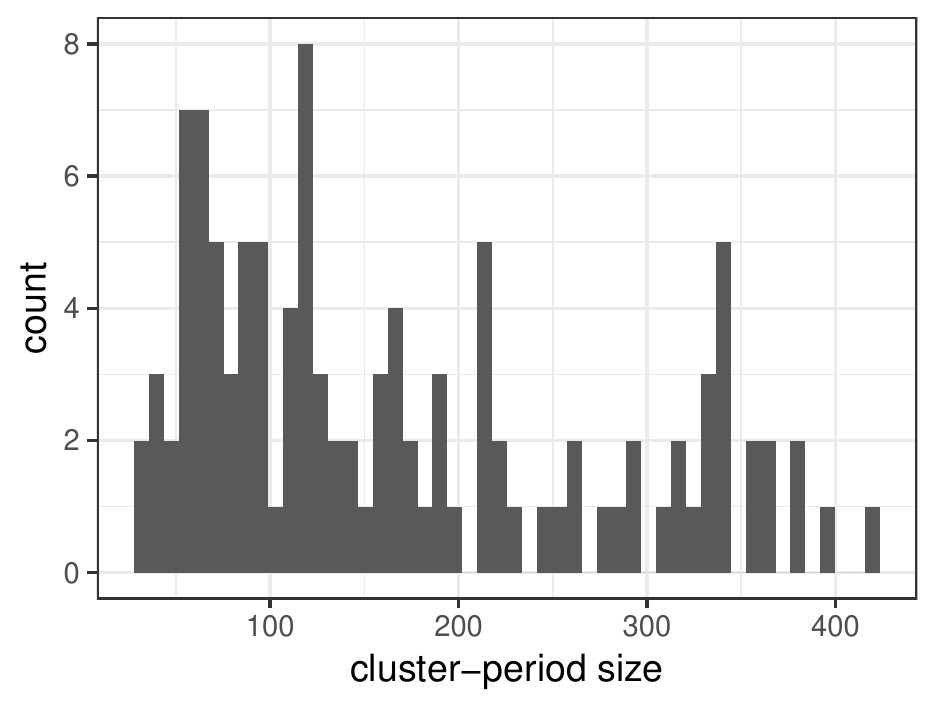}
    \caption{A histogram of the number of individuals in each of the cluster-period cells within the \redaps stepped-wedge cluster-randomized trial.}
    \label{fig: cluster period counts}
\end{figure}

\subsection{Related work and objectives}
Since the introduction of the stepped-wedge design, there have been numerous methodological proposals for statistical analysis. The overwhelming majority of the methods are based on mixed-models \citep{li2021mixed}, where the parameters of interest typically become the model-based coefficients. This is an attractive approach when there is no informative cluster size \citep{fang2025model}, as the treatment effect coefficients (at least in the linear mixed model case) can be interpreted as the average treatment effect estimand \citep{Wang2024}. 
However, a recognized limitation of the mixed model-based approaches is that they may not be well-suited for informative cluster-size situations. Specifically, \citet{Kahan2024} and \citet{lee2025estimated} have pointed out that some of the mixed model-based approaches may perform poorly when the cluster size is informative of the treatment effect. In extensive simulation studies, \citet{fang2025model} have also shown that the mixed-model treatment effect coefficients are often biased for the potential outcomes estimands when the cluster-period sizes are informative in stepped-wedge designs, and have recommended a model-robust standardization procedure to correctly target the causal estimands.

On the other hand, it is increasingly popular, particularly in the recent causal inference literature for cluster-randomized trials \citep{Kahan2024,li2025model}, to take an estimand-aligned approach. In an estimand-aligned approach, the causal estimand of interest is specified a priori, and estimation and inference are developed for that specific estimand. No modeling assumptions are required to define the causal estimands, but models may be used to assist in the estimation. There is a growing body of work taking an estimand-aligned, design-based approach to analyzing cluster-randomized experiments. In parallel cluster-randomized trials, \cite{Su2021Model-AssistedExperiments} study model-assisted estimators that use working regression models to improve efficiency without compromising design-based validity, and \cite{lu2023design} develop design-based theory under cluster rerandomization. A common thread in this literature is that models are leveraged solely for efficiency, with consistency and inference guaranteed by the randomization itself. Our work adopts this same model-assisted philosophy and extends it to the stepped-wedge cluster-randomized design under individual-level noncompliance, a setting arising from the \redaps trial that, to our knowledge, has not previously been addressed within the design-based framework. Several other causal inference methods have emerged under the stepped wedge designs but do not tackle noncompliance. For example, \citet{Ji2017RandomizationInsurance} primarily consider the null of no effect or constant effects at the individual level, and develop a permutation test that remains valid under correlation specification. \citet{zhang2025multiple} propose methods for testing lagged treatment effects. \citet{Chen2025} and \citet{Wang2024} define time-dependent causal estimands of interest and develop model-robust estimation procedures. \citet{fang2025model} further propose model-robust standardization methods to allow for more general working regression models. 

A prevailing theme in these works is the utilization of statistical models purely for improving statistical efficiency, and due to randomization, the estimation and inference for the target estimands are often not compromised under (partial) model misspecification.  Despite these recent developments, there has not yet been sufficient attention to estimand-aligned methods for addressing individual-level noncompliance in stepped-wedge trials.

Motivated by the analytical complexities of the \redaps trial, our objective is to supply the missing methodological foundation for stepped-wedge designs with noncompliance by developing a cohesive suite of assumptions, target estimands, and design-based inference procedures. We highlight four primary contributions. First, we develop a unified design-based framework that grounds intention-to-treat and instrumental-variable (Wald) analyses in a common finite-population causal model, with the effect ratio as the central estimand. Within this framework, we construct a testing procedure under which failure to reject the ITT null logically precludes rejection of the Wald null, resolving by construction the contradictory ITT and IV conclusions reported in the original analysis of \cite{Courtright2024}. Second, we develop model-assisted ANCOVA estimators for the effect ratio through a residualization of the outcome, and further provide cluster-period average and scaled cluster-period total ANCOVA variants. Third, we study Horvitz–Thompson estimators of the effect ratio, for which we establish a design-based central limit theorem that, to our knowledge, is new for stepped-wedge designs and applies equally in the absence of noncompliance, together with a provably conservative variance estimator that remains well defined even in the one-cluster-per-sequence design, where only a single cluster crosses over per period. Finally, we provide practical guidance for implementing these methods. Through extensive simulations, we evaluate the finite-sample performance of our proposed methods along with variance estimators utilizing small-sample adjustments. We summarize our recommendations in a decision flowchart and conclude with a comprehensive reanalysis of the \redaps trial.

\section{Preliminaries} \label{section: prelims}
Before developing our methodology, we briefly recap the motivating application and crystallize the scientific questions that drive it. Recall from Section \ref{sec: X trial} that the \redaps trial is a pragmatic, stepped-wedge cluster-randomized trial in which hospitals transitioned at randomized times to ordering palliative care by default, with the default order encouraging but not guaranteeing receipt of palliative care. Among patients for whom a default order fired, 43.9\% received a consultation compared to 16.6\% among patients without a default order \citep{Courtright2024}. This application raises three questions of central interest. First, what is the effect of the randomized default order itself on patient-centered outcomes---hospital length of stay and 30-day readmission---that is, the intention-to-treat effect? Second, what is the causal effect of receiving palliative care among patients whose receipt is determined by the default order, after properly accounting for noncompliance? Third, can these two questions be addressed within a single, coherent inferential framework that avoids the contradictory conclusions reported in the original analysis \citep{Courtright2024}? These questions are clinically consequential given the expanding role of palliative care for seriously ill patients. The remainder of this section introduces the notation, potential-outcomes framework, and assumptions needed to formalize and answer them.

\subsection{Notation and assumptions}

We pursue potential outcomes estimands under a finite-population framework, where all potential outcomes are fixed, and the randomness in the observed data comes from the treatment assignment alone. We consider a trial with $I$ clusters (indexed by $i$) and $J + 2$ periods (indexed by $j$), where period $0$ is the pre-rollout and period $J + 1$ is the post-rollout. In the \redaps trial, $I = 11$ and $J = 10$. We use $Z_{ij} \in \{0, 1\}$ to denote randomized intervention status (ordering palliative care by default) for cluster $i \in \{1, \ldots, I\}$ in rollout period $j \in \{0, 1, \ldots, J, J + 1\}$. The number of clusters in each rollout period $j$ randomized to intervention, $I_j$, is known and fixed a priori. The rollout starts with $I_1$ clusters randomized in period 1, and in period 2, the previously selected $I_1$ clusters remain under intervention, while $I_2 - I_1$ out of the $I - I_1$ remaining clusters are randomized to intervention, and so forth. This process continues until all clusters receive the intervention in period $J + 1$. Thus, for all $i$, $Z_{i0} = 0$ and $Z_{i,J+1} = 1$ by definition, and $0 = I_0 < I_1 \leq I_2 \leq \ldots \leq I_J < I_{J+1} = I$. In the \redaps trial, at each time point, only one hospital crossed over, i.e., $I_j = j$ for all $j$. There are $N_{ij}$ individuals (indexed by $k$) in cluster-period $(i, j)$ and $N_j \equiv \sum_{i=1}^I N_{ij} $ is the total number of individuals in period $j$. 

We proceed to define potential outcomes. We let $A_i = a \in \mathcal{A} = \{1,\ldots, J, J +1\}$ denote the period index for which cluster $i$ is first randomized to intervention (the so-called adoption time), and therefore $Z_{ij} = \mathbbm 1(A_i \leq j)$. Without any further assumptions, each individual has corresponding potential outcomes $Y_{ijk}^{a_i,\bm{a}_{-i}}$ that depend not only on the adoption time for cluster $i$ but also those for other clusters ($\bm{a}_{-i}$ is the vector of adoption times for clusters other than cluster $i$). Note that the randomized intervention $Z$ (making a default palliative care order) may not correspond to the true treatment of interest, due to noncompliance. Thus, we introduce potential outcomes $D_{ijk}^{a_i,\bm{a}_{-i}}$ for $D$ representing the treatment of interest, i.e., actual receipt of palliative care in the \redaps trial. For each individual $k$ included from cluster $i$ during period $j$ (under a cross-sectional sampling design), we define $X_{ijk}$ as the vector of baseline covariates that are measured prior to exposure to the treatment assignment. Appendix Figure \ref{fig:stepped} presents an example diagram for a stepped-wedge design with the observed data structure. Next, we introduce three assumptions, which were initially introduced in \citet{Chen2025}. These assumptions are standard for stepped-wedge designs in the absence of noncompliance, and new assumptions necessary for studying noncompliance will be introduced in Section \ref{ER_estimand}.

\begin{assumption}[Cluster-level SUTVA] \label{assumption: cluster sutva}
    Let $Y_{i j k}^{a_i, \boldsymbol{a}_{-i}}$ and $D_{i j k}^{a_i, \boldsymbol{a}_{-i}}$ denote the potential outcomes for an individual given the adoption time of all clusters, then $Y_{i j k}^{a_i, \boldsymbol{a}_{-i}}=Y_{i j k}^{a_i, \boldsymbol{a}_{-i}^*}, \text{ and } D_{i j k}^{a_i, \boldsymbol{a}_{-i}}=D_{i j k}^{a_i, \boldsymbol{a}_{-i}^*}, \forall \boldsymbol{a}_{-i} \neq \boldsymbol{a}_{-i}^*$.
\end{assumption}

Assumption \ref{assumption: cluster sutva} requires that the potential outcomes of individuals in one cluster are not dependent on the assignment of a differing cluster, i.e., there are no spillover effects across clusters. This is a standard assumption made for cluster-randomized trials and is plausible in the \redaps trial because the 11 hospitals in the trial did not interact. Under Assumption \ref{assumption: cluster sutva}, we can write potential outcomes as $Y_{i j k}^{a_i}$ and $D_{i j k}^{a_i}$ without loss of generality, and in ensuing sections, we also drop the subscript $i$, writing instead $Y_{i j k}^{a}$ and $D_{i j k}^{a}$. We next introduce an assumption that requires that the potential outcomes of an individual in a cluster do not depend on how long the cluster has been in the intervention phase. 

\begin{assumption}[Intervention duration irrelevance] \label{assumption: treat dur irrel}
    There is only a single version of randomized intervention across different periods, such that variations in intervention duration are irrelevant to the potential outcomes. That is, for each $k \in\left\{1, \ldots, N_{i j}\right\}$, (i) $Y_{i j k}^a=Y_{i j k}^{a^{\prime}}=Y_{i j k}(1) \text{ and } D_{i j k}^a=D_{i j k}^{a^{\prime}}=D_{i j k}(1)$, if $\max \left\{a, a^{\prime}\right\} \leq j$; (ii) $Y_{i j k}^a=Y_{i j k}^{a^{\prime}}=$ $Y_{i j k}(0) \text{ and } D_{i j k}^a=D_{i j k}^{a^{\prime}}=$ $D_{i j k}(0)$, if $\min \left\{a, a^{\prime}\right\}>j$; (iii) $Y_{i j k}^a=Y_{i j k}(1), Y_{i j k}^{a^{\prime}}=Y_{i j k}(0), D_{i j k}^a=D_{i j k}(1), D_{i j k}^{a^{\prime}}=D_{i j k}(0)$, if $a \leq j<a^{\prime}$, for $a, a^{\prime} \in \mathcal{A}$ and $j \in\{0,1, \ldots, J, J+1\}$.
\end{assumption}

Assumption \ref{assumption: treat dur irrel} requires that an individual's potential outcomes (for both $D$ and $Y$) do not depend on how long their cluster has been randomized (or not randomized) to intervention. Intervention duration irrelevance is plausible in settings where there is no learning effect or a weakening effect over time \citep{Kenny2022, Wang2024}, and it greatly reduces the number of distinct potential outcomes. This is a reasonable starting point for exploring noncompliance in stepped-wedge designs and is the assumption adopted for the primary analysis of \redaps. In the context of the \redaps trial, this would require that there are no significant changes in how effective the default order and the administration of palliative care are over time. In the context of staggered adoption settings, this assumption has been referred to as \emph{invariance to history} \citep{Athey2022}. As \citet{Athey2022} argued, this type of assumption is more plausible when units of study are different across the time periods. In the \redaps trial, if the palliative care teams got better over time at managing the increased demand resulting from the default order, such that they would not have gotten to some patient on an earlier date but would have on a later date, this would constitute a learning effect (on the $D$ potential outcomes). On the other hand, it is also possible that palliative care teams get burnt out over time following initiation of the default order, such that the odds of seeing the patient go down over time, which would constitute a weakening effect (on the $D$ potential outcomes). In Section \ref{section: data analysis}, as a heuristic check, we test for an association between the amount of time a patient's hospital has (or has not) initiated the default order and a patient's receipt of palliative care and length of stay. We perform the tests separately for patients who received the default order and those who did not. Ultimately, we did not find statistically significant evidence of an association, giving some reassurance on the plausibility of treatment duration irrelevance in the \redaps trial.

As a result of Assumptions \ref{assumption: cluster sutva} and \ref{assumption: treat dur irrel}, for rollout periods $j$, we may write $Y_{ijk} = Z_{ij}Y_{ijk}(1) + (1 - Z_{ij} )Y_{ijk}(0)$ and $D_{ijk} = Z_{ij}D_{ijk}(1) + (1 - Z_{ij} )D_{ijk}(0)$. To complete the data specification, we assume the cluster-period size $N_{ij}$ to be unaffected by the intervention, thus ruling out post-randomization selection bias \citep{Li2022a}. We further define the number of individuals in cluster $i$ as $N_i =\sum_{j=1}^J N_{ij}$ and the total number of individuals across all clusters and periods as $ N = \sum_{i=1}^I \sum_{j=1}^J N_{ij}$.

\begin{assumption}[Stepped wedge randomization] \label{assumption: sw randomization}
    Write $\mathcal{Y}$, $\mathcal{D}$, and $\mathcal{X}$ as the collection of all potential outcomes and covariates across individuals and cluster-periods, then
    \begin{align*}
        \mathbb P\left(\bm{Z}_i=\bm{z} | \mathcal{Y}, \mathcal{D}, \mathcal{X}\right)=\left(\begin{array}{c}
        I \\
        I_1, I_2-I_1, \ldots, I_{J+1}-I_J
        \end{array}\right)^{-1},
    \end{align*}
    where $Z_{i 0}=0$ and $Z_{i, J+1}=1$ deterministically for all clusters $i$. Moreover, for $j = 1,\ldots, J$, the cluster-level propensity scores $I_j/I$ are strictly bounded away from 0 and 1.
\end{assumption}

This assumption entails that the randomization conducted does, in fact, follow the stepped-wedge design. Importantly, we write $e_j = I_j/I$ as the cluster-level propensity score fixed by design, and have $0 = e_0 < e_1 \leq e_2 \leq \ldots \leq e_J < e_{J+1} = 1$. The cluster-level propensity scores for period $0$ and $J+1$ are $e_0 = 0$ and $e_{J+1} = 1$, respectively. Thus, positivity for the pre- and post-rollout period is violated. Without further assumptions, it is therefore not possible to draw causal inferences about the pre and post-rollout individuals. This leads us to focus on estimands defined only on the rollout period individuals, i.e., individuals in periods $j \in \{1,\ldots J\}$. 

\subsection{Causal estimands}\label{ER_estimand}

We introduce two estimands of interest, which are the sample average treatment effect of the randomized intervention and the effect ratio, aggregated over all clusters, rollout time periods, and individuals. We call the first $\tau_Y$ and the second $\lambda$.

\begin{definition}
    The \emph{sample average intention-to-treat effect} of the randomized intervention is $$\tau_Y \equiv \frac{1}{N}\sum_{i=1}^I \sum_{j=1}^J \sum_{k = 1}^{N_{ij}} \{Y_{ijk}(1)-Y_{ijk}(0)\}.$$
\end{definition}

The $\tau$ estimand has been previously considered by \citet{Chen2025}, and they propose several analysis of covariance (ANCOVA) estimators with a working independence assumption for targeting such an estimand. The interpretation of the $\tau_Y$ estimand is relatively straightforward, as it is a sample average of treatment effects, where the average gives equal weight to each \emph{individual} in the study. Since we are in a noncompliance setting, it is also important to introduce an estimand more closely aligned with the actual treatment received, not just the randomized intervention. For this, we introduce a second estimand, which is closely related to the Wald estimand in the standard instrumental variable setting. In the stepped-wedge design, the randomized intervention $Z$ can be viewed as a valid instrumental variable.

\begin{definition}
\label{def: effect ratio}
    The \emph{effect ratio} is defined as 
    \begin{align*}
         \lambda \equiv \frac{\sum_{i = 1}^I \sum_{j = 1}^J \sum_{k = 1}^{N_{ij}} \{Y_{ijk}(1)-Y_{ijk}(0)\}} {\sum_{i = 1}^I \sum_{j = 1}^J \sum_{k = 1}^{N_{ij}} \{D_{ijk}(1)-D_{ijk}(0)\}}.
    \end{align*}
\end{definition} 

The effect ratio, $\lambda$, (and related variants) has been considered in many works for various study designs \citep{Baiocchi2010, Kang2018, Fogarty2021Biased, Zhang2022BridgingRatio, Park2023, pashley2024iv, chen2025manipulating}, though not yet in the stepped-wedge design. It is exactly the ratio of two sample average treatment effects, namely, the effect of randomization on the outcome and on the treatment. Under certain assumptions (introduced below) on the randomized intervention $Z$, the effect ratio has an intuitive interpretation as the sample average causal effect among the complier subpopulation \citep{Kang2018}, also referred to as the local average treatment effect \citep{schochet2024design, aronow2024randomization}. 

\begin{assumption}[Instrument variable assumptions] \label{assumption: iv} 1. Exclusion restriction: $D_{ijk}(1) = D_{ijk}(0)$ implies $Y_{ijk}(1) = Y_{ijk}(0)$. 2. Monotonicity: $D_{ijk}(1) \geq D_{ijk}(0) \ \forall i,j,k$. 3. Relevance: $\tau_D \equiv N^{-1}\sum_{i = 1}^I \sum_{j = 1}^J \sum_{k = 1}^{N_{ij}} \{D_{ijk}(1)-D_{ijk}(0)\} \neq 0$.
\end{assumption}

In the \redaps trial, the exclusion restriction is plausible because the default order was designed to influence patient outcomes primarily by prompting palliative care consultations, and alternative pathways from the assignment to outcomes, such as changes in clinician behavior absent a consultation, were expected to be minimal. We also expect monotonicity to hold, as there should have been no patients who would not receive palliative care if assigned to receive it by default but would receive it if assigned to usual care \citep{Courtright2024}. Finally, the quantity $\tau_D = N^{-1}\sum_{i = 1}^I \sum_{j = 1}^J \sum_{k = 1}^{N_{ij}} \{D_{ijk}(1)-D_{ijk}(0)\}$ is the effect of default order on palliative care receipt. Given the strong association between default order and palliative care receipt, relevance almost certainly holds. Under Assumption \ref{assumption: iv}, $\lambda$ can be interpreted as the sample average treatment effect of palliative care $D$ among the complier population (individuals for which receipt of palliative care or not matches the default order, i.e., $D_{ijk}(1) =1, D_{ijk}(0) = 0$). Intuitively, due to monotonicity and relevance, the denominator of $\lambda$ (scaled by $N$) measures the proportion of compliers and is nonzero. Meanwhile, the numerator (scaled by $N$) measures the average effect of randomization $Z$, which, due to the exclusion restriction and monotonicity, is the effect of the treatment $D$ in the compliers, but scaled by $N$ rather than the number of compliers. Note that the exclusion restriction and monotonicity are \emph{not} required for the inference procedures for $\lambda$ presented in the ensuing sections to be valid. However, they ensure that $\lambda$ has a natural, intuitive interpretation. In the next section, we introduce methods for estimation and inference for the effect ratio $\lambda$.

\section{Estimation and inference under noncompliance} \label{section: estimation and inference}

The statistical procedures we propose are built upon the following simple observation. Consider the null hypothesis $H_0: \lambda = \lambda_0$. Under $H_0$, since $\lambda$ is a ratio of sample average treatment effects, a simple rearrangement from Definition \ref{def: effect ratio} yields 
\begin{equation*}
\begin{aligned}
    \tau_{\lambda_0} &\equiv \frac{1}{N} \sum_{i=1}^I \sum_{j=1}^J \sum_{k=1}^{N_{ij}} \{Y_{ijk}(1) - \lambda_0D_{ijk}(1)\} - \{Y_{ijk}(0) - \lambda_0D_{ijk}(0)\}  \\ &= \frac{1}{N} \sum_{i=1}^I \sum_{j=1}^J \sum_{k=1}^{N_{ij}} \{Y_{ijk}(1)  - Y_{ijk}(0)\} - \lambda_0\{D_{ijk}(1) - D_{ijk}(0)\} = 0.
\end{aligned}
\end{equation*}
Essentially, under the null, the sample average treatment effect of the randomized intervention $Z$ on the residualized outcome $Y - \lambda_0 D$ ($\tau_{\lambda_0}$) is zero. Equipped with this observation, testing the null $H_0$ is equivalent to testing $\tau_{\lambda_0} = 0$, which is no more difficult than testing $\tau = 0$. As a result, we can leverage already developed tools for estimation and inference about $\tau$ to estimate and make inferences for $\lambda$. More specifically, we leverage the unadjusted, ANCOVA-I, and ANCOVA-III estimators proposed in \citet{Chen2025}, essentially utilizing them as estimators of $\tau_{\lambda_0}$. We also provide alternative Horvitz-Thompson estimators for $\tau_{\lambda_0}$. We emphasize that although we leverage ANCOVA (and other) working models to construct estimation and inference procedures for the finite-population effect ratio $\lambda_0$, for our results, we never assume that the working regression models are correctly specified. Always working under the design-based perspective, we only leverage models to perform regression adjustment and potentially improve efficiency. In this sense, our procedures can be viewed as model-assisted rather than model-based \citep{Su2021Model-AssistedExperiments}.

Furthermore, our theoretical results are derived in an asymptotic regime where the number of periods $J$ is fixed and the number of clusters $I \to \infty$. While a one-cluster-per-sequence design formally couples the two via $J = I - 1$, this coupling is an artifact of fixing the number of clusters per sequence at one, rather than an intrinsic feature of how stepped-wedge trials are conducted in practice. The total duration of a trial, and hence the number of periods, is typically dictated by external constraints such as the funding period, and is fixed at the planning stage; when additional resources permit enlarging the trial, investigators generally enroll additional clusters within the existing sequence structure rather than adding new periods. Our asymptotic regime is meant to capture exactly this scenario: the number of sequences and periods is held fixed by design, while the number of clusters per sequence grows, allowing us to characterize the limiting behavior of the estimators under realistic trial expansion. The one-cluster-per-sequence design is then best understood as a small-sample instance of this regime, with a single cluster in each of the $J$ sequences. Consistent with this view, reviews of published stepped-wedge trials have found that the number of periods rarely exceeds a handful, regardless of the number of clusters \citep{nevins2024adherence}. On the other hand, a regime in which $I$ and $J$ are both fixed may also be of interest, but non-asymptotic inference for sample average treatment effects has only recently been explored even in the simplest settings, such as individual complete or Bernoulli randomization with unadjusted Horvitz–Thompson estimators \citep{sandoval2026nonasymptotic}. We therefore adopt the $I \to \infty$, $J$ fixed regime, and provide empirical evidence that our methods perform well even in the challenging (small sample) one-cluster-per-sequence setting.

\subsection{Model-assisted ANCOVA estimators}

Following \citet{Chen2025}, we first introduce ANCOVA estimators for $\tau_{\lambda_0}$. The unadjusted, ANCOVA-I, and ANCOVA-III estimators perform the following regressions, respectively:
\begin{equation}
\label{eq}
\begin{aligned}
    &Y_{ijk} - \lambda_0 D_{ijk} \sim \beta_j + \theta_j Z_{ij}, \\
    &Y_{ijk} - \lambda_0 D_{ijk} \sim \beta_j + \theta_j Z_{ij} + X_{ijk}^c \eta, \\
    &Y_{ijk} - \lambda_0 D_{ijk} \sim \beta_j + \theta_j Z_{ij} + X_{ijk}^c\gamma + Z_{ij}X_{ijk}^c \eta,
\end{aligned}
\end{equation}
where $X_{ijk}^c = X_{ijk}-N_j^{-1}\sum_{i=1}^I\sum_{k=1}^{N_{ij}}X_{ijk}$. The unadjusted regression includes period-specific intercepts and period-specific randomized intervention coefficients. The ANCOVA-I regression includes period-specific intercepts, period-specific randomized intervention coefficients, and covariate fixed effects. The ANCOVA-III regression includes period-specific intercepts,  period-specific randomized intervention coefficients, and period-specific randomized intervention by covariate interactions. The models are increasing in complexity (number of parameters), and Table \ref{tab: ancova comparison} gives a summary. We focus on the latter two estimators as they performed well empirically in \citet{Chen2025}, by balancing finite-sample stability with model complexity.

Each of the resulting fitted regressions will produce period-specific coefficient estimates $\widehat{\theta}_j(\lambda_0)$. Then, given these estimates, one can compute, for each regression specification, the estimator $\widehat{\tau}_{\lambda_0} \equiv N^{-1}\sum_{j=1}^J N_j \widehat{\theta}_j(\lambda_0)$. In addition, conservative standard error estimators $S(\lambda_0)$ for $\widehat{\tau}_{\lambda_0}$ can be constructed in a similar manner as in Section 3.3 of \citet{Chen2025}. It is then straightforward to conduct a test of $H_0: \lambda = \lambda_0$ by comparing the deviate $\widehat{\tau}_{\lambda_0} / S({\lambda_0})$ to a standard normal distribution, which results in an asymptotically valid test. Confidence intervals can be readily computed by test inversion.

\begin{proposition} \label{prop: ancova test validity}
    Consider an asymptotic regime where $I \to \infty$, and $J$ is a fixed constant. Under Assumptions \ref{assumption: cluster sutva}-\ref{assumption: sw randomization} and regularity conditions to be introduced in the Appendix, for a fixed $\lambda_0$, $\widehat{\tau}_{\lambda_0} \stackrel{p}{\to} \tau_{\lambda_0}$. In addition, under the null $H_0: \lambda = \lambda_0$,
    \begin{equation}
        \begin{aligned}
            &\limsup_{I \to \infty} \mathbb P\{\widehat{\tau}_{\lambda_0} / S({\lambda_0}) \geq k | \mathcal{\mathcal{Y}, \mathcal{D}, \mathcal{X}} \} \leq 1 - \Phi(k), \\ 
            &\limsup_{I \to \infty} \mathbb P\{\widehat{\tau}_{\lambda_0} / S({\lambda_0}) \leq -k | \mathcal{\mathcal{Y}, \mathcal{D}, \mathcal{X}} \} \leq \Phi(-k),
        \end{aligned}
    \end{equation}
    for any $k > 0$, where $\Phi(\cdot)$ denotes the standard Gaussian cumulative distribution function.
\end{proposition}

The regularity conditions (see Appendix \ref{appendix: ancova} for more discussion) are adapted from \cite{Chen2025}, and essentially restrict the tail behavior of potential outcomes and require the existence of certain moments and limits. Proposition \ref{prop: ancova test validity} adapts the model-assisted ANCOVA theory of \cite{Chen2025} to the noncompliance setting: the key step is the observation that, for fixed $\lambda_0$, testing $H_0: \lambda = \lambda_0$ reduces to a sample average treatment effect null for the residualized outcome $Y_{ijk} - \lambda_0 D_{ijk}$, which lets us repurpose their estimators and design-based variance estimators (originally developed for a single treatment effect estimand) for inference on the ratio estimand $\lambda$. Consistency and conservative inference then follow by applying their Theorem 1 to the residualized outcomes. What is specific to our setting is that inference for $\lambda$ proceeds by inverting a family of such tests indexed by $\lambda_0$, which in turn yields the coherence property formalized in Remark \ref{rmk:coherence}. As a result of Proposition \ref{prop: ancova test validity}, an asymptotically valid level-$\alpha$ test for $H_0: \lambda = \lambda_0$ rejects when $|\widehat{\tau}_{\lambda_0}/S(\lambda_0)| \geq z_{1-\alpha/2}$. Moreover, a natural point estimate $\widehat{\lambda}$ for $\lambda$ is simply the solution to the equation $\widehat{\tau}_{\widehat{\lambda}} \equiv N^{-1}\sum_{j=1}^{J} N_j \widehat{\theta}_j(\lambda) = 0$.

\begin{remark}\label{rmk:coherence}
    Observe that by design, the test of $H_0: \lambda = 0$ is exactly equivalent to testing $H_0: \tau_0 = \tau = 0$. Thus, as long as the same point/variance estimators for $\tau_0 = \tau$ are used, the test of $H_0: \lambda = 0$ rejects if and only if the test of $H_0: \tau_0 = \tau = 0$ rejects. As a result, peculiar cases where an intention-to-treat analysis fails to reject but the IV-based noncompliance analysis rejects are not possible using our procedures.
\end{remark}

Proposition \ref{prop: ancova test validity} is derived under an asymptotic regime where the number of clusters $I$ tends to infinity and the number of periods $J$ remains fixed. We view this as a reasonable starting point for theoretical analysis, though we point out that the standard error estimators $S(\lambda_0)$ introduced by \citet{Chen2025} are not well-defined when only one cluster crosses over at each time point $(I_j = 1)$, as they involve division by $I_j - 1$. (A simple modification is to replace $I_j-1$ by $I_j$, and the asymptotic results remain valid). Nevertheless, to avoid the division-by-zero issue and account for the small number of clusters, we propose two additional modifications. First, we utilize an approximate leave-one-cluster-out jackknife variance estimator \citep{Bell2002}, as implemented using the `CR3' option in the \texttt{clubSandwich} package. We denote the standard deviation estimator using CR3 by $S(\lambda)^\mathrm{CR3}$. In addition, instead of comparing the standardized deviate to a Gaussian quantile, we compare to a $t$-distribution with $I - 2$ degrees of freedom to improve finite-sample performance. Though we do not provide formal theoretical justification for these modifications, they appear to perform well in our simulations and are easily implemented using existing software. Finally, we construct confidence intervals for $\lambda$ by inverting hypothesis tests. Specifically, we find the values for $\lambda$ that solve  $\widehat{\tau}_{\lambda} / S({\lambda})^\mathrm{CR3} = \pm t_{1-\alpha/2, I-2}$.

\begin{remark}
In parallel cluster randomized experiments, \cite{Su2021Model-AssistedExperiments} introduce cluster average and scaled cluster total regression adjustment, demonstrating their advantages over individual-level adjustment. Analogs of such regression adjustments are possible for the ANCOVA estimators in the stepped-wedge design. Specifically, the ANCOVA-I and ANCOVA-III cluster-period average estimators perform the following regressions with weights $N_{ij}$, respectively: $\overline Y_{ij} - \lambda_0 \overline D_{ij} \sim \beta_j + \theta_j Z_{ij} + \overline X_{ij}^c \eta$ and $\overline Y_{ij} - \lambda_0 \overline D_{ij} \sim \beta_j + \theta_j Z_{ij} + \overline X_{ij}^c\gamma + Z_{ij}\overline X_{ij}^c \eta$, 
where $\overline Y_{ij} = N_{ij}^{-1}\sum_{k=1}^{N_{ij}} Y_{ijk}$ and $\overline D_{ij} = N_{ij}^{-1}\sum_{k=1}^{N_{ij}} D_{ijk}$ are cluster-period average outcomes and treatments, and $\overline X_{ij}^c=N_{ij}^{-1}\sum_{k=1}^{N_{ij}}X_{ijk}-N_j^{-1}\sum_{i=1}^I\sum_{k=1}^{N_{ij}}X_{ijk}=\overline X_{ij} - \overline X_j$ is the centered cluster-period average of covariates. Meanwhile, the ANCOVA-I and ANCOVA-III scaled cluster-period total estimators perform the following regressions, respectively: $\widetilde Y_{ij} - \lambda_0 \widetilde D_{ij} \sim \beta_j + \theta_j Z_{ij} + \widetilde X_{ij}^c \eta$ and $\widetilde Y_{ij} - \lambda_0 \widetilde D_{ij} \sim \beta_j + \theta_j Z_{ij} + \widetilde X_{ij}^c\gamma + Z_{ij}\widetilde X_{ij}^c \eta$, 
where $\widetilde Y_{ij} = IN_{ij}\overline Y_{ij}/N_j$, $\widetilde D_{ij} = IN_{ij}\overline D_{ij}/N_j$, and $\widetilde X_{ij}^c=IN_{ij}\overline X_{ij}^c/N_j$ are rescaled cluster-period average outcomes, treatments, and covariates.
We more formally introduce these model variants and prove their asymptotic properties in Appendix \ref{appendix: aggregate estimators}, showing that they remain model-assisted estimators in stepped-wedge designs. While Proposition \ref{prop: ancova test validity} can be viewed as a direct consequence of results of \cite{Chen2025} applied to the noncompliance setting, the theoretical results in Appendix \ref{appendix: aggregate estimators} require significant, non-trivial extensions.
\end{remark}

\subsection{Horvitz-Thompson estimators}

As an alternative to the ANCOVA-style estimators from the previous section, we also outline Horvitz-Thompson (HT) style estimators for $\tau_{\lambda_0}$, with technical details presented in the Appendix. Consider the following estimator for $\tau_{\lambda_0}$:
\begin{equation}
    \widehat{\tau}_{\lambda_0, \text{HT}} \equiv \frac{1}{N} \sum_{i=1}^I \sum_{j=1}^J \sum_{k=1}^{N_{ij}} \frac{Z_{ij}}{e_{ij}}(Y_{ijk} - \lambda_0D_{ijk}) - \frac{1-Z_{ij}}{1-e_{ij}}(Y_{ijk} - \lambda_0D_{ijk}) = \widehat{\tau}_{1, \lambda_0, \text{HT}} - \widehat{\tau}_{0, \lambda_0, \text{HT}}.
\end{equation}
Here the $e_{ij}\equiv \mathbb P(Z_{ij} = 1 | \mathcal{Y},\mathcal{D},\mathcal{X})$ are the known propensity scores. This estimator simply contrasts intervened and non-intervened cluster-period residualized outcomes, weighting by the inverse of the treatment propensity score at that time period. It is straightforward to verify that the Horvitz-Thompson estimator is unbiased under randomization, as formalized below.

\begin{lemma} \label{lemma: ht unbiased}
    $\mathbb E(\widehat{\tau}_{\lambda_0,\text{HT}} | \mathcal{\mathcal{Y}, \mathcal{D}, \mathcal{X}} ) = \tau_{\lambda_0}$.
\end{lemma}

To conduct inference, it is necessary to compute a variance estimate. Though the true variance of the estimator depends on inherently inestimable quantities (i.e., products of treated and control potential outcomes, which are never observed simultaneously), it is possible to construct a conservative variance estimate, following \citet{Aronow2017}, for example. The variance estimate has a complicated-looking form, so we defer the details to the Appendix. We call this variance estimator $\widehat{\text{Var}}(\widehat{\tau}_{\lambda_0,\text{HT}})$. We point out that this variance estimator (unlike the ANCOVA variance estimators in \citet{Chen2025}) is well-defined even in the one-cluster-per-sequence stepped wedge design, like in the \redaps trial. Also, the variance estimator is conservative in the following sense:
\begin{lemma}
    \label{lemma: conservative variance}
    $\mathbb E\{\widehat{\mathrm{Var}}(\widehat{\tau}_{\lambda_0,\text{HT}}) | \mathcal{\mathcal{Y}, \mathcal{D}, \mathcal{X}} \} \geq {\mathrm{Var}}(\widehat{\tau}_{\lambda_0,\text{HT}} | \mathcal{\mathcal{Y}, \mathcal{D}, \mathcal{X}})$.
\end{lemma}

The conservativeness of the variance estimator is not a flaw, as it is well-known in the finite population causal inference literature that variances are often not consistently estimable due to being a function of products of distinct potential outcomes, which are never observed simultaneously. In the Appendix, we also present a variance estimator that ignores such terms and thus is no longer guaranteed to be conservative, though it can perform reasonably well in practice, as will be seen in the simulation studies in Section \ref{section: sims}. Finally, let $ S_{HT}({\lambda_0}) \equiv \sqrt{\widehat{\text{Var}}(\widehat{\tau}_{\lambda_0,\text{HT}})}$. Appealing to a normal approximation, we can show that the deviate can be used as a basis for testing the null $H_0: \lambda = \lambda_0$.

\begin{proposition}
    \label{prop: ht test validity}
    Consider an asymptotic regime where $I \to \infty$, and $J$ is a fixed constant. Under Assumptions \ref{assumption: cluster sutva}-\ref{assumption: sw randomization} and regularity conditions to be introduced in the Appendix, for a fixed $\lambda_0$, $\widehat{\tau}_{\lambda_0, HT} \overset{p}{\to} \tau_{\lambda_0}$. In addition, under the null $H_0: \lambda = \lambda_0$,
    \begin{equation}
        \begin{aligned}
            &\limsup_{I \to \infty} \mathbb P\{\widehat{\tau}_{\lambda_0, HT} / S_{HT}({\lambda_0}) \geq k | \mathcal{\mathcal{Y}, \mathcal{D}, \mathcal{X}} \} \leq 1-\Phi(k), \\ 
            &\limsup_{I \to \infty} \mathbb P\{\widehat{\tau}_{\lambda_0, HT} / S_{HT}({\lambda_0}) \leq -k | \mathcal{\mathcal{Y}, \mathcal{D}, \mathcal{X}} \} \leq \Phi(-k),
        \end{aligned}
    \end{equation}
    where $\Phi(\cdot)$ denotes the standard Gaussian cumulative distribution function.
\end{proposition}

The regularity conditions essentially restrict the tail behavior of potential outcomes and impose boundedness restrictions (Appendix \ref{appendix: horvitz-thompson}). One may notice that, as in Proposition \ref{prop: ancova test validity}, Proposition \ref{prop: ht test validity} applies to an asymptotic regime where the number of clusters $I$ tends to infinity and the number of periods $J$ remains fixed. We view this as a reasonable starting point for theoretical analysis, and simulations confirm that the procedure can perform well even in small cluster regimes where the number of clusters is of the same order as the number of periods.

It is well-known that Horvitz-Thompson estimators can suffer from large variance, especially when propensity scores $e_j, 1-e_j$ are small. This is exacerbated by potentially heterogeneous cluster-period sizes $N_{ij}$. Both are concerns in the \redaps trial, as the smallest propensities are 1/11, and the cluster-period sizes vary greatly. Thus, it is of interest to utilize regression adjustment and suitably augment the Horvitz-Thompson estimator. We now outline how such estimators can be constructed, and in so doing propose a novel way for the pre- and post-rollout period data to be utilized. This is in contrast to the ANCOVA estimators, which discard the pre- and post-rollout period data due to non-positivity.

We first introduce some additional notation. Specifically, we let $\tau_{a,\lambda_0} = N^{-1}\sum_{i} \sum_{j} \sum_k Y_{ijk}(a)-\lambda_0 D_{ijk}(a)$ for $a = 0,1$. Also, let $\overline{R}_{ij}(1) = \sum_{k = 1}^{N_{ij}} Y_{ijk}(1) - \lambda_0 D_{ijk}(1)$, $\overline{R}_{ij}(0) = \sum_{k = 1}^{N_{ij}} Y_{ijk}(0) - \lambda_0 D_{ijk}(0)$, and $\overline{R}_{ij} = \sum_{k = 1}^{N_{ij}} Y_{ijk} - \lambda_0 D_{ijk}$. Note that $\overline{R}_{ij}(1)$, $\overline{R}_{ij}(0)$, and $\overline{R}_{ij}$ all depend on $\lambda_0$, though we omit this dependence from the notation for readability. The $R$ refers to the residualized outcome. The data from the pre- and post-rollout periods can be naturally utilized for covariance adjustment. Following Section 7.1 of \citet{Aronow2017}, consider estimating functions $g_1(x)$ and $g_0(x)$ of covariates that approximate $Y(1)-\lambda_0 D(1)$ and $Y(0)-\lambda_0 D(0)$ for individuals with covariates $x$ using the pre- and post-rollout data. Viewing these fitted functions $g_1$ and $g_0$ as fixed, define $\widehat{\tau}_{1,\lambda_0,\text{HT}}^{\text{aug}} = N^{-1}\sum_{i=1}^I \sum_{j=1}^J Z_{ij}/e_{ij}\{\overline{R}_{ij} - \overline{g}_1(x_{ij})\} + \overline{g}_1(x_{ij})$, $\widehat{\tau}_{0,\lambda_0,\text{HT}}^{\text{aug}} = N^{-1}\sum_{i=1}^I \sum_{j=1}^J (1 -Z_{ij})/(1-e_{ij})\{\overline{R}_{ij} - \overline{g}_0(x_{ij})\} + \overline{g}_0(x_{ij})$, and $\widehat{\tau}_{\lambda_0,\text{HT}}^{\text{aug}} = \widehat{\tau}_{1,\lambda_0,\text{HT}}^{\text{aug}} - \widehat{\tau}_{0,\lambda_0,\text{HT}}^{\text{aug}}$. It is straightforward to show that such estimators remain unbiased under randomization. The variance estimators (introduced in Appendix \ref{appendix: horvitz-thompson}) can be adjusted by replacing $\overline{R}_{ij}$ with $\overline{R}_{ij} - \overline{g}_{Z_{ij}}(x_{ij})$. The arguments and assumptions required for the central limit theorem can likewise be adapted to account for the regression adjustment. To the best of our knowledge, Proposition \ref{prop: ht test validity} presents a novel CLT for Horvitz-Thompson estimators in stepped-wedge designs and is also readily applicable to settings without noncompliance.

\subsection{Testing the sharp null}
The results in the previous two subsections were derived in the fixed $J$, $I \to \infty$ regime, which serves as a good approximation in many practical settings \citep{nevins2024adherence}. A complementary strategy, which requires no asymptotic approximation whatsoever, is to test Fisher's sharp null of no effect of the randomized intervention $Z$ before considering inference on the effect ratio; such a test is finite-sample valid and can therefore offer additional reassurance when the number of clusters $I$ is limited, as in one-cluster-per-sequence designs. Note that under the exclusion restriction, a test of no effect of $Z$ also serves as a test of no effect of the actual treatment $D$. \cite{Ji2017RandomizationInsurance} were the first to propose such a procedure, which can be implemented with any test statistic. We explore an extension of \cite{Ji2017RandomizationInsurance} under a proportional effect model suited for the noncompliance setting in Appendix \ref{appendix: sharp null} and also test the sharp null in the \redaps trial using ANCOVA-based test statistics.

\section{Simulation studies} \label{section: sims}
To evaluate the finite-sample performance of several estimators and testing procedures, we conduct an extensive simulation study. The exact details of the data-generating process are deferred to Appendix \ref{section: detailed simulation setup}. We compare the unadjusted, ANCOVA-I, and ANCOVA-III estimators; three Horvitz–Thompson estimators (one unadjusted, one regression-adjusted using only the pre- and post-rollout periods, and one regression-adjusted using the full data); and the model-based \texttt{ivmodel} estimator used in the original analysis of \citet{Courtright2024}. For the regression-adjusted HT estimators, two separate linear models (for intervened and control) of the outcome on the cluster-period covariate $X_{ij1}$ and individual covariate $X_{ijk2}$ supply the augmentation functions $g_1$ and $g_0$; see the Appendix for details. We also consider several variance estimators: for ANCOVA, the design-based $S(\lambda)$ of \citet{Chen2025} (replacing $I_j - 1$ with $1$ when $I_j = 1$) and the CR0 and CR3 estimators from \texttt{clubSandwich}, compared to Gaussian and $t$ ($I-2$ degrees of freedom) quantiles; for HT, we utilize the conservative variance estimator $\widehat{\text{Var}}(\widehat{\tau}_{\lambda_0,\text{HT}})$ as well as a simplified estimator that is not guaranteed to be conservative, $\widehat{\text{Var}}_{\text{simp}}(\widehat{\tau}_{\lambda_0,\text{HT}})$. The exact forms of both are presented in Appendix \ref{appendix: horvitz-thompson}. For the \texttt{ivmodel} estimators, we compare both the conditional likelihood ratio (CLR) method \citep{Moreira2003}, which was used in the original \redaps analysis and is widely used, and the Fuller method \citep{Fuller1977}.

We report mean squared error (MSE), the type I error rate of testing $H_0: \lambda = \lambda_0$ at the true $\lambda_0$ (level $\alpha = 0.05$), and the power against $H_0: \lambda = 0$. Appendix \ref{section: detailed sim results} reports the complete set of figures and tables for both informative and uninformative cluster-size scenarios, including the full numerical results. Our findings (see Figures \ref{fig: summary results}-\ref{tab: iv model results}, Tables \ref{tab: ancova results}-\ref{tab: iv model results} in Appendix \ref{section: detailed sim results}) (1) The unadjusted estimators had larger MSE and substantially less power than their adjusted counterparts across all three estimator types. (2) The only methods that controlled (or nearly controlled) type I error here were the ANCOVA methods with CR3 variance estimation compared to a $t$-distribution, and the regression-adjusted (pre- and post-rollout) HT estimator with either HT variance estimator. (3) Almost all methods achieve validity as the number of clusters grows. The exceptions were the \texttt{ivmodel} CLR method and the full-data adjusted HT estimator with the simplified variance estimator, which struggled most in the one-cluster-per-sequence setting. The primary source of their anti-conservativeness is anti-conservative variance estimation when the number of clusters is limited. (4) Among methods that eventually control type I error, the highest power was attained by the adjusted ANCOVA-I and ANCOVA-III methods, the adjusted \texttt{ivmodel} Fuller method, and the full-data regression-adjusted HT estimator with the conservative variance estimator. The latter two, however, suffer type I error inflation in smaller samples (\texttt{ivmodel} Fuller for $I \leq 60$, adjusted HT for $I \leq 30$). (5) Cluster-period average ANCOVA estimators outperform the scaled cluster-period total estimators. Figure \ref{fig: all ancova summary results} in Appendix \ref{appendix: aggregate sim results} shows a non-negligible MSE and power gap in favor of the average estimators, which perform comparably to (but slightly worse than) the individual-level ANCOVA estimators.

\subsection{Practical recommendations} \label{section: recs}
Based on the findings from the simulations, we outline practical recommendations for the analysis of stepped-wedge designs with individual-level noncompliance. These are illustrated in the flowchart in Appendix Figure \ref{fig: recommendation flowchart}. For all practical purposes, using an adjusted ANCOVA method with CR3 variance estimation and comparing to a $t$-distribution is our near-universal recommendation. This method worked well with both small and large numbers of clusters, as well as in the one-cluster-per-sequence setting. Perhaps its only limitation is that the variance estimation is not theoretically justified from a design-based perspective (as compared to the theoretical justification of the design-based variance estimation for the HT estimator). For the one-cluster-per-sequence setting, if a theoretically justified variance estimate is desired, we recommend the regression-adjusted HT estimator, using only pre- and post-rollout periods for adjustment, and the provably conservative variance estimate. Finally, when there are many clusters and the design is not one-cluster-per-sequence, the \texttt{ivmodel} Fuller method appears reliable and has good power. We also note that the conservatism in the HT estimators appears to stem from overly conservative variance estimates. It may be interesting to see if the ideas of \citet{Harshaw2021} can be utilized to close this gap, and the gap in power between the (adjusted) HT and other estimators. We do not recommend utilizing unadjusted estimators, as all performed worse than their regression-adjusted counterparts, due to a combination of increased finite-sample bias and variance. We also recommend exclusively testing the sharp null if finite-sample validity is necessary, which is reflected in the flowchart in Figure \ref{fig: recommendation flowchart}.

\section{Case Study: Reanalysis of the \redaps trial} \label{section: data analysis}
We apply our methods to re-analyze the \redaps trial. We restrict our attention to tests based on the ANCOVA-I and ANCOVA-III estimators utilizing the CR3 variance estimator compared to a $t$-distribution (our recommended methods), and tests based on the \texttt{ivmodel} package based on the Fuller and CLR methods (the CLR method was used in the original analysis of \citet{Courtright2024}). Table \ref{tab: REDAPS table 1} reports baseline covariates by intervention arm in the rollout periods. All standardized mean differences fall below 0.2, with `days between enrollments' and the categorical covariate `source of admission' having the highest differences. These baseline covariates were chosen based on domain knowledge due to their expected correlation with the outcomes of interest. The residual imbalances in baseline covariates further motivate us to utilize the regression-adjusted ANCOVA estimators. Moreover, there were no issues with missing data in the \redaps study, as the primary endpoint for the outcomes was specifically chosen so that outcomes could be ascertained for all patients. We analyze both the primary length of stay outcome as well as a secondary readmission count outcome considered by \cite{Courtright2024}. The length of stay outcome, measured in log hours, ranged from 3.67 to 7.85. The readmission count outcome was an integer ranging from 0 to 4.
\subsection{Testing association between time on intervention  and receipt of palliative care and median length of stay}

In this subsection, as a heuristic for testing intervention duration irrelevance (Assumption \ref{assumption: treat dur irrel}), we explore the potential association between time on intervention/control and 1) receipt of palliative care and 2) length of stay, conditional on baseline covariates. 
This diagnostic is relevant because, when the treatment effect instead varies with the duration of exposure, a summary estimand that assumes duration irrelevance recovers only a weighted average of the exposure-time-specific effects, with weights that need not be uniform across exposure periods, complicating its interpretation \citep[see, e.g.,][]{Kenny2022, Chen2026time}.
The comparisons are performed in the intervention and control groups separately, since Assumption \ref{assumption: treat dur irrel} pertains to the duration in which a hospital has been exposed to intervention/control. Time on intervention/control can be encoded as either categorical or continuous (measured in the number of time periods on intervention/control), so there are $2^3 = 8$ total tests. For the receipt of palliative care, we fit the following logistic regressions: (1) $D \sim 1 + \text{time on intervention} + \text{calendar period} + \text{baseline covariates} \text{ among $Z$ = 1}$, (2) $D \sim 1 + \text{time on control} + \text{baseline covariates} \text{ among $Z = 0$}$. 
For the length of stay, we fit the following linear regressions: (1) $Y \sim 1 + \text{time on intervention} + \text{calendar period} + \text{baseline covariates} \text{ among $Z$ = 1}$, (2) $Y \sim 1 + \text{time on control} + \text{baseline covariates} \text{ among $Z = 0$}$. 
Note that for the control regressions, calendar period is omitted since it is redundant with time on control in the stepped-wedge design. For all regressions, we conduct tests for the coefficient(s) corresponding to time on intervention/control being equal to 0, and use cluster-robust `CR3' standard errors as implemented in the \texttt{clubSandwich} package when coding time on intervention/control as continuous, and `CR2' when coding them as categorical, due to numerical instability of `CR3'. The results coding time on intervention/control as continuous are presented in Appendix Table \ref{tab: test duration irrelevance}. We see that none of the tests were rejected. Meanwhile, in Table \ref{tab: test duration irrelevance categorical}, we see that all but one of the 36 $p$-values corresponding to a time on intervention/control category fall above 0.05. In sum, these analyses provide some reassurance (but not definitive proof) that the intervention/control duration is irrelevant to potential outcomes, if one is willing to believe that the patients are exchangeable/comparable across clusters, conditional on baseline covariates.

\begin{table}[ht!]
    \centering
    \caption{Baseline characteristics by treatment group during rollout periods, reported as mean (SD) and absolute standardized mean differences (SMD).}
    \label{tab: REDAPS table 1}
    \small
    \begin{tabular}{lllll}
        \toprule
        & Usual care & Default order & abs. \\ 
        Characteristic &   ($n$ = 10533) & ($n$ = 7982) & SMD \\ 
        \midrule
        Age, mean (SD)  &   77.87 (8.25) & 78.00 (8.28) & 0.016 \\ 
        Source of admission & & \\
        \quad Transfer from another hospital &   0.07  & 0.13 & 0.195 \\ 
        \quad Home  &   0.87  & 0.81 & 0.168 \\ 
        \quad Skilled nursing facility  &   0.06 & 0.06 & 0.000\\ 
        Elixhauser Comorbidity Index, mean (SD)   & 14.99 (11.35) & 15.01 (11.62) & 0.002\\ Eligible Diagnosis & & \\
        \quad Dementia  &  0.30  & 0.32 & 0.037\\ 
        \quad Kidney failure  &   0.14  & 0.11 & 0.079 \\ 
        \quad Chronic obstructive pulmonary heart disease &   0.69  & 0.70 & 0.020 \\ 
        Days between repeated enrollments, mean (SD)  &   72.77 (133.54) & 95.56 (177.23) & 0.145\\ Location at enrollment & & \\
        \quad Intensive care unit  &   0.24  & 0.27  & 0.071\\ \quad Hospital ward  &   0.76  & 0.83 & 0.071  \\ 
        \bottomrule
    \end{tabular}
\end{table}

\begin{table}[ht!]
\centering
\caption{Tests of association between time on intervention/control and $D$/$Y$ using covariate-adjusted logistic/linear regression with CR2 cluster-robust standard errors, treating time on intervention/control as a categorical variable. Each column reports $p$-values for a factor level.}
\label{tab: test duration irrelevance categorical}
\begin{tabular}{ccccccccccc}
  \hline
$D/Y$ & inter./cont. & 1 & 2 & 3 & 4 & 5 & 6 & 7 & 8 & 9 \\ 
  \hline
$D$ & inter. & 0.75 & 0.97 & 0.33 & 0.42 & 0.85 & 0.85 & 0.16 & 0.41 & 0.10 \\ 
  $Y$ & inter. & 0.70 & 0.07 & 0.94 & 0.33 & 0.56 & 0.33 & 0.87 & 0.23 & 0.47 \\ 
  $D$ & cont. & 0.61 & 0.80 & 0.45 & 0.73 & 0.57 & 0.23 & 0.64 & 0.51 & 0.04 \\ 
  $Y$ & cont. & 0.58 & 0.87 & 0.24 & 0.16 & 0.77 & 0.85 & 0.96 & 0.62 & 0.09 \\ 
   \hline
\end{tabular}
\end{table}

\subsection{Analysis results for Fisher's sharp null}
We first report results from testing Fisher's sharp null of no effect, namely $H_0^F: Y_{ijk}(1) = Y_{ijk}(0) \ \forall i,j,k$ for both outcomes. Letting $\Omega$ denote the possible stepped-wedge randomizations ($|\Omega| = 11!$ in the \redaps trial), a finite-sample valid $p$-value for an arbitrary test statistic $t$ with observed test statistic value $q$ can be computed as follows: $t(\bm{Z}, \bm{Y}, \bm{X})$ as follows:
\begin{equation*} 
    \mathbb P\{t(\bm{Z}, \bm{Y}, \bm{X}) \geq q | \mathcal{Y}, \mathcal{D}, \mathcal{X}\} \equiv \frac{|\bm{z} \in \Omega: t(\bm{z}, \bm{Y}, \bm{X}) \geq q|}{|\Omega|} \approx \frac{1 + \sum_{l=1}^M \mathbbm{1}\{t(\bm{\widetilde{z}}_l,\bm{Y}, \bm{X}) \geq q\}}{M + 1},
\end{equation*}
where we draw $M = 5000$ stepped-wedge randomizations $\bm{\widetilde{z}}_l$ at random (according to Assumption \ref{assumption: sw randomization}) as a Monte Carlo approximation. We delegate additional details of the procedure to Appendix \ref{appendix: sharp null} (also see \cite{Ji2017RandomizationInsurance}). For our analysis, we chose ANCOVA deviates as the test statistic. Specifically, we compute deviates $\widehat{\tau}_0 / S(0)$ where the numerators are computed using the unadjusted, ANCOVA-I, and ANCOVA-III estimators, and the denominators by the corresponding design-based standard deviation estimates introduced in \cite{Chen2025}, replacing instances of $I_j-1$ with 1. The results are reported in Table \ref{tab: redaps results sharp null}. We reiterate that the $p$-values in Table \ref{tab: redaps results sharp null} are finite-sample valid, up to Monte Carlo error. Notably, both adjusted ANCOVA-based deviate statistics reject the null of no effect for the length of stay outcome using a one-sided test at level $\alpha = 0.1$, though no tests reject the null of no effect for the secondary readmission outcome.

\begin{table}[htbp]
    \centering
        \caption{Testing the sharp null in \redaps using ANCOVA deviates with design-based variance estimators for primary ranked length of stay and secondary 30-day readmission outcomes.}
    \label{tab: redaps results sharp null}
    \small
    \begin{tabular}{lcc|cc}
        \toprule
        & \multicolumn{2}{c}{Ranked length of stay} & \multicolumn{2}{c}{Readmission within 30 days count} \\
        \cmidrule(lr){2-3} \cmidrule(lr){4-5}
        & 1-sided $p$-value & 2-sided $p$-value & 1-sided $p$-value & 2-sided $p$-value  \\
        \midrule
        unadj. ANCOVA dev. & 0.124 & 0.228 &  0.107 & 0.159 \\ 
        ANCOVA-I dev. & 0.096 & 0.173 & 0.327 &  0.570  \\ ANCOVA-III dev. & 0.090 & 0.170 & 0.318 & 0.564\\
        \bottomrule
    \end{tabular}
\end{table}

\subsection{Analysis results for the effect ratio}
We next conduct inference on the effect ratio $\lambda$. Table \ref{tab: redaps results combined} reports point estimates and confidence intervals for the effect ratio, which under Assumption \ref{assumption: iv}, equals the average effect of a palliative care consultation on the length of stay outcome among the compliers. A shorter length of stay is desirable, so a negative effect indicates that palliative care is beneficial. The adjusted CLR \texttt{ivmodel} method produced a statistically significant result (point estimate $-0.1062$, 95\% CI [-0.186,-0.0271]). Given the inflated type I error of the CLR method in simulation, however, this is likely an unreliable finding. The methods that performed well in simulation for the one-cluster-per-sequence setting, namely the ANCOVA methods, produced more extreme point estimates of roughly $-0.18$ and confidence intervals that contained 0, though the 90\% confidence intervals just barely covered 0. In Table \ref{tab: redaps results combined}, we also report results for the secondary outcome, readmission count within 30 days. The original analysis in \citet{Courtright2024} performed only ITT analyses for this secondary outcome and found a statistically significant negative effect of the randomized default order intervention on readmission count. Fewer re-admissions are desirable, so a negative effect indicates that palliative care is beneficial. Our results also suggest a negative effect of palliative care on readmission count among compliers, though the confidence intervals include 0, except for the unadjusted \texttt{ivmodel} CLR method, which, again, may not be robust in the one-cluster-per-sequence setting. In addition to the individual-level ANCOVA estimators, we perform analyses using the cluster-period average and scaled cluster-period total ANCOVA estimators (inferences are for $\tau_Y$ for the latter). Similar to the pattern observed in simulation, the cluster-period average ANCOVA estimators were less efficient (produced wider intervals) than the individual ANCOVA estimators, and the scaled cluster-period total ANCOVA estimators were the least efficient. The inefficiency of the scaled cluster-period total estimators was extremely pronounced for the primary length of stay outcome, which was due to extremely large (estimated) variances.

\begin{table}[ht!]
    \centering
    \caption{\redaps trial results: primary ranked length of stay and secondary readmission count within 30 days outcomes. For the (individual) ANCOVA, \texttt{ivmodel}, and average cluster-period ANCOVA, inference is performed on the effect ratio $\lambda$. For the scaled cluster-period total ANCOVA, inferences are for the intention-to-treat, $\tau_Y$.}
    \label{tab: redaps results combined}
    \small
    \begin{tabular}{lccc|ccc}
        \toprule
        & \multicolumn{3}{c}{Ranked length of stay} & \multicolumn{3}{c}{Readmission within 30 days count} \\
        \cmidrule(lr){2-4} \cmidrule(lr){5-7}
        Estimator & Point & 90\% CI & 95\% CI & Point & 90\% CI & 95\% CI \\
        \midrule
        ANCOVA-I & -0.18 & [-0.46, 0.01] & [-0.60, 0.07] & -0.02 & [-0.11, 0.03] & [-0.16, 0.04] \\
        ANCOVA-III & -0.18 & [-0.46, 0.00] & [-0.62, 0.05] & -0.02 & [-0.09, 0.03] & [-0.13, 0.04] \\
        unadj. ANCOVA & -0.15 & [-0.37, 0.04] & [-0.49, 0.10] & -0.07 & [-0.16, -0.01] & [-0.22, 0.01] \\
        \hline
        \texttt{ivmodel}, adj. CLR & -0.11 & [-0.17, -0.04] & [-0.19, -0.03] & -0.02 & [-0.06, 0.03] & [-0.07, 0.04] \\
        \texttt{ivmodel}, adj. Fuller & -0.11 & [-0.28, 0.07] & [-0.32, 0.11] & -0.02 & [-0.08, 0.05] & [-0.09, 0.06] \\
        \texttt{ivmodel}, unadj. CLR & -0.08 & [-0.14, -0.01] & [-0.16, 0.01] & -0.06 & [-0.10, -0.01] & [-0.11, 0.00] \\
        \texttt{ivmodel}, unadj. Fuller & -0.07 & [-0.23, 0.09] & [-0.26, 0.12] & -0.06 & [-0.12, 0.00] & [-0.13, 0.02] \\
        \hline
        avg ANCOVA-I & -0.21 & [-0.66, 0.03] & [-0.87, 0.09] & 0.04 & [-0.03, 0.11] & [-0.06, 0.13] \\
        avg ANCOVA-III & -0.16 & [-0.81, 0.23] & [-1.10, 0.32] & 0.05 & [-0.04, 0.17] & [-0.06, 0.22] \\
        avg unadj. ANCOVA & -0.14 & [-0.34, 0.04] & [-0.49, 0.12] & -0.06 & [-0.18, 0.00] & [-0.25, 0.01] \\
        \hline
        scaled ANCOVA-I & -1.33 & [-3.45, 0.79] & [-3.95, 1.29] & -0.05 & [-0.16, 0.06] & [-0.18, 0.08] \\
        scaled ANCOVA-III & -1.28 & [-3.41, 0.85] & [-3.90, 1.35] & -0.04 & [-0.14, 0.06] & [-0.17, 0.09] \\
        scaled unadj. ANCOVA & -1.39 & [-3.72, 0.94] & [-4.27, 1.49] & -0.08 & [-0.20, 0.04] & [-0.23, 0.07] \\
        \bottomrule
    \end{tabular}
\end{table}

\subsection{Sensitivity analysis for the sample complier average effect}
Under Assumption \ref{assumption: iv}, the effect ratio $\lambda$ can be interpreted as the sample complier average treatment effect. Moreover, by design, all components of Assumption \ref{assumption: iv} are highly likely to hold \citep{Courtright2024}. Nevertheless, we probe the sensitivity of our estimates and confidence intervals to potential violations of the monotonicity (no defiers) component of Assumption \ref{assumption: iv}. To do so, we introduce two sensitivity parameters. Let $\gamma$ denote the ratio of defiers to compliers in the finite population of individuals in the rollout periods of the trial. Next, let $\text{SATE}_D$ denote the sample average treatment effect among the defiers. A straightforward calculation (see Appendix \ref{appendix: sensitivity}) shows that $\text{SATE}_C$, defined as the sample average treatment effect among compliers, takes the form $(1-\gamma)\lambda - \gamma \text{SATE}_D$. Thus, for posited values of $\gamma$ and $\text{SATE}_D$, we can obtain estimates/confidence bounds for $\text{SATE}_C$, since our inferences are always valid for $\lambda$, the effect ratio. Note that $\gamma = 0$ recovers the monotonicity (no defiers) assumption, in which case $\text{SATE}_C = \lambda$. Figure \ref{fig: sensitivity} plots point estimates for $\text{SATE}_C$ over a grid of values for $\gamma$ and $\text{SATE}_D$. Since the confidence intervals for both outcomes include zero, we focus on sensitivity of the point estimates. The original estimates come from the ANCOVA-I/III estimators. Note that for the point estimate for the primary length of stay outcome to cross zero, the treatment effect among defiers would need to be substantially more extreme than -0.5 even if the ratio of defiers to compliers is 0.2. We believe that such a large ratio of defiers to compliers is extremely unlikely in the \redaps trial. For the secondary readmission outcome, less extreme values for $\gamma$ and $\text{SATE}_D$ can pull the point estimate to 0; for example, $\gamma = 0.1$ and $\text{SATE}_D = -0.2$ suffice.

\begin{figure}[ht]
  \centering
  \begin{subfigure}{0.48\textwidth}
    \includegraphics[width=\textwidth]{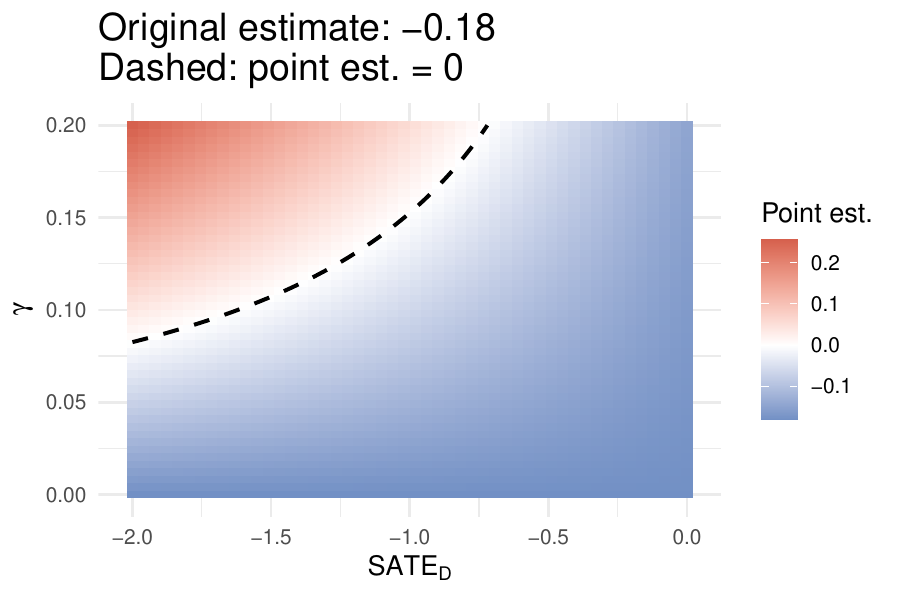}
    \caption{Sensitivity plot for the length of stay outcome.}
  \end{subfigure}
  \hfill
  \begin{subfigure}{0.48\textwidth}
    \includegraphics[width=\textwidth]{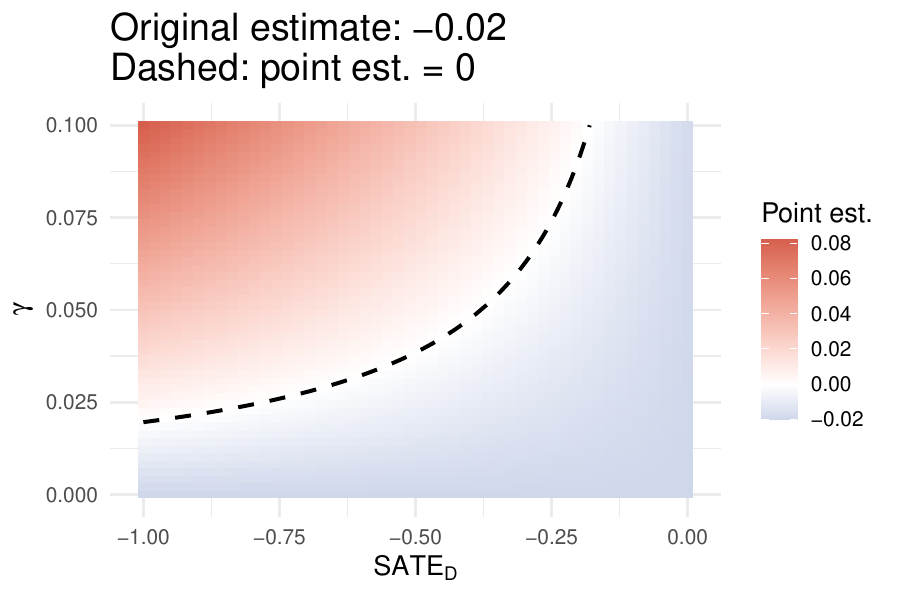}
    \caption{Sensitivity plot for the readmission outcome.}
  \end{subfigure}
  \caption{$\gamma$, the proportion of defiers to compliers, is on the $y$-axis and $\text{SATE}_D$, the sample average treatment effect of a palliative care consultation among the defiers, is on the $x$-axis. }
  \label{fig: sensitivity}
\end{figure}
\vspace{-2ex}
\section{Discussion}
Stepped-wedge designs present interesting new statistical challenges that necessitate new solutions. Noncompliance is one such challenge that has received surprisingly little attention in the stepped-wedge design literature, despite being a common occurrence in pragmatic trials, as exemplified by the \redaps palliative care trial. To the best of our knowledge, this is the first work that takes an estimand-aligned causal inference perspective to address noncompliance in stepped-wedge designs, targeting the average treatment effect among compliers under minimal assumptions. Unlike most existing approaches based on parametric models, our procedures remain valid regardless of model specification and preserve alignment between intention-to-treat and instrumental variable inferences. In particular, they avoid paradoxical outcomes such as rejecting the null hypothesis for the complier average causal effect while the intention-to-treat null is not rejected, a conflict observed in the original analysis of the \redaps trial \citep{Courtright2024}, due to a lack of appropriate methodology at the time of analysis. Moreover, given the connection between stepped-wedge designs and staggered adoption, our methods may be applicable to the staggered adoption setting in the economics literature \citep{Athey2022}. 

Our simulation studies demonstrate that ANCOVA estimators with CR3 variance estimation and $t$-reference distributions offer strong finite-sample performance, including in the challenging one-cluster-per-sequence setting. These methods consistently controlled Type I error and achieved good power, while also benefiting from covariate adjustment. In contrast, instrumental variable estimators using the \texttt{ivmodel} package, particularly the conditional likelihood ratio (CLR) method, exhibited notable Type I error inflation in finite samples. Although the Horvitz–Thompson estimators provide conservative variance estimation with strong theoretical guarantees, they tended to have less power, especially without regression adjustment.

Our reanalysis of the \redaps trial found that while the adjusted CLR-based IV method yielded a significant reduction in length of stay, our recommended ANCOVA procedures produced more negative point estimates but did not reach statistical significance at the conventional 0.05 level. This divergence echoes our simulation findings that suggest the CLR method may have inflated type I error in small-sample, one-cluster-per-sequence settings. Hence, conclusions drawn from model-based IV estimators should be interpreted cautiously in such designs.

Our work leaves several directions open for future research. One limitation is the assumption of intervention duration irrelevance, which may be violated in settings with learning, fatigue, or waning effects. It would be valuable to develop alternative estimands and estimators that explicitly account for such dynamics; see \citet{li2021mixed} and \citet{Wang2024} for an elaboration of possible dynamic treatment effect structures in stepped-wedge designs. Additionally, while the CR3 variance estimator performs well empirically, it lacks a formal design-based justification in stepped-wedge contexts. Conversely, the Horvitz–Thompson approach offers theoretical conservatism but at the cost of substantial power loss. Developing variance estimators that combine the empirical reliability of CR3 with the design-based justification of Horvitz–Thompson methods remains an important open challenge in stepped-wedge and other complicated designs. Relatedly, our methods adopt the $I \to \infty$, $J$ fixed regime, which matches the way larger stepped-wedge trials are typically envisioned (by adding clusters within a fixed sequence structure rather than by adding periods). Developing valid design-based inference when both $I$ and $J$ are fixed remains an important open problem, as non-asymptotic guarantees for sample average treatment effects are only beginning to be understood even in far simpler experimental designs \citep{sandoval2026nonasymptotic}. For such settings, the finite-sample-valid randomization tests we considered still provide a practical alternative in the interim. Though all outcomes were measured in the \redaps trial, missing or censored outcomes can occur in stepped-wedge designs \citep{ryan2025power}. For instance, in cross-sectional designs with survival endpoints, the outcome may be right-censored, a setting for which estimand-aligned methods have not been developed. One possible route is to construct pseudo-observations for the survival quantity of interest and then apply our ANCOVA or Horvitz-Thompson estimators to these pseudo-observations \citep{kjaersgaard2016instrumental}. For missing outcomes, the design-based framework of \cite{heng2025design} could in principle be adapted, but the specific extensions would merit additional research. Finally, to compute confidence intervals, our procedure relies on inverting hypothesis tests, which can be computationally burdensome. It can be of interest in future work to find a more computationally efficient procedure, potentially in the vein of \cite{ren2024model}, who construct computationally efficient confidence intervals for the effect ratio in a standard instrumental variable setting.


\section*{Acknowledgments}
F.L and M.O.H are supported by a Patient-Centered Outcomes Research Institute Award (ME-2022C2-27676). The statements presented are solely the responsibility of the authors and do not necessarily represent the views of PCORI, its Board of Governors or Methodology Committee, or the National Institutes of Health.

{
    \singlespacing
    \setlength{\bibsep}{0pt}
    \bibliographystyle{agsm}
    \if0\anon
  \bibliography{references-blinded}     
\else
  \bibliography{references}     
\fi
    
}

\setcounter{figure}{0}
\setcounter{table}{0}
\renewcommand{\thefigure}{\Alph{section}.\arabic{figure}}
\renewcommand{\thetable}{\Alph{section}.\arabic{table}}
\appendix
In the supplementary material, all expectations, probabilities, and (co)variance expressions are implicitly conditional on $\mathcal{Y}, \mathcal{D}, \mathcal{X}$. The explicit conditioning is omitted for better readability. As in the main text, we use the notation $e_{ij} \equiv \mathbb{P}(Z_{ij} = 1)$. We also introduce the notation $e^{ab}_{ij,i'j'} \equiv \mathbb{P}(Z_{ij} = a, Z_{i'j'} = b)$ for $a,b \in \{0,1\}$.

\section{Results for the ANCOVA estimators} \label{appendix: ancova}
\subsection{Summary of ANCOVA estimators}
\begin{table}[h]
    \centering
    \renewcommand{\arraystretch}{1.3} 
    \begin{tabular}{l c c c c}
        \toprule
        \textbf{Estimator} & \textbf{Mean model for}  & \textbf{Covariate effects} & \textbf{Treatment-by-} \\
                           &  $(Y_{ijk} - \lambda_0 D_{ijk})$                  &  & \textbf{covariate interactions} \\
        \midrule
        Unadjusted & $\beta_j + \theta_j Z_{ij}$ &  $\times$ & $\times$ \\
        ANCOVA-I  & $\beta_j + \theta_j Z_{ij} + X_{ijk}^c \eta$ & $\checkmark$ & $\times$  \\
        ANCOVA-III & $\beta_j + \theta_j Z_{ij} + X_{ijk}^c\gamma + Z_{ij}X_{ijk}^c \eta$ & $\checkmark$ & $\checkmark$ \\
        
        \bottomrule
    \end{tabular}
    \caption{Comparison of ANCOVA estimators with different covariate adjustments.}
    \label{tab: ancova comparison}
\end{table}

\subsection{Proof sketch of Proposition \ref{prop: ancova test validity}}

To demonstrate consistency and a central limit theorem for the ANCOVA-type estimators, one can largely adapt the regularity conditions of \citet{Chen2025}. Proposition \ref{prop: ancova test validity} directly follows from Theorem 1 of \citet{Chen2025}, with their outcomes $(Y_{ijk})$ replaced by the residualized counterparts $(Y_{ijk} - \lambda_0 D_{ijk})$, and by invoking the analogous versions of regularity conditions (i)-(v) of their Theorem 1. We note that the notation is quite cumbersome, and a formal proof is not particularly interesting. Thus, we refer the reader to \citet{Chen2025} for details.

\section{Results for the Horvitz-Thompson estimators} \label{appendix: horvitz-thompson}

\subsection{Variance estimation}

To estimate the variance, we introduce Horvitz-Thompson style estimators with added correction terms that ensure conservativeness. Taking the convention that $0/0 = 0$, let 
\begin{align*}
    \widehat{\text{Var}}&(\widehat{\tau}_{1,\lambda_0,\text{HT}}) = \frac{1}{N^2}\sum_{i=1}^I \sum_{j=1}^J \mathbbm{1}(Z_{ij} = 1)\frac{1-e_{ij}}{e_{ij}^2} N_{ij}^2 \overline{R}_{ij}^{2} \displaybreak[0]\\ 
    &+ \frac{1}{N^2}\sum_{i=1}^I \sum_{j=1}^J\sum_{(i',j') \neq (i,j)} \mathbbm{1}(Z_{ij} = 1)\mathbbm{1}(Z_{i'j'} = 1)\frac{e^{11}_{ij,i'j'} - e_{ij}e_{i'j'}}{e^{11}_{ij,i'j'}e_{ij}e_{i'j'}} N_{ij}N_{i'j'} \overline{R}_{ij}\overline{R}_{i'j'} \displaybreak[0]\\ 
    &+ \frac{1}{N^2}\sum_{i=1}^I \sum_{j=1}^J\sum_{(i',j'): e^{11}_{ij,i'j'} = 0} \left\{ \frac{\mathbbm{1}(Z_{ij} = 1)N_{ij}^2\overline{R}_{ij}^2}{2e_{ij}} + \frac{\mathbbm{1}(Z_{i'j'} = 1)N_{i'j'}^2\overline{R}_{i'j'}^2}{2e_{i'j'}} \right\}, \displaybreak[0]\\ 
    \widehat{\text{Var}}&(\widehat{\tau}_{0,\lambda_0,\text{HT}}) =  \frac{1}{N^2}\sum_{i=1}^I \sum_{j=1}^J \mathbbm{1}(Z_{ij} = 0)\frac{e_{ij}}{(1-e_{ij})^2} N_{ij}^2 \overline{R}_{ij}^{2} \displaybreak[0]\\ 
    &+ \frac{1}{N^2}\sum_{i=1}^I \sum_{j=1}^J\sum_{(i',j') \neq (i,j)} \mathbbm{1}(Z_{ij} = 0)\mathbbm{1}(Z_{i'j'} = 0)\frac{e^{00}_{ij,i'j'} - (1-e_{ij})(1-e_{i'j'})}{e^{00}_{ij,i'j'}(1-e_{ij})(1-e_{i'j'})} N_{ij}N_{i'j'} \overline{R}_{ij}\overline{R}_{i'j'} \displaybreak[0] \\  
    &+ \frac{1}{N^2}\sum_{i=1}^I \sum_{j=1}^J\sum_{(i',j'): e^{00}_{ij,i'j'} = 0} \left\{ \frac{\mathbbm{1}(Z_{ij} = 0)N_{ij}^2\overline{R}_{ij}^2}{2(1-e_{ij})}+\frac{\mathbbm{1}(Z_{i'j'} = 0)N_{i'j'}^2\overline{R}_{i'j'}^2}{2(1-e_{i'j'})} \right\}, \displaybreak[0]\\
    \widehat{\text{Cov}}&(\widehat{\tau}_{1,\lambda_0,\text{HT}},\widehat{\tau}_{0,\lambda_0,\text{HT}}) = -\frac{1}{N^2}\sum_{i=1}^I \sum_{j=1}^J \sum_{(i',j'): e^{10}_{ij,i'j'} = 0} \left\{\frac{\mathbbm{1}(Z_{ij}=1)N_{ij}^2 \overline{R}_{ij}^2}{2e_{ij}} + \frac{\mathbbm{1}(Z_{i'j'}=0)N_{i'j'}^2 \overline{R}_{i'j'}^2}{2(1-e_{i'j'})}\right\} \displaybreak[0] \\ 
    &+ \frac{1}{N^2}\sum_{i=1}^I \sum_{j=1}^J\sum_{(i',j') \neq (i,j)}\mathbbm{1}(Z_{ij} = 1)\mathbbm{1}(Z_{i'j'} = 0)\frac{e^{10}_{ij,i'j'} - e_{ij}(1-e_{i'j'})}{e^{10}_{ij,i'j'}e_{ij}(1-e_{i'j'})} N_{ij}N_{i'j'} \overline{R}_{ij}\overline{R}_{i'j'}.
\end{align*}
It is relatively straightforward to show that the variance estimators have expectation larger than the true variance, and the covariance estimator has an expectation smaller than the true variance. Now, let 
\begin{equation} \label{eqn: conservative variance estimator}
    \widehat{\text{Var}}(\widehat{\tau}_{\lambda_0,\text{HT}}) = \widehat{\text{Var}}(\widehat{\tau}_{1,\lambda_0,\text{HT}}) + \widehat{\text{Var}}(\widehat{\tau}_{0,\lambda_0,\text{HT}}) - 2\widehat{\text{Cov}}(\widehat{\tau}_{1,\lambda_0,\text{HT}},\widehat{\tau}_{0,\lambda_0,\text{HT}}).
\end{equation}
Having defined the variance estimator, we now demonstrate its conservativeness.

\begin{proof}[Proof of Lemma \ref{lemma: conservative variance}]
    The argument follows in a manner similar to that of Proposition 5.6 in \citet{Aronow2017}. Thus, the formal proof is omitted. The general intuition is that some components of the estimators are unbiased for some components of the true variance, whereas other components are unbiased for an upper bound of the other (non-identified) components of the variance (by the inequality $ab \leq \frac{1}{2}(a^2 + b^2)$).
\end{proof}

In the main manuscript, we also allude to potentially anti-conservative variance estimators. These estimators omit the components that unbiasedly estimate the upper bounds on the non-identified components of the variance. They are explicitly written below:
\begin{align*}
    \widehat{\text{Var}}_{\text{simp}}&(\widehat{\tau}_{1,\lambda_0,\text{HT}}) = \frac{1}{N^2}\sum_{i=1}^I \sum_{j=1}^J \mathbbm{1}(Z_{ij} = 1)\frac{1-e_{ij}}{e_{ij}^2} N_{ij}^2 \overline{R}_{ij}^{2} \displaybreak[0]\\ 
    &+ \frac{1}{N^2}\sum_{i=1}^I \sum_{j=1}^J\sum_{(i',j') \neq (i,j)} \mathbbm{1}(Z_{ij} = 1)\mathbbm{1}(Z_{i'j'} = 1)\frac{e^{11}_{ij,i'j'} - e_{ij}e_{i'j'}}{e^{11}_{ij,i'j'}e_{ij}e_{i'j'}} N_{ij}N_{i'j'} \overline{R}_{ij}\overline{R}_{i'j'} , \displaybreak[0]\\ 
    \widehat{\text{Var}}_{\text{simp}}&(\widehat{\tau}_{0,\lambda_0,\text{HT}}) =  \frac{1}{N^2}\sum_{i=1}^I \sum_{j=1}^J \mathbbm{1}(Z_{ij} = 0)\frac{e_{ij}}{(1-e_{ij})^2} N_{ij}^2 \overline{R}_{ij}^{2} \displaybreak[0]\\ 
    &+ \frac{1}{N^2}\sum_{i=1}^I \sum_{j=1}^J\sum_{(i',j') \neq (i,j)} \mathbbm{1}(Z_{ij} = 0)\mathbbm{1}(Z_{i'j'} = 0)\frac{e^{00}_{ij,i'j'} - (1-e_{ij})(1-e_{i'j'})}{e^{00}_{ij,i'j'}(1-e_{ij})(1-e_{i'j'})} N_{ij}N_{i'j'} \overline{R}_{ij}\overline{R}_{i'j'}, \displaybreak[0]\\
    \widehat{\text{Cov}}_{\text{simp}}&(\widehat{\tau}_{1,\lambda_0,\text{HT}},\widehat{\tau}_{0,\lambda_0,\text{HT}}) \displaybreak[0]\\
    &= \frac{1}{N^2}\sum_{i=1}^I \sum_{j=1}^J\sum_{(i',j') \neq (i,j)}\mathbbm{1}(Z_{ij} = 1)\mathbbm{1}(Z_{i'j'} = 0)\frac{e^{10}_{ij,i'j'} - e_{ij}(1-e_{i'j'})}{e^{10}_{ij,i'j'}e_{ij}(1-e_{i'j'})} N_{ij}N_{i'j'} \overline{R}_{ij}\overline{R}_{i'j'}.
\end{align*}
Analogous to the conservative variance estimator, the simplified one takes the following form: 
\begin{equation} \label{eqn: simplified variance estimator}
    \widehat{\text{Var}}_{\text{simp}}(\widehat{\tau}_{\lambda_0,\text{HT}}) = \widehat{\text{Var}}_{\text{simp}}(\widehat{\tau}_{1,\lambda_0,\text{HT}}) + \widehat{\text{Var}}_{\text{simp}}(\widehat{\tau}_{0,\lambda_0,\text{HT}}) - 2\widehat{\text{Cov}}_{\text{simp}}(\widehat{\tau}_{1,\lambda_0,\text{HT}},\widehat{\tau}_{0,\lambda_0,\text{HT}}).
\end{equation}

\subsection{Consistency and central limit theorem}

We wish to outline conditions under which the estimator $\widehat{\tau}_{\lambda_0, \text{HT}} = \widehat{\tau}_{1,\lambda_0, \text{HT}}- \widehat{\tau}_{0,\lambda_0, \text{HT}}$ is consistent and satisfies a central limit theorem. To do so, first recall that the stepped wedge treatment assignment simply places $I$ clusters into $J$ period bins, determining when treatment is first initiated for the cluster. We can equivalently view the treatment assignment in the following manner: Let $(V_{1},\ldots, V_I)$ be a random vector which is a random draw from the set of permutations of $(1,\ldots,I)$, with each permutation having equal probability $(1/I!)$ of being chosen. Then the clusters $\{i: V_i \leq I_1\}$ are assigned to be given treatment at time 1, clusters $\{i: I_1 < V_i \leq I_2\}$ are assigned to treatment at time 2, and so on. Next, consider an $I \times I$ grid denoted by $c(\cdot,\cdot)$. Let $i = 1,\ldots,I$ index the clusters and $l = 1,\ldots,I$ index the position in line of cluster $i$. We wish to write the Horvitz-Thompson estimator in the form $\sum_{i=1}^I c(i,V_i)$. Essentially, we simply need to construct define $c(i,l)$ such that it corresponds to the contribution of cluster $i$ to the test statistic when cluster $i$ is placed $l$th in line; equivalently, when $V_i = l$. For each $l = 1,\ldots,I$, define $P_l$ to be the period of treatment initiation when the cluster is placed $l$th in line. In the case where exactly one cluster is treated at each time period, $P_l = l$. Now, define $c(i,l)$ as follows:
\begin{equation}
    c(i,l) = \frac{1}{N}\sum_{j = P_l}^J \frac{N_{ij}\overline{R}_{ij}(1)}{e_{ij}} - \frac{1}{N}\sum_{j = 1}^{P_l - 1} \frac{N_{ij}\overline{R}_{ij}(0)}{1-e_{ij}},
\end{equation}
with the convention that when $P_l = 1$, the second term is 0. We next compute row means, denoted by $\overline{c}_{i\cdot}$, and column means, denoted by $\overline{c}_{\cdot l}$. Specifically,
\begin{align*}
    \overline{c}_{i\cdot} &= \frac{1}{N}\frac{1}{I}\sum_{l = 1}^I\sum_{j = P_l}^J \frac{N_{ij}\overline{R}_{ij}(1)}{e_{ij}} - \frac{1}{N}\frac{1}{I}\sum_{l = 1}^I\sum_{j = 1}^{P_l - 1} \frac{N_{ij}\overline{R}_{ij}(0)}{1-e_{ij}} \displaybreak[0]\\ 
    &= \frac{1}{N}\frac{1}{I}\sum_{l = 1}^I\sum_{j = P_l}^J \frac{N_{ij}\overline{R}_{ij}(1)}{I_j/I} - \frac{1}{N}\frac{1}{I}\sum_{l = 1}^I\sum_{j = 1}^{P_l - 1} \frac{N_{ij}\overline{R}_{ij}(0)}{1-I_j/I} \displaybreak[0]\\ 
    &= \frac{1}{N}\sum_{l = 1}^I\sum_{j = P_l}^J \frac{N_{ij}\overline{R}_{ij}(1)}{I_j} - \frac{1}{N}\sum_{l = 1}^I\sum_{j = 1}^{P_l - 1} \frac{N_{ij}\overline{R}_{ij}(0)}{I-I_j} \displaybreak[0]\\ 
    &= \frac{1}{N}\sum_{j = 1}^J N_{ij}\overline{R}_{ij}(1) - \frac{1}{N}\sum_{j = 1}^{J} N_{ij}\overline{R}_{ij}(0),
\end{align*}
where the last equality is due to there being $I_j$ and $I-I_j$ occurrences of $N_{ij}\overline{R}_{ij}(1)$ and $N_{ij}\overline{R}_{ij}(0)$ respectively when summing over $l=1,\ldots,I$. So $\overline{c}_{i\cdot}$ is the overall (residualized) treatment effect in cluster $i$. The column means are as follows:
\begin{align*}
    \overline{c}_{\cdot l} &= \frac{1}{N}\frac{1}{I}\sum_{i = 1}^I\sum_{j = P_l}^J \frac{N_{ij}\overline{R}_{ij}(1)}{e_{ij}} - \frac{1}{N}\frac{1}{I}\sum_{l = 1}^I\sum_{j = 1}^{P_l - 1} \frac{N_{ij}\overline{R}_{ij}(0)}{1-e_{ij}} \displaybreak[0]\\ 
    &=  \frac{1}{N}\sum_{j = P_l}^J \frac{I^{-1}\sum_{i = 1}^I N_{ij}\overline{R}_{ij}(1)}{e_{ij}} - \frac{1}{N}\sum_{j = 1}^{P_l - 1} \frac{I^{-1}\sum_{l = 1}^I N_{ij}\overline{R}_{ij}(0)}{1-e_{ij}},
\end{align*}
which is a weighted difference treatment and control potential outcomes with a cutoff at $P_l$, averaged over clusters. Finally, it is easy to check that the overall mean of the $c$ matrix is exactly $\tau_{\lambda_0}$. We next introduce a regularity condition to ensure the Horvitz-Thompson estimator satisfies a Central Limit Theorem (CLT).

\begin{assumption} \label{assumption: ht clt regularity}
    \begin{equation*}
        \lim_{I \to \infty}\frac{\max_{1 \leq i,l \leq I} \{c(i,l) - \overline{c}_{i\cdot}-\overline{c}_{\cdot l}+ \tau_{\lambda_0}\}^2}{I^{-1}\sum_{i=1}^I\sum_{l=1}^I \{c(i,l) - \overline{c}_{i\cdot}-\overline{c}_{\cdot l}+ \tau_{\lambda_0}\}^2} \to 0.
    \end{equation*}
\end{assumption}

\begin{proposition} \label{prop: ht clt}
    The distribution of $\widehat{\tau}_{\lambda_0,\text{HT}}$ is asymptotically normal under Assumption \ref{assumption: ht clt regularity}.
\end{proposition}

\begin{proof}[Proof of Proposition \ref{prop: ht clt}]
Given the construction of $\overline{c}_{i\cdot}$ and $\overline{c}_{\cdot l}$ and the fact that the Horvitz-Thompson estimator can be exactly written as $\sum_{i=1}^I c(i,V_i)$, the statement follows directly from Theorem 3 of \citet{hoeffding_combinatorial_clt}.
\end{proof}

Assumption \ref{assumption: ht clt regularity} essentially requires that the contrast in (potential) outcomes before and after some time period $l$ in a certain cluster $i$ is not too large, relative to other time period and cluster combinations. For further analysis, we will consider an asymptotic regime where the number of periods is fixed, and the propensity scores $e_{ij}$ are as well, and some additional boundedness conditions.

\begin{assumption}[Bounded cluster sizes, outcomes, and time periods] \label{assumption: bounded_regularity}
    $0 < c_1< N_{ij} < c_2$ for all $i,j$. $|\overline{R}_{ij}(1)|,|\overline{R}_{ij}(0)| < c_3$. $J$ is fixed, as well as the $e_{ij}$, so that $I_j  = O(I) \to \infty$.
\end{assumption}

Assumption \ref{assumption: bounded_regularity} is a sufficient, not necessary, condition. Because our framework is design-based and finite-population, the cluster-period sizes $N_{ij} $ are fixed constants for any realized trial, so the boundedness condition constrains only the hypothetical sequence of trials underlying the $I \to \infty$ asymptotics, rather than the cluster sizes of a given study. It can potentially be relaxed to allow $N_{ij} $ to grow with $I$ at the cost of additional technical conditions. However, in its current form, Assumption \ref{assumption: bounded_regularity} is an easy-to-interpret and transparent condition that suffices for the theoretical results. To demonstrate consistency, we need to check that the variance of $\widehat{\tau}_a$ goes to zero for $a = 0,1$.

\begin{lemma} \label{lem:lemma-3}
    Under Assumption \ref{assumption: bounded_regularity}, $\mathrm{Var}(\widehat{\tau}_{a,\lambda_0,\text{HT}}) \to 0$ for $a = 0,1$.
\end{lemma}

\begin{proof}[Proof of Lemma \ref{lem:lemma-3}]
    We consider first the $\text{Var}(\widehat{\tau}_{1,\lambda_0,\text{HT}})$ term. For the first term, we have, as $I \to \infty$,
    \begin{equation*}
        \begin{aligned}
            \frac{1}{N^2}\sum_{i=1}^I \sum_{j=1}^J \frac{e_{ij}(1-e_{ij})}{e_{ij}^2} N_{ij}^2 \overline{R}_{ij}(1)^{2} &\leq \frac{1}{c_1^2I^2J^2} \sum_{i=1}^I \sum_{j=1}^J \frac{1-e_{ij}}{e_{ij}} c_2^2 c_3^{2} \to 0,
        \end{aligned}
    \end{equation*}
    since it is a summation over $I$ terms and the $e_{ij}$ do not vary with $I$. We can break the second term into 3. Using Lemma \ref{lemma: order 2 joint} repeatedly,
    \begin{equation*}
        \begin{aligned}
            \frac{1}{N^2}\sum_{i=1}^I \sum_{j=1}^J\sum_{j' \neq j} \frac{e^{11}_{ij,ij'} - e_{ij}e_{ij'}}{e_{ij}e_{ij'}} N_{ij}N_{ij'} \overline{R}_{ij}(1)\overline{R}_{ij'}(1) \leq \frac{1}{c_1^2I^2J^2}\sum_{i=1}^I \sum_{j=1}^J (J-1) C c_2^2c_3^2 \to 0.
        \end{aligned}
    \end{equation*}
    Also,
    \begin{align*}
        \frac{1}{N^2}\sum_{i=1}^I \sum_{j=1}^J\sum_{i' \neq i} &\frac{e^{11}_{ij,i'j} - e_{ij}e_{i'j}}{e_{ij}e_{i'j}} N_{ij}N_{i'j} \overline{R}_{ij}(1)\overline{R}_{i'j}(1) \displaybreak[0]\\ 
        &= \frac{1}{N^2}\sum_{i=1}^I \sum_{j=1}^J\sum_{i' \neq i} \frac{I_j(I_j-1)/\{I(I-1)\} - I_j^2/I^2}{I_j^2/I^2} N_{ij}N_{i'j} \overline{R}_{ij}(1)\overline{R}_{i'j}(1) \displaybreak[0]\\ 
        &= \frac{1}{N^2}\sum_{i=1}^I \sum_{j=1}^J\sum_{i' \neq i} \frac{I+1/I_j}{I(I-1)} N_{ij}N_{i'j} \overline{R}_{ij}(1)\overline{R}_{i'j}(1) 
        \\ 
        &\leq \frac{1}{c_1^2I^2J^2}\sum_{i=1}^I \sum_{j=1}^J\sum_{i' \neq i} \frac{I+1/I_j}{I(I-1)} c_2^2c_3^2 \leq \frac{1}{c_1^2I^2J^2}\sum_{i=1}^I \sum_{j=1}^J \frac{I+1/I_j}{I} c_2^2c_3^2 \to 0.
    \end{align*}
    Furthermore,
    \begin{align*}
        \frac{1}{N^2}\sum_{i=1}^I \sum_{j=1}^J\sum_{i' \neq i}\sum_{j' \neq j} & \frac{e^{11}_{ij,i'j'} - e_{ij}e_{i'j'}}{e_{ij}e_{i'j'}} N_{ij}N_{i'j'} \overline{R}_{ij}(1)\overline{R}_{i'j'}(1) \displaybreak[0]\\ 
        &= \frac{1}{N^2}\sum_{i=1}^I \sum_{j=1}^J\sum_{i' \neq i}\sum_{j' > j} \frac{I_j(I_{j'}-1)/\{I(I-1)\} - I_jI_{j'}/I^2}{I_jI_{j'}/I^2} N_{ij}N_{i'j'} \overline{R}_{ij}(1)\overline{R}_{i'j'}(1) \displaybreak[0]\\ 
        &+ \frac{1}{N^2}\sum_{i=1}^I \sum_{j=1}^J\sum_{i' \neq i}\sum_{j' < j} \frac{I_{j'}(I_{j}-1)/\{I(I-1)\} - I_jI_{j'}/I^2}{I_jI_{j'}/I^2} N_{ij}N_{i'j'} \overline{R}_{ij}(1)\overline{R}_{i'j'}(1) \displaybreak[0]\\ 
        &= \frac{1}{N^2}\sum_{i=1}^I \sum_{j=1}^J\sum_{i' \neq i}\sum_{j' > j} \frac{I_{j'}-I}{I_{j'}(I-1)} N_{ij}N_{i'j'} \overline{R}_{ij}(1)\overline{R}_{i'j'}(1) \displaybreak[0]\\ 
        &+ \frac{1}{N^2}\sum_{i=1}^I \sum_{j=1}^J\sum_{i' \neq i}\sum_{j' < j} \frac{I_{j}-I}{I_{j}(I-1)} N_{ij}N_{i'j'} \overline{R}_{ij}(1)\overline{R}_{i'j'}(1) \displaybreak[0]\\ 
        &\leq \frac{1}{c_1^2I^2J^2}\sum_{i=1}^I \sum_{j=1}^J\sum_{j' > j} (1-I/I_{j'}) c_2^2c_3^2 \displaybreak[0]\\ 
        &+ \frac{1}{c_1^2I^2J^2}\sum_{i=1}^I \sum_{j=1}^J\sum_{j' < j}  (1-I/I_j) c_2^2c_3^2 \to 0.
    \end{align*}
    Similar steps show the analogous result for $a = 0$, i.e. $\text{Var}(\widehat{\tau}_{1,\lambda_0,\text{HT}}) \to 0$.
\end{proof}

\subsection{Proof of Proposition \ref{prop: ht test validity}}

In this subsection, we put everything together to prove the validity of tests based on the Horvitz-Thompson estimator. To this end, we first make another mild assumption, akin to Condition 6 of \citet{Aronow2017}.

\begin{assumption} \label{assumption: ht converging variance}
    As $I \to \infty$,
    \begin{equation*}
        N\mathrm{Var}(\widehat{\tau}_{\lambda_0,\text{HT}}) \to c > 0.
    \end{equation*}
\end{assumption}
\begin{remark}
We note that Assumption \ref{assumption: ht converging variance} can be weakened to $N^\alpha\mathrm{Var}(\widehat{\tau}_{\lambda_0,\text{HT}}) \to c > 0$, for some $0 < \alpha < 1$. This is because Proposition \ref{prop:proposition-4} below continues to hold if $N$ is replaced with $N^\alpha$ for some $\alpha \in (0,1)$, as $\text{Var}\{N^\alpha\widehat{\text{Var}}(\widehat{\tau}_{\lambda_0,\text{HT}})\}$ is smaller than $\text{Var}\{N\widehat{\text{Var}}(\widehat{\tau}_{\lambda_0,\text{HT}})\}$ by definition.
\end{remark}
To ultimately obtain asymptotic validity of comparing the deviate with the estimated variance to a standard normal, we need to show that $N\widehat{\text{Var}}(\widehat{\tau}_{\lambda_0,\text{HT}})$ converges to something larger than $c$. We formalize this in the following Proposition.

\begin{proposition} \label{prop:proposition-4}
    Under the asymptotic regime of Assumption \ref{assumption: bounded_regularity}, as $I \to \infty$, $N\widehat{\mathrm{Var}}(\widehat{\tau}_{\lambda_0,\text{HT}}) \to c' \geq c > 0$.
\end{proposition}

\begin{proof}[Proof of Proposition \ref{prop:proposition-4}]
    We make a few observations. First, under Assumption \ref{assumption: bounded_regularity}, as $I \to \infty$, there are no $(i,j)$, $(i',j')$ pairs for which $e^{11}_{ij,i'j'} = 0$, nor for which $e^{00}_{ij,i'j'} = 0$. Thus, the single-arm variance estimators are exactly unbiased. Next, we know that $\mathbb E\{\widehat{\text{Var}}(\widehat{\tau}_{\lambda_0,\text{HT}})\} \geq \text{Var}(\widehat{\tau}_{\lambda_0,\text{HT}})$ since the covariance estimator is downward biased. Similar to \citet{Aronow2017}, we will show that $N\widehat{\text{Var}}(\widehat{\tau}_{\lambda_0,\text{HT}})$ concentrates around its mean by showing each of its three components concentrates around their means. It is sufficient to separately show that the variances of $N\widehat{\text{Var}}(\widehat{\tau}_{1,\lambda_0,\text{HT}})$, $N\widehat{\text{Var}}(\widehat{\tau}_{0,\lambda_0,\text{HT}})$, and $N\widehat{\text{Cov}}(\widehat{\tau}_{1,\lambda_0,\text{HT}},\widehat{\tau}_{0,\lambda_0,\text{HT}})$ all tend to zero (Chebyshev's inequality then implies the desired result). Compared to other proofs, we omit some tedious details for brevity and readability. 
    
    First consider $N\widehat{\text{Var}}(\widehat{\tau}_{1,\lambda_0,\text{HT}})$. Under fixed $J$ asymptotics, there are no instances of $(i,j),(i'j')$ such that $e_{ij,i'j'}^{11} = 0$. We get 
    \begin{align*}
        &\text{Var}\{N\widehat{\text{Var}}(\widehat{\tau}_{1,\lambda_0,\text{HT}})\} \displaybreak[0]\\
        &= \frac{1}{N^2}\sum_{(i,j)} \sum_{(i',j')} \sum_{(k,l)} \sum_{(k',l')} \text{Cov}\{\mathbbm{1}(Z_{ij} =1)\mathbbm{1}(Z_{i'j'} =1),\mathbbm{1}(Z_{kl} =1)\mathbbm{1}(Z_{k'l'} =1)\} \displaybreak[0]\\ 
        &\quad\times\frac{(e^{11}_{ij,i'j'} - e_{ij}e_{i'j'})}{e^{11}_{ij,i'j'}e_{ij}e_{i'j'}} \frac{(e^{11}_{kl,k'l'} - e_{kl}e_{k'l'})}{e^{11}_{kl,k'l'}e_{kl}e_{k'l'}} N_{ij}N_{i'j'} N_{kl}N_{k'l'} \overline{R}_{ij}(1)\overline{R}_{i'j'}(1)\overline{R}_{kl}(1)\overline{R}_{k'l'}(1) \displaybreak[0]\\ &\lesssim  \frac{1}{I^2}\sum_{(i,j)} \sum_{(i',j')} \sum_{(k,l)} \sum_{(k',l')} \text{Cov}\{\mathbbm{1}(Z_{ij} =1)\mathbbm{1}(Z_{i'j'} =1),\mathbbm{1}(Z_{kl} =1)\mathbbm{1}(Z_{k'l'} =1)\} \displaybreak[0]\\ 
        &\quad\times \frac{(e^{11}_{ij,i'j'} - e_{ij}e_{i'j'})}{e^{11}_{ij,i'j'}e_{ij}e_{i'j'}} \frac{(e^{11}_{kl,k'l'} - e_{kl}e_{k'l'})}{e^{11}_{kl,k'l'}e_{kl}e_{k'l'}} 
    \end{align*}
    There are several cases we need to consider.
    
    \begin{enumerate}
        \item $i = i' = k = k'$. There are only $I \times J^4$ such terms, and under fixed $J$ asymptotics, the summands are all of constant order. Thus, the contribution of such terms is $O(I/I^2) = O(1/I)$.
        
        \item Two distinct clusters. There are $O(I^2)$ such terms. There are several subcases to consider. First, consider the case where 3 clusters are the same, and one is different. Without the loss of generality (WLOG), take $i = i' = k$, and fix $j, j', l, l'$. Then, the summands take the form 
        \begin{align*}
            \text{Cov}\{\mathbbm{1}(Z_{ij} =1)\mathbbm{1}(Z_{ij'} =1),\mathbbm{1}(Z_{il} =1)\mathbbm{1}(Z_{k'l'} =1)\}  \times \frac{(e^{11}_{ij,ij'} - e_{ij}e_{ij'})}{e^{11}_{ij,ij'}e_{ij}e_{ij'}} \frac{(e^{11}_{il,k'l'} - e_{il}e_{k'l'})}{e^{11}_{il,k'l'}e_{il}e_{k'l'}}. 
        \end{align*}
        Using Lemma \ref{lemma: order 2 joint}, one can verify that $(e^{11}_{il,k'l'} - e_{il}e_{k'l'}) = O(1/I)$. All of the other terms are of constant order. Thus, the contribution of such terms is $O(I^2 \times 1/I \times 1/I^2) = O(1/I).$ Next, consider the case where $i = i'$ and $k = k'$, and fix $j, j', l, l'$. Then, the summands take the form 
        \begin{align*}
            \text{Cov}\{\mathbbm{1}(Z_{ij} =1)\mathbbm{1}(Z_{ij'} =1),\mathbbm{1}(Z_{kl} =1)\mathbbm{1}(Z_{kl'} =1)\}  \times \frac{(e^{11}_{ij,ij'} - e_{ij}e_{ij'})}{e^{11}_{ij,ij'}e_{ij}e_{ij'}} \frac{(e^{11}_{kl,kl'} - e_{kl}e_{k'l'})}{e^{11}_{kl,k'l'}e_{kl}e_{k'l'}}. 
        \end{align*}
        Using Lemmas \ref{lemma: order 2 joint} and \ref{lemma: order 4 treatment}, one can verify that $\text{Cov}\{\mathbbm{1}(Z_{ij} =1)\mathbbm{1}(Z_{ij'} =1),\mathbbm{1}(Z_{kl} =1)\mathbbm{1}(Z_{kl'} =1)\} = O(1/I)$. All of the other terms are of constant order. Thus, the contribution of such terms is $O(I^2 \times 1/I \times 1/I^2) = O(1/I).$ Finally, we consider the case where $i,i'$ and $k,k'$. WLOG, take $i = k$ and $i' = k'$, and fix $j, j', l, l'$. Then, the summands take the form
        \begin{align*}
            \text{Cov}\{\mathbbm{1}(Z_{ij} =1)\mathbbm{1}(Z_{i'j'} =1),\mathbbm{1}(Z_{il} =1)\mathbbm{1}(Z_{i'l'} =1)\}  \times \frac{(e^{11}_{ij,i'j'} - e_{ij}e_{i'j'})}{e^{11}_{ij,i'j'}e_{ij}e_{i'j'}} \frac{(e^{11}_{il,i'l'} - e_{il}e_{i'l'})}{e^{11}_{il,i'l'}e_{il}e_{i'l'}}. 
        \end{align*}
        From Lemma \ref{lemma: order 2 joint}, we know $(e^{11}_{ij,i'j'} - e_{ij}e_{i'j'})$ and $(e^{11}_{il,i'l'} - e_{il}e_{i'l'})$ are both $O(1/I)$. All of the other terms are of constant order. Thus, the contribution of such terms is $O(I^2 \times 1/I^2 \times 1/I^2) = O(1/I^2)$.
        
        \item Three distinct clusters. There are two types of such terms. First, terms of the type where $i = i'$ and $k$, $k'$ are distinct. There are $2 \times I(I-1)(I-2)\times J^4$ such terms. Then, the summands take the form 
        \begin{align*}
            \text{Cov}\{\mathbbm{1}(Z_{ij} =1)\mathbbm{1}(Z_{ij'} =1),\mathbbm{1}(Z_{kl} =1)\mathbbm{1}(Z_{k'l'} =1)\}  \times \frac{(e^{11}_{ij,ij'} - e_{ij}e_{ij'})}{e^{11}_{ij,ij'}e_{ij}e_{ij'}} \frac{(e^{11}_{kl,k'l'} - e_{kl}e_{k'l'})}{e^{11}_{kl,k'l'}e_{kl}e_{k'l'}}. 
        \end{align*}
        As before, one can verify $(e^{11}_{kl,k'l'} - e_{kl}e_{k'l'}) = O(1/I)$. We next examine $\text{Cov}\{\mathbbm{1}(Z_{ij} =1)\mathbbm{1}(Z_{ij'} =1),\mathbbm{1}(Z_{kl} =1)\mathbbm{1}(Z_{k'l'} =1)]= \mathbb P(Z_{ij} =1,Z_{ij'} =1,Z_{kl} = 1, Z_{k'l'} = 1) - \mathbb P(Z_{ij} =1,Z_{ij'} =1)\mathbb P(Z_{kl} = 1, Z_{k'l'} = 1) $. This requires examining the six cases from Lemma \ref{lemma: order 4 treatment}. We will only illustrate one of the six cases, as the others are similar.
        
        Consider $l \leq j \leq l' \leq j'$. Then, we have 
        \begin{align*}
            \mathbb P(Z_{ij} =1,Z_{ij'} =1,Z_{kl} = 1, Z_{k'l'} = 1) = \frac{I_l}{I}\times \frac{I_j-1}{I-1}\times\frac{I_l'-2}{I-2}.
        \end{align*}
        Next, $\mathbb P(Z_{ij} =1,Z_{ij'} =1) = I_j/I$ and
        \begin{align*}
            \mathbb P(Z_{kl} = 1, Z_{k'l'} = 1) = \frac{I_l}{I}\times\frac{I_l'-1}{I-1}.
        \end{align*}
        Thus, the covariance term is of order $O(1/I)$ and the contribution of such terms is $O\{2 \times I(I-1)(I-2)\times J^4 \times 1/I\times 1/I \times 1/I^2\} = O(1/I)$. The second type of terms occurs when $i=k$ and $i'$ and $k'$ are distinct, and there are $4 \times I(I-1)(I-2)\times J^4$ such terms. Here, the terms take the form 
        \begin{align*}
            \text{Cov}\{\mathbbm{1}(Z_{ij} =1)\mathbbm{1}(Z_{i'j'} =1),\mathbbm{1}(Z_{il} =1)\mathbbm{1}(Z_{k'l'} =1)\}  \times \frac{(e^{11}_{ij,i'j'} - e_{ij}e_{i'j'})}{e^{11}_{ij,i'j'}e_{ij}e_{i'j'}} \frac{(e^{11}_{il,k'l'} - e_{il}e_{k'l'})}{e^{11}_{il,k'l'}e_{il}e_{k'l'}}. 
        \end{align*}
        In this case, both $(e^{11}_{ij,i'j'} - e_{ij}e_{i'j'})$ and $(e^{11}_{il,k'l'} - e_{il}e_{k'l'})$ are both $O(1/I)$. So the contribution of such terms is $O\{4 \times I(I-1)(I-2)\times J^4 \times 1/I\times 1/I \times 1/I^2\} = O(1/I)$.
        
        \item All distinct. There are $I(I-1)(I-2)(I-3)\times J^4$ such terms. The summands take the form 
        \begin{align*}
            \text{Cov}\{\mathbbm{1}(Z_{ij} =1)\mathbbm{1}(Z_{i'j'} =1),\mathbbm{1}(Z_{kl} =1)\mathbbm{1}(Z_{k'l'} =1)\}  \times \frac{(e^{11}_{ij,i'j'} - e_{ij}e_{i'j'})}{e^{11}_{ij,i'j'}e_{ij}e_{i'j'}} \frac{(e^{11}_{kl,k'l'} - e_{kl}e_{k'l'})}{e^{11}_{kl,k'l'}e_{kl}e_{k'l'}}.
        \end{align*}
        Here, $(e^{11}_{ij,i'j'} - e_{ij}e_{i'j'})$ and $(e^{11}_{kl,k'l'} - e_{kl}e_{k'l'})$ are both $O(1/I)$. For the covariance, there are 24 possible orderings of $j,j',l,l'$, but WLOG, we consider just one, where $l' \leq j' \leq l \leq j$. By Lemma \ref{lemma: order 4 treatment}, 
        \begin{align*}
            \mathbb P(Z_{ij} =1,Z_{i'j'} =1,Z_{kl} =1, Z_{k'l'} =1) = \frac{I_l'}{I}\times \frac{I_j'-1}{I-1}\times \frac{I_l-2}{I-2}\times \frac{I_j-3}{I-3},
        \end{align*}
        and 
        \begin{align*}
            \mathbb P(Z_{ij} =1,Z_{i'j'} =1) = \frac{I_j'}{I}\times \frac{I_j-1}{I-1}, ~~ \mathbb P(Z_{kl} =1,Z_{k'l'} =1) = \frac{I_l'}{I}\times \frac{I_l-1}{I-1}.
        \end{align*}
        Then, subtracting the product of the latter terms from the first, we get that the covariance term is also $O(1/I)$. So the contribution of all such terms is of order $O\{I(I-1)(I-2)(I-3)\times J^4 \times 1/I\times 1/I \times 1/I \times 1/I^2\} = O(1/I) $. 
    \end{enumerate}
    We do not include the steps for $N\widehat{\text{Var}}(\widehat{\tau}_{0,\lambda_0,\text{HT}})$ as they mirror those of $N\widehat{\text{Var}}(\widehat{\tau}_{1,\lambda_0,\text{HT}})$. Finally, we consider the covariance estimator. First, note that in the stepped wedge design, the only instances where $e^{10}_{ij,i'j'} = 0$ occur when $i = i'$ and $j < j'$. We now rewrite $\text{Var}\{N\widehat{\text{Cov}}(\widehat{\tau}_{1,\lambda_0,\text{HT}}, \widehat{\tau}_{1,\lambda_0,\text{HT}})\}$ as the sum $A + B + C$, where  
    \begin{align*}
        A &= \frac{1}{N^2} \sum_{(i,j)}\sum_{i' = i, j \leq j'}\sum_{(k,l)}\sum_{k' = k, l \leq l'} \frac{N_{ij}^2 \overline{R}_{ij}(1)^2N_{kl}^2 \overline{R}_{kl}(1)^2}{4e_{ij}e_{kl}}\text{Cov}\{\mathbbm{1}(Z_{ij} = 1),\mathbbm{1}(Z_{kl} = 1)\} \displaybreak[0]\\ 
        &\quad+ \frac{N_{ij}^2 \overline{R}_{ij}(1)^2N_{k'l'}^2 \overline{R}_{k'l'}(0)^2}{4e_{ij}(1-e_{k'l'})}\text{Cov}\{\mathbbm{1}(Z_{ij} = 1),\mathbbm{1}(Z_{k'l'} = 0)\} \displaybreak[0]\\ 
        &\quad+ \frac{N_{i'j'}^2 \overline{R}_{i'j'}(0)^2N_{kl}^2 \overline{R}_{kl}(1)^2}{4(1-e_{i'j'})e_{kl}}\text{Cov}\{\mathbbm{1}(Z_{i'j'} = 0),\mathbbm{1}(Z_{kl} = 1)\} \\ 
        &\quad+ \frac{N_{i'j'}^2 \overline{R}_{i'j'}(0)^2N_{k'l'}^2 \overline{R}_{k'l'}(0)^2}{4(1-e_{i'j'})(1-e_{k'l'})}\text{Cov}\{\mathbbm{1}(Z_{i'j'} = 0),\mathbbm{1}(Z_{k'l'} = 0)\} \displaybreak[0]\\
        & \lesssim \frac{1}{I^2} \sum_{(i,j)}\sum_{i' = i, j \leq j'}\sum_{(k,l)}\sum_{k' = k, l \leq l'} \frac{1}{e_{ij}e_{kl}}\text{Cov}\{\mathbbm{1}(Z_{ij} = 1),\mathbbm{1}(Z_{kl} = 1)\} \displaybreak[0]\\ 
        &\quad+ \frac{1}{e_{ij}(1-e_{k'l'})}\text{Cov}\{\mathbbm{1}(Z_{ij} = 1),\mathbbm{1}(Z_{k'l'} = 0)\} + \frac{1}{(1-e_{i'j'})e_{kl}}\text{Cov}\{\mathbbm{1}(Z_{i'j'} = 0),\mathbbm{1}(Z_{kl} = 1)\} \displaybreak[0]\\ 
        &\quad + \frac{1}{(1-e_{i'j'})(1-e_{k'l'})}\text{Cov}\{\mathbbm{1}(Z_{i'j'} = 0),\mathbbm{1}(Z_{k'l'} = 0)\}, \displaybreak[0]\\
        B &= \frac{1}{N^2} \sum_{(i,j)} \sum_{(i',j'): e_{ij,i'j'}^{10} > 0} \sum_{(k,l)} \sum_{(k',l'): e_{kl,k'l'}^{10} > 0} \text{Cov}\{\mathbbm{1}(Z_{ij}=1)\mathbbm{1}(Z_{i'j'} = 0), \mathbbm{1}(Z_{kl}=1)\mathbbm{1}(Z_{k'l'} = 0)\} \displaybreak[0]\\ 
        &\quad\times  N_{ij}N_{i',j'}N_{kl}N_{k',l'} \frac{e^{10}_{ij,i'j'} - e_{ij}(1-e_{i'j'})}{e^{10}_{ij,i'j'}e_{ij}(1-e_{i'j'})} \frac{e^{10}_{kl,k'l'} - e_{kl}(1-e_{k'l'})}{e^{10}_{kl,k'l'}e_{kl}(1-e_{k'l'})}\overline{R}_{ij}(1)\overline{R}_{i'j'}(0)\overline{R}_{kl}(1)\overline{R}_{k'l'}(0) \displaybreak[0]\\ 
        &\lesssim \frac{1}{I^2} \sum_{(i,j)} \sum_{(i',j'): e_{ij,i'j'}^{10} > 0} \sum_{(k,l)} \sum_{(k',l'): e_{kl,k'l'}^{10} > 0} \text{Cov}\{\mathbbm{1}(Z_{ij}=1)\mathbbm{1}(Z_{i'j'} = 0), \mathbbm{1}(Z_{kl}=1)\mathbbm{1}(Z_{k'l'} = 0)\} \displaybreak[0]\\ 
        &\quad\times\frac{e^{10}_{ij,i'j'} - e_{ij}(1-e_{i'j'})}{e^{10}_{ij,i'j'}e_{ij}(1-e_{i'j'})} \frac{e^{10}_{kl,k'l'} - e_{kl}(1-e_{k'l'})}{e^{10}_{kl,k'l'}e_{kl}(1-e_{k'l'})}, \displaybreak[0]\\ 
        C &= -\frac{2}{N^2} \sum_{(i,j)} \sum_{(i',j'): e_{ij,i'j'}^{10} > 0} \sum_{(k,l)} \sum_{(k',l'): e_{kl,k'l'}^{10} = 0} \text{Cov}\{\mathbbm{1}(Z_{ij} = 1)\mathbbm{1}(Z_{i'j'} = 0), \mathbbm{1}(Z_{kl} = 1)\} \displaybreak[0]\\ 
        &\qquad \times \frac{e^{10}_{ij,i'j'} - e_{ij}(1-e_{i'j'})}{e^{10}_{ij,i'j'}e_{ij}(1-e_{i'j'})}\frac{N_{kl}^2\overline{R}_{kl}(1)^2}{2e_{kl}} N_{ij}N_{i'j'} \overline{R}_{ij}(1)\overline{R}_{i'j'}(0) \displaybreak[0]\\
        &\quad+ \text{Cov}\{\mathbbm{1}(Z_{ij} = 1)\mathbbm{1}(Z_{i'j'} = 0), \mathbbm{1}(Z_{k'l'} = 0)\} \displaybreak[0]\\ 
        &\qquad \times \frac{e^{10}_{ij,i'j'} - e_{ij}(1-e_{i'j'})}{e^{10}_{ij,i'j'}e_{ij}(1-e_{i'j'})}\frac{N_{k'l'}^2\overline{R}_{k'l'}(0)^2}{2(1-e_{k'l'})} N_{ij}N_{i'j'} \overline{R}_{ij}(1)\overline{R}_{i'j'}(0) \displaybreak[0]\\ 
        &\lesssim -\frac{2}{I^2} \sum_{(i,j)} \sum_{(i',j'): e_{ij,i'j'}^{10} > 0} \sum_{(k,l)} \sum_{(k',l'): e_{kl,k'l'}^{10} = 0} \text{Cov}\{\mathbbm{1}(Z_{ij} = 1)\mathbbm{1}(Z_{i'j'} = 0), \mathbbm{1}(Z_{kl} = 1)\} \displaybreak[0]\\ 
        &\qquad \times \frac{e^{10}_{ij,i'j'} - e_{ij}(1-e_{i'j'})}{e^{10}_{ij,i'j'}e_{ij}(1-e_{i'j'})2e_{kl}} \displaybreak[0]\\
        &\quad+ \text{Cov}\{\mathbbm{1}(Z_{ij} = 1)\mathbbm{1}(Z_{i'j'} = 0), \mathbbm{1}(Z_{k'l'} = 0)\} \times \frac{e^{10}_{ij,i'j'} - e_{ij}(1-e_{i'j'})}{e^{10}_{ij,i'j'}e_{ij}(1-e_{i'j'})2(1-e_{k'l'})}
    \end{align*}
    We now demonstrate that each of the terms $A, B, C$ tends to 0. 
    
    For the $A$ term, there are $O(I^2)$ terms as we sum over $i'=i$ and $k'=k$. The first case is when $i=i'=k=k'$. There are only $O(I)$ such terms, and the summands are of constant order, so after scaling by $1/I^2$, the contribution from such terms tends to 0. When $i \neq k$, of which there are $O(I^2)$ such terms, the covariance terms can be shown to be $O(1/I)$ using Lemma \ref{lemma: order 2 joint}. Thus, the contribution from such terms is $O(1/I^2I^21/I) = O(1/I)$.
    
    For the $B$ term, there are many cases to consider. Technically, we do not enumerate all of them, but all other cases can be addressed in the same manner as one of the following cases.
    \begin{enumerate}
        \item $i=i'=k=k'$. There are only $O(I)$ such terms, and all of the summands are of constant orders, so the contribution from such terms is only $O(I/I^2) \to 0$.
        
        \item $i = i' = k$, $k'$ distinct. There are $O(I^2)$ such terms. The summands take the form 
        \begin{align*}
            \text{Cov}\{\mathbbm{1}(Z_{ij}=1)\mathbbm{1}(Z_{ij'} = 0), \mathbbm{1}(Z_{il}=1)\mathbbm{1}(Z_{k'l'} = 0)\} \times \frac{e^{10}_{ij,ij'} - e_{ij}(1-e_{ij'})}{e^{10}_{ij,ij'}e_{ij}(1-e_{ij'})} \frac{e^{10}_{il,k'l'} - e_{il}(1-e_{k'l'})}{e^{10}_{il,k'l'}e_{il}(1-e_{k'l'})}.
        \end{align*}
        By Lemma \ref{lemma: order 2 joint}, the terms $e^{10}_{il,k'l'} - e_{il}(1-e_{k'l'})$ are $O(1/I)$ and the remaining parts of the summand are of constant order, so the total contribution from such terms is $O(1/I^2 \times I^2 \times 1/I) = O(1/I)$.
        
        \item $i = i'$, $k = k'$. There are $O(I^2)$ such terms. The summands take the form 
        \begin{align*}
            \text{Cov}\{\mathbbm{1}(Z_{ij}=1)\mathbbm{1}(Z_{ij'} = 0), \mathbbm{1}(Z_{kl}=1)\mathbbm{1}(Z_{kl'} = 0)\} \times \frac{e^{10}_{ij,ij'} - e_{ij}(1-e_{ij'})}{e^{10}_{ij,ij'}e_{ij}(1-e_{ij'})} \frac{e^{10}_{kl,kl'} - e_{kl}(1-e_{kl'})}{e^{10}_{kl,kl'}e_{kl}(1-e_{kl'})}.
        \end{align*}
        By Lemma \ref{lemma: order 4 mixed}, the terms $\text{Cov}\{\mathbbm{1}(Z_{ij}=1)\mathbbm{1}(Z_{ij'} = 0), \mathbbm{1}(Z_{kl}=1)\mathbbm{1}(Z_{kl'} = 0)\}$ are $O(1/I)$ and the remaining parts of the summand are of constant order, so the total contribution from such terms is $O(1/I^2 \times I^2 \times 1/I) = O(1/I)$.
    
        \item $i =k$, $i' = k'$. There are $O(I^2)$ such terms. The summands take the form 
        \begin{align*}
            \text{Cov}\{\mathbbm{1}(Z_{ij}=1)\mathbbm{1}(Z_{i'j'} = 0), \mathbbm{1}(Z_{il}=1)\mathbbm{1}(Z_{i'l'} = 0)\} \times \frac{e^{10}_{ij,i'j'} - e_{ij}(1-e_{i'j'})}{e^{10}_{ij,i'j'}e_{ij}(1-e_{i'j'})} \frac{e^{10}_{il,i'l'} - e_{il}(1-e_{i'l'})}{e^{10}_{il,i'l'}e_{il}(1-e_{i'l'})}.
        \end{align*}
        By Lemma \ref{lemma: order 2 joint}, the terms $e^{10}_{ij,i'j'} - e_{ij}(1-e_{i'j'})$ and $e^{10}_{il,i'l'} - e_{il}(1-e_{i'l'})$ are $O(1/I)$ and the remaining parts of the summand are of constant order, so the total contribution from such terms is $O(1/I^2 \times I^2 \times 1/I^2) = O(1/I^2)$.
    
        \item $i = i'$, $k$, $k'$ distinct. There are $O(I^3)$ such terms. The summands take the form 
        \begin{align*}
            \text{Cov}\{\mathbbm{1}(Z_{ij}=1)\mathbbm{1}(Z_{ij'} = 0), \mathbbm{1}(Z_{kl}=1)\mathbbm{1}(Z_{k'l'} = 0)\} \times \frac{e^{10}_{ij,ij'} - e_{ij}(1-e_{ij'})}{e^{10}_{ij,ij'}e_{ij}(1-e_{ij'})} \frac{e^{10}_{kl,k'l'} - e_{kl}(1-e_{k'l'})}{e^{10}_{kl,k'l'}e_{kl}(1-e_{k'l'})}.
        \end{align*}
        By Lemmas \ref{lemma: order 2 joint} and \ref{lemma: order 4 mixed}, the terms $e^{10}_{kl,k'l'} - e_{kl}(1-e_{k'l'})$ and  $\text{Cov}\{\mathbbm{1}(Z_{ij}=1)\mathbbm{1}(Z_{ij'} = 0), \mathbbm{1}(Z_{kl}=1)\mathbbm{1}(Z_{k'l'} = 0)\}$ are $O(1/I)$ and the remaining parts of the summand are of constant order, so the total contribution from such terms is $O(1/I^2 \times I^3 \times 1/I^2) = O(1/I)$.
    
        \item $i = k$, $i'$, $k'$ distinct. There are $O(I^3)$ such terms. The summands take the form 
        \begin{align*}
            \text{Cov}\{\mathbbm{1}(Z_{ij}=1)\mathbbm{1}(Z_{i'j'} = 0), \mathbbm{1}(Z_{il}=1)\mathbbm{1}(Z_{k'l'} = 0)\} \times \frac{e^{10}_{ij,i'j'} - e_{ij}(1-e_{i'j'})}{e^{10}_{ij,i'j'}e_{ij}(1-e_{i'j'})} \frac{e^{10}_{il,k'l'} - e_{il}(1-e_{k'l'})}{e^{10}_{il,k'l'}e_{il}(1-e_{k'l'})}.
        \end{align*}
        By Lemma \ref{lemma: order 2 joint}, the terms $e^{10}_{ij,i'j'} - e_{ij}(1-e_{i'j'})$ and $e^{10}_{il,k'l'} - e_{il}(1-e_{k'l'})$ are $O(1/I)$ and the remaining parts of the summand are of constant order, so the total contribution from such terms is $O(1/I^2 \times I^3 \times 1/I^2) = O(1/I)$.
    
        \item $i$, $i'$, $k$, $k'$ distinct. There are $O(I^4)$ such terms. The summands take the form 
        \begin{align*}
            \text{Cov}\{\mathbbm{1}(Z_{ij}=1)\mathbbm{1}(Z_{i'j'} = 0), \mathbbm{1}(Z_{kl}=1)\mathbbm{1}(Z_{k'l'} = 0)\} \times \frac{e^{10}_{ij,i'j'} - e_{ij}(1-e_{i'j'})}{e^{10}_{ij,i'j'}e_{ij}(1-e_{i'j'})} \frac{e^{10}_{kl,k'l'} - e_{kl}(1-e_{k'l'})}{e^{10}_{kl,k'l'}e_{kl}(1-e_{k'l'})}.
        \end{align*} 
        By Lemmas \ref{lemma: order 2 joint} and \ref{lemma: order 4 mixed}, all of the terms $\text{Cov}\{\mathbbm{1}(Z_{ij}=1)\mathbbm{1}(Z_{i'j'} = 0), \mathbbm{1}(Z_{kl}=1)\mathbbm{1}(Z_{k'l'} = 0)\}$, $e^{10}_{ij,i'j'} - e_{ij}(1-e_{i'j'})$, and $e^{10}_{kl,k'l'} - e_{kl}(1-e_{k'l'})$ are all $O(1/I)$, and the remaining parts of the summand are of constant order, so the total contribution from such terms is $O(1/I^2 \times I^4 \times 1/I^3) = O(1/I)$.
    \end{enumerate}
    For the $C$ term, recall that $e^{10}_{kl,k'l'} = 0$ only when $k = k'$ and $l < l'$. We again have several cases:
    \begin{enumerate}
        \item $i = i' = k$. There are $O(I)$ such terms. The summands of constant order, so the contribution from such terms is $O(1/I^2 \times I) = O(1/I)$.

        \item $i = i'$, $k$ distinct. There are $O(I^2)$ such terms. The two summands take the forms 
        \begin{align*}
            \text{Cov}\{\mathbbm{1}(Z_{ij} = 1)\mathbbm{1}(Z_{ij'} = 0), \mathbbm{1}(Z_{kl} = 1)\} \times \frac{e^{10}_{ij,ij'} - e_{ij}(1-e_{ij'})}{e^{10}_{ij,ij'}e_{ij}(1-e_{ij'})2e_{kl}}
        \end{align*}
        and 
        \begin{align*}
            \text{Cov}\{\mathbbm{1}(Z_{ij} = 1)\mathbbm{1}(Z_{ij'} = 0), \mathbbm{1}(Z_{kl'} = 0)\} \times \frac{e^{10}_{ij,ij'} - e_{ij}(1-e_{ij'})}{e^{10}_{ij,ij'}e_{ij}(1-e_{ij'})2(1-e_{kl'})}.
        \end{align*}
        By Lemmas \ref{lemma: order 2 joint} and \ref{lemma: order 3 mixed}, $\text{Cov}\{\mathbbm{1}(Z_{ij} = 1)\mathbbm{1}(Z_{ij'} = 0), \mathbbm{1}(Z_{kl} = 1)\}$ and $\text{Cov}\{\mathbbm{1}(Z_{ij} = 1)\mathbbm{1}(Z_{ij'} = 0), \mathbbm{1}(Z_{kl'} = 0)\}$ are $O(1/I)$, and the other parts of the summand is of constant order, so the contribution from such terms is $O(1/I^2 \times I^2 \times 1/I) = O(1/I)$.
    
        \item $i = k$, $i'$ distinct. There are $O(I^2)$ such terms. The two summands take the forms 
        \begin{align*}
            \text{Cov}\{\mathbbm{1}(Z_{ij} = 1)\mathbbm{1}(Z_{i'j'} = 0), \mathbbm{1}(Z_{il} = 1)\} \times \frac{e^{10}_{ij,i'j'} - e_{ij}(1-e_{i'j'})}{e^{10}_{ij,i'j'}e_{ij}(1-e_{i'j'})2e_{il}}
        \end{align*}
        and 
        \begin{align*}
            \text{Cov}\{\mathbbm{1}(Z_{ij} = 1)\mathbbm{1}(Z_{i'j'} = 0), \mathbbm{1}(Z_{il'} = 0)\} \times \frac{e^{10}_{ij,i'j'} - e_{ij}(1-e_{i'j'})}{e^{10}_{ij,i'j'}e_{ij}(1-e_{i'j'})2(1-e_{il'})}.
        \end{align*} 
        By Lemma \ref{lemma: order 2 joint}, $e^{10}_{ij,i'j'} - e_{ij}(1-e_{i'j'}) = O(1/I)$, and the other parts of the summands are of constant order, so the contribution from such terms is $O(1/I^2 \times I^2 \times 1/I) = O(1/I)$.
    
        \item $i$, $i'$, $k$ all distinct. There are $O(I^3)$ such terms. The two summands take the forms 
        \begin{align*}
            \text{Cov}\{\mathbbm{1}(Z_{ij} = 1)\mathbbm{1}(Z_{i'j'} = 0), \mathbbm{1}(Z_{kl} = 1)\} \times \frac{e^{10}_{ij,i'j'} - e_{ij}(1-e_{i'j'})}{e^{10}_{ij,i'j'}e_{ij}(1-e_{i'j'})2e_{kl}} 
        \end{align*} 
        and 
        \begin{align*}
            \text{Cov}\{\mathbbm{1}(Z_{ij} = 1)\mathbbm{1}(Z_{i'j'} = 0), \mathbbm{1}(Z_{kl'} = 0)\} \times \frac{e^{10}_{ij,i'j'} - e_{ij}(1-e_{i'j'})}{e^{10}_{ij,i'j'}e_{ij}(1-e_{i'j'})2(1-e_{kl'})}.
        \end{align*}
        In this case, by Lemmas \ref{lemma: order 2 joint} and \ref{lemma: order 3 mixed}, $e^{10}_{ij,i'j'} - e_{ij}(1-e_{i'j'})$, $\text{Cov}\{\mathbbm{1}(Z_{ij} = 1)\mathbbm{1}(Z_{i'j'} = 0), \mathbbm{1}(Z_{kl} = 1)\}$, and $\text{Cov}\{\mathbbm{1}(Z_{ij} = 1)\mathbbm{1}(Z_{i'j'} = 0), \mathbbm{1}(Z_{kl'} = 0)\}$ are all $O(1/I)$. The other parts of the summands are of constant order, so the contribution from such terms is $O(1/I^2 \times I^3 \times 1/I^2) = O(1/I)$.
    \end{enumerate}
    We have thus established the desired result after showing the variances of $N\widehat{\text{Var}}(\widehat{\tau}_{1,\lambda_0, \text{HT}})$, $N\widehat{\text{Var}}(\widehat{\tau}_{0,\lambda_0, \text{HT}})$, and $N\widehat{\text{Cov}}(\widehat{\tau}_{1,\lambda_0, \text{HT}},\widehat{\tau}_{0,\lambda_0, \text{HT}})$ all tend to zero.
\end{proof}

We now have all of the ingredients to prove Proposition \ref{prop: ht test validity}.

\begin{proof}[Proof of Proposition \ref{prop: ht test validity}]
    Under the null $H_0: \lambda = \lambda_0$, $\tau_{\lambda_0} = 0$. Note that 
    \begin{align*}
        (\widehat{\tau}_{\lambda_0,\text{HT}}-\tau_{\lambda_0}) / \sqrt{\widehat{\text{Var}}(\widehat{\tau}_{\lambda_0,\text{HT}})} &= \widehat{\tau}_{\lambda_0,\text{HT}} / \sqrt{\widehat{\text{Var}}(\widehat{\tau}_{\lambda_0,\text{HT}})} \displaybreak[0]\\ &= \widehat{\tau}_{\lambda_0,\text{HT}} / \sqrt{{\text{Var}}(\widehat{\tau}_{\lambda_0,\text{HT}})}\times \sqrt{{\text{Var}}(\widehat{\tau}_{\lambda_0,\text{HT}})} / \sqrt{\widehat{\text{Var}}(\widehat{\tau}_{\lambda_0,\text{HT}})} \displaybreak[0]\\ &\stackrel{d}{\to} N(0,1) \times \sqrt{c/c'}.
    \end{align*}
    Thus, since $c/c' \leq 1$, asymptotically, the tails of $(\widehat{\tau}_{\lambda_0,\text{HT}}-\tau_{\lambda_0}) / \sqrt{\widehat{\text{Var}}(\widehat{\tau}_{\lambda_0,\text{HT}})} = \widehat{\tau}_{\lambda_0,\text{HT}} / S_{HT}(\lambda_0)$ are dominated by those of the standard normal distribution, proving the result.
\end{proof}

\section{Detailed simulation setup}
\label{section: detailed simulation setup}
There are two dichotomous factors in our numerical experiment, which amounts to $2 \times 2 = 4$ distinct settings. The first factor is whether the cluster size is informative or not, and the second factor is whether $I = J - 1$ or not. We hold constant how we simulate the cluster-period sizes and the cluster-period and individual covariates. For $j \in\{0,1, \ldots, J+1\}$, we simulate the cluster-period size from $N_{i j} \sim \mathcal{U}(10,90)+2(j+1)^{1.5}$ rounded to the nearest integer, and $\mathcal{U}$ stands for the uniform distribution; $X_{i j 1} \sim \mathcal{B}(0.5)$ represents an exogenous cluster-period-level summary variable following the Bernoulli distribution with $\mathbb{P}\left(X_{i j 1}=1\right)=0.5$, $X_{i j k 2} \sim i / I+\mathcal{U}(-1,1)$ is an individual-level covariate, and $\bar{X}_{j 2}$ denotes the sample average. Let $C_{ijk} \in \{0, 1, 2\}$ denote the compliance status, where 0 corresponds to complier ($D_{ijk}(1) = 1, D_{ijk}(0) = 0$), 1 corresponds to always-taker ($D_{ijk}(1) =  D_{ijk}(0) = 1$), and 2 corresponds to never-taker ($D_{ijk}(1) = D_{ijk}(0) = 0$). We exclude defiers ($D_{ijk}(1) = 0, D_{ijk}(0) = 1$). $c_F$ and $c_I$ are constants chosen to ensure that the compliance rate was close to 30\%, which is roughly what it was estimated to be in the REDAPS trial. Finally, $c_i \sim \mathcal{N}(0,0.1)$ is a cluster effect and $e_{i j k} \sim \mathcal{N}(0,0.9)$ is individual noise. In Appendix \ref{appendix: heavy tail sims}, we consider an identical setting except with $e_{i j k}$ drawn from a $t$-distribution with 4 degrees of freedom, which has substantially heavier tails. The results are qualitatively similar to the Gaussian error setting. However, in the one-cluster-per-sequence setting, the type I error rate of the unadjusted ANCOVA and ANCOVA-I procedures with CR3 standard errors and $t$ distribution are slightly inflated, while the ANCOVA-III procedure properly maintains the nominal test size.

We enforce the exclusion restriction, and in a slight abuse of notation, we let $Y_{i j k}(0)$ and $Y_{i j k}(1)$ denote potential outcomes with respect to the actual treatment $D$, not the randomized intervention (of course, they are identical among compliers). For the case without informative cluster size, we generate potential outcomes as follows:
\begin{align*}
    & Y_{i j k}(0)=\frac{j+1}{J+2}+X_{i j 1}+\left(X_{i j k 2}-\overline{X}_{j 2}\right)^2+c_i+e_{i j k}, \\
    & Y_{i j k}(1)=Y_{i j k}(0)+0.5 X_{i j 1}+\left(X_{i j k 2}-\overline{X}_{j 2}\right)^3.
\end{align*}
For the case without informative cluster size, we generate compliance status using a multinomial logistic regression with 3 levels (excluding defiers) as follows $(c_I = -0.5, c_F = 2.5)$: 
\begin{align*}
    &\log \frac{\mathbb P(C_{ijk} = 1)}{\mathbb P(C_{ijk} = 0)} = \left\{ c_I + \frac{j+1}{J+2} + 0.7X_{i j 1}+0.5\left(X_{i j k 2}-\overline{X}_{j 2}\right)^3+c_i+e_{i j k}\right\}/ c_F, \\ & \log \frac{\mathbb P(C_{ijk} = 2)}{\mathbb P(C_{ijk} = 0)} = \left\{c_I -\frac{j+1}{J+2} - 0.4X_{i j 1}+\left(X_{i j k 2}-\overline{X}_{j 2}\right)^2-c_i+e_{i j k}\right\}/ c_F.
\end{align*}
For the case with informative cluster size, we generate potential outcomes as follows:
\begin{align*}
    & Y_{i j k}(0)=\frac{j+1}{J+2}+X_{i j 1}+\left(X_{i j k 2}-\overline{X}_{j 2}\right)^2+c_i+e_{i j k}, \\
    & Y_{i j k}(1)=Y_{i j k}(0)+\frac{2 N_{i j} I}{(J+2)^{-1} \sum_{j=0}^{J+1} N_j}+0.5 X_{i j 1}+\left(X_{i j k 2}-\overline{X}_{j 2}\right)^3.
\end{align*}
For the case with informative cluster size, we generate compliance status using a multinomial logistic regression with 3 levels (excluding defiers) as follows $(c_I = -0.5, c_F = 2.5)$: 
\begin{align*}
    &\log \frac{\mathbb P(C_{ijk} = 1)}{\mathbb P(C_{ijk} = 0)} = \left\{c_I + \frac{j+1}{J+2} + \frac{2 N_{i j} I}{(J+2)^{-1} \sum_{j=0}^{J+1} N_j}+ 0.7X_{i j 1}+0.5\left(X_{i j k 2}-\overline{X}_{j 2}\right)^3+c_i+e_{i j k}\right\} / c_F, \\ 
    & \log \frac{\mathbb P(C_{ijk} = 2)}{\mathbb P(C_{ijk} = 0)} = \left\{c_I -\frac{j+1}{J+2} - \frac{2 N_{i j} I}{(J+2)^{-1} \sum_{j=0}^{J+1} N_j} - 0.4X_{i j 1}+\left(X_{i j k 2}-\overline{X}_{j 2}\right)^2-c_i+e_{i j k}\right\} / c_F,
\end{align*}

For the case when there are more clusters than periods, we fix the number of periods at $J = 5$ and vary $I \in \{12, 30, 60, 90\}$. For the case when $I = J - 1$, we take $J = 11$ to emulate the number of clusters and periods in the REDAPS trial. In practice, we would not expect $J$ to be much larger than $11$.

We expand on the description of the three Horvitz-Thompson estimators. One uses no regression adjustment, the second uses the \emph{pre- and post-rollout periods} and fits two separate linear regression models (for intervened and control) of the outcome on the cluster-period covariate $X_{ij1}$ and the individual covariate $X_{ijk2}$, and the third uses the \emph{full data} and fits two separate linear regression models (for intervened and control) of the outcome on the cluster-period covariate $X_{ij1}$ and the individual covariate $X_{ijk2}$. For the second, the model fit on the pre-rollout period individuals becomes the $g_0$ (control) function, and the model fit on the post-rollout period individuals becomes the $g_1$ (intervened) function that are used for augmenting the Horvitz-Thompson estimator. 

\section{Detailed simulation results} \label{section: detailed sim results}

In this section, we collect  detailed simulation results in Figures \ref{fig: ancova results}-\ref{fig: ivmodel results} and Tables \ref{tab: ancova results}-\ref{tab: iv model results}. Figure \ref{fig: summary results} compares the best performing methods of each type on the same plot. Figure \ref{fig: ancova results} collects the results for the ANCOVA estimators. The CR0 and CR3 variance estimates (the latter with $t$ correction) are denoted accordingly, and the $S(\lambda)$ variance estimator is denoted by DB. Figure \ref{fig: ht results} reports the results for the Horvitz-Thompson estimators. cons refers to whether or not the type-I error using the provably conservative variance estimator is utilized, as opposed to the simplified, potentially anti-conservative estimator. Figure \ref{fig: ivmodel results} collects results for the \texttt{ivmodel} procedures, based on Fuller and conditional likelihood ratio (CLR) methods \citep{Fuller1977, Moreira2003}. Tables \ref{tab: ancova results}-\ref{tab: iv model results} include the results on informative cluster size, bias results, and for the ANCOVA estimators, results when using `CR3' variance estimation but with a Gaussian rather than $t$ reference distribution.

\begin{figure}[ht!]
    \centering
    \includegraphics[width=\linewidth]{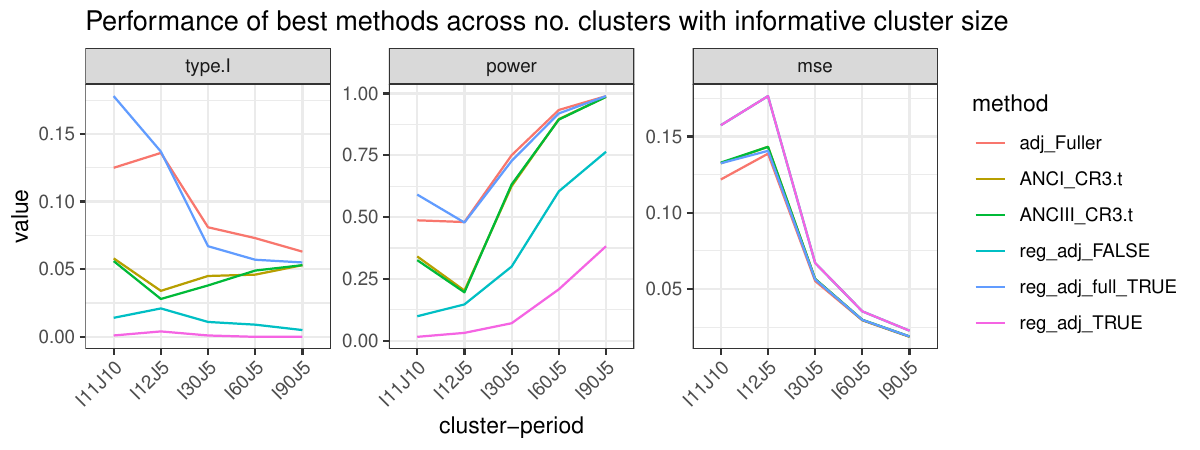}
    \caption{Simulation results for a subset of estimators under different cluster-period configurations and informative cluster size. `adj\_Fuller' corresponds to the adjusted \texttt{ivmodel} method with Fuller standard errors. `ANCI\_CR3.t' and `ANCIII\_CR3.t' are ANCOVA-I and III with CR3 standard errors compared to a $t$ distribution. `reg\_adj' and `reg\_adj\_full' are regression adjusted Horvitz-Thompson estimators using only pre and post rollout and full data, respectively, and `TRUE' refers to using the provably conservative variance estimator as opposed to the simplified, potentially anti-conservative estimator. `mse' is the Monte-Carlo mean squared error over 1000 iterations.`type.I' counts the proportion of iterations where testing at the true $\lambda_0$ resulted in a rejection. `power' counts the proportion of iterations where the false null $\lambda = 0$ was rejected.}
    \label{fig: summary results}
\end{figure}

\begin{figure}
    \centering
    \includegraphics[width=\linewidth]{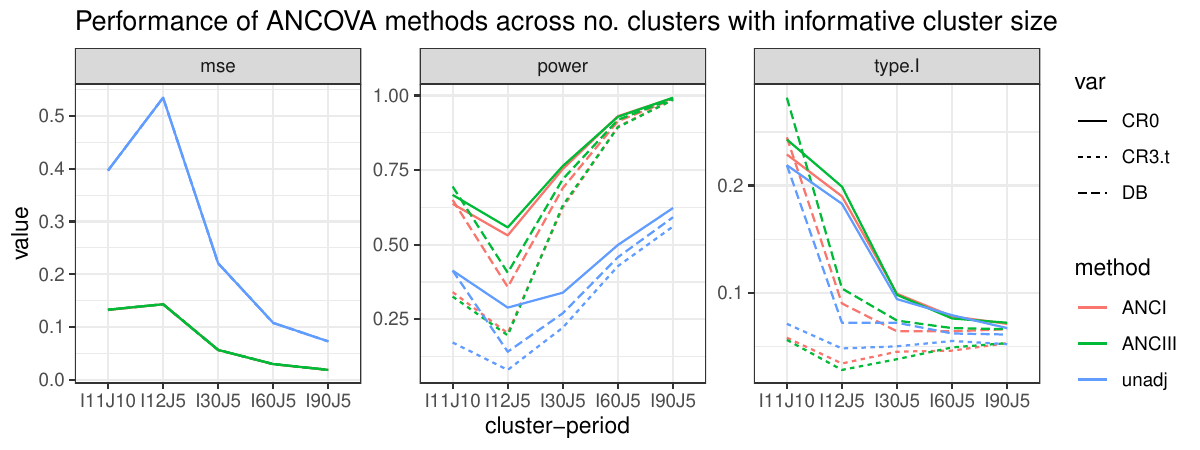}
    \caption{Simulation results for ANCOVA estimators under different cluster-period configurations and informative cluster size. `mse' is the mean squared error over 1000 iterations. `type.I' counts the proportion of iterations where testing at the true $\lambda_0$ resulted in a rejection. `var' refers to variance/reference distribution - `CR0', `CR3' comparing to a $t$-distribution, and design-based from \cite{Chen2025}. `power' counts the proportion of iterations where the false null $\lambda = 0$ was rejected.}
    \label{fig: ancova results}
\end{figure}

\begin{figure}
    \centering
    \includegraphics[width=\linewidth]{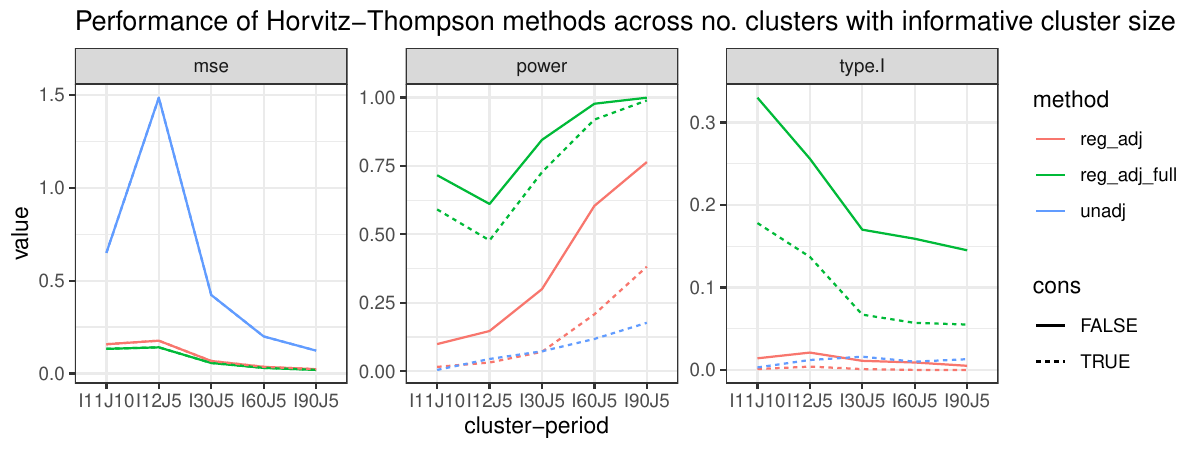}
    \caption{Simulation results for Horvitz-Thompson estimators under different cluster-period configurations and informative cluster size. The unadjusted HT estimator is omitted from the `power' and `type.I' plots as it produced ill-defined variance estimates over most of the simulation runs.  reg\_adj\_full is the adjusted HT estimator using all data for adjustment rather than only the pre and post rollout (reg\_adj). `mse' is the Monte-Carlo mean squared error over 1000 iterations. `type.I' counts the proportion of iterations where testing at the true $\lambda_0$ resulted in a rejection. `cons' denotes whether the provably conservative was used. `power' counts the proportion of iterations where the false null $\lambda = 0$ was rejected.}
    \label{fig: ht results}
\end{figure}

\begin{figure}
    \centering
    \includegraphics[width=\linewidth]{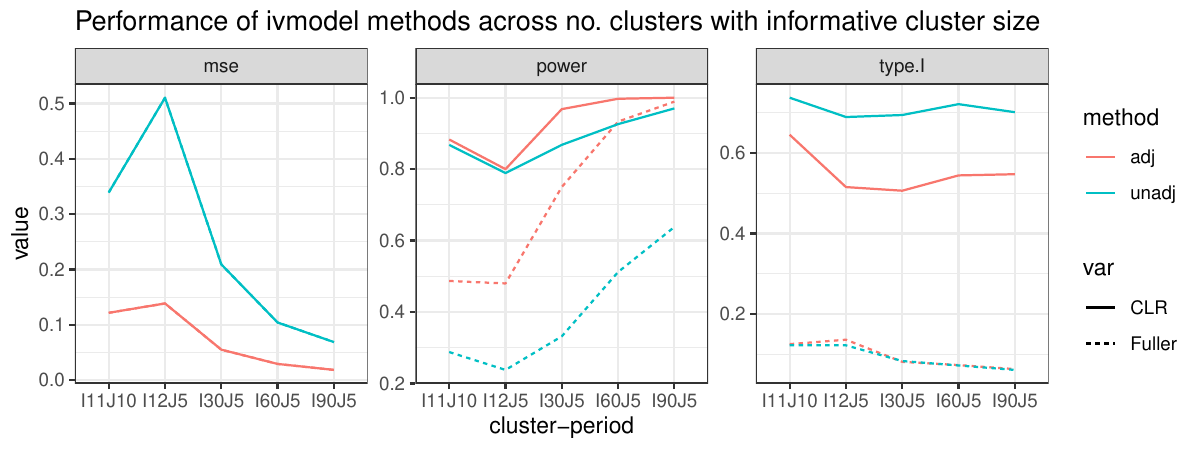}
    \caption{Simulation results for ivmodel estimators under different cluster-period configurations and informative cluster size. `mse' is the Monte-Carlo mean squared error over 1000 iterations.`type.I' counts the proportion of iterations where testing at the true $\lambda_0$ resulted in a rejection. `var' refers to Fuller or CLR variance estimators. The `power' columns count the proportion of iterations where the false null $\lambda = 0$ was rejected.}
    \label{fig: ivmodel results}
\end{figure}

\begin{sidewaystable}[ht]
    \centering
    \footnotesize
    \begin{tabular}{lllrrrrrrrrrr}
        \toprule
        IJ & method & inform & bias & mse & CR0.type.I & CR3.type.I & CR3.type.I.t & DB.type.I & CR0.power & CR3.power & CR3.power.t & DB.power \\ 
        \midrule
        I11J10 & ANCI & TRUE & -0.011 & 0.132 & 0.229 & 0.094 & 0.058 & 0.245 & 0.637 & 0.432 & 0.341 & 0.650 \\ 
        I11J10 & ANCIII & TRUE & 0.014 & 0.133 & 0.243 & 0.076 & 0.056 & 0.282 & 0.667 & 0.427 & 0.326 & 0.695 \\ 
        I11J10 & unadj & TRUE & -0.014 & 0.397 & 0.219 & 0.115 & 0.071 & 0.219 & 0.413 & 0.236 & 0.172 & 0.413 \\ 
        I12J5 & ANCI & TRUE & -0.038 & 0.143 & 0.190 & 0.062 & 0.034 & 0.090 & 0.531 & 0.281 & 0.203 & 0.358 \\ 
        I12J5 & ANCIII & TRUE & -0.020 & 0.143 & 0.199 & 0.046 & 0.028 & 0.104 & 0.558 & 0.265 & 0.196 & 0.406 \\ 
        I12J5 & unadj & TRUE & -0.048 & 0.535 & 0.183 & 0.065 & 0.048 & 0.072 & 0.289 & 0.109 & 0.080 & 0.141 \\ 
        I30J5 & ANCI & TRUE & -0.006 & 0.057 & 0.099 & 0.052 & 0.045 & 0.064 & 0.754 & 0.652 & 0.625 & 0.690 \\ 
        I30J5 & ANCIII & TRUE & 0.000 & 0.057 & 0.098 & 0.046 & 0.038 & 0.074 & 0.764 & 0.648 & 0.632 & 0.720 \\ 
        I30J5 & unadj & TRUE & 0.005 & 0.221 & 0.094 & 0.062 & 0.050 & 0.072 & 0.339 & 0.245 & 0.222 & 0.270 \\ 
        I60J5 & ANCI & TRUE & -0.002 & 0.030 & 0.078 & 0.052 & 0.046 & 0.064 & 0.931 & 0.900 & 0.896 & 0.914 \\ 
        I60J5 & ANCIII & TRUE & 0.001 & 0.030 & 0.076 & 0.055 & 0.049 & 0.067 & 0.929 & 0.903 & 0.894 & 0.919 \\ 
        I60J5 & unadj & TRUE & 0.001 & 0.108 & 0.079 & 0.057 & 0.055 & 0.062 & 0.499 & 0.441 & 0.428 & 0.458 \\ 
        I90J5 & ANCI & TRUE & -0.001 & 0.019 & 0.071 & 0.058 & 0.053 & 0.066 & 0.992 & 0.987 & 0.986 & 0.988 \\ 
        I90J5 & ANCIII & TRUE & 0.001 & 0.019 & 0.072 & 0.057 & 0.053 & 0.066 & 0.993 & 0.989 & 0.986 & 0.990 \\ 
        I90J5 & unadj & TRUE & -0.005 & 0.073 & 0.067 & 0.057 & 0.052 & 0.061 & 0.623 & 0.576 & 0.561 & 0.592 \\ 
        I11J10 & ANCI & FALSE & -0.013 & 0.124 & 0.195 & 0.081 & 0.046 & 0.200 & 0.310 & 0.139 & 0.092 & 0.334 \\ 
        I11J10 & ANCIII & FALSE & 0.003 & 0.127 & 0.224 & 0.081 & 0.054 & 0.256 & 0.367 & 0.157 & 0.109 & 0.415 \\ 
        I11J10 & unadj & FALSE & -0.037 & 0.538 & 0.214 & 0.102 & 0.071 & 0.214 & 0.237 & 0.125 & 0.094 & 0.237 \\ 
        I12J5 & ANCI & FALSE & -0.030 & 0.134 & 0.162 & 0.045 & 0.025 & 0.064 & 0.246 & 0.090 & 0.063 & 0.112 \\ 
        I12J5 & ANCIII & FALSE & -0.021 & 0.133 & 0.192 & 0.045 & 0.027 & 0.090 & 0.288 & 0.102 & 0.062 & 0.160 \\ 
        I12J5 & unadj & FALSE & -0.056 & 0.685 & 0.162 & 0.059 & 0.037 & 0.066 & 0.183 & 0.064 & 0.039 & 0.064 \\ 
        I30J5 & ANCI & FALSE & -0.005 & 0.054 & 0.076 & 0.040 & 0.030 & 0.053 & 0.264 & 0.185 & 0.158 & 0.206 \\ 
        I30J5 & ANCIII & FALSE & -0.002 & 0.053 & 0.094 & 0.046 & 0.034 & 0.065 & 0.308 & 0.195 & 0.177 & 0.241 \\ 
        I30J5 & unadj & FALSE & -0.001 & 0.279 & 0.089 & 0.058 & 0.050 & 0.066 & 0.136 & 0.078 & 0.067 & 0.094 \\ 
        I60J5 & ANCI & FALSE & -0.001 & 0.028 & 0.066 & 0.046 & 0.041 & 0.051 & 0.408 & 0.339 & 0.321 & 0.364 \\ 
        I60J5 & ANCIII & FALSE & -0.000 & 0.029 & 0.085 & 0.054 & 0.050 & 0.069 & 0.437 & 0.373 & 0.356 & 0.402 \\ 
        I60J5 & unadj & FALSE & -0.003 & 0.138 & 0.079 & 0.058 & 0.056 & 0.067 & 0.131 & 0.100 & 0.093 & 0.109 \\ 
        I90J5 & ANCI & FALSE & -0.000 & 0.018 & 0.048 & 0.036 & 0.035 & 0.041 & 0.491 & 0.439 & 0.431 & 0.461 \\ 
        I90J5 & ANCIII & FALSE & 0.001 & 0.018 & 0.068 & 0.052 & 0.046 & 0.057 & 0.521 & 0.481 & 0.471 & 0.499 \\ 
        I90J5 & unadj & FALSE & -0.007 & 0.093 & 0.070 & 0.055 & 0.053 & 0.063 & 0.166 & 0.148 & 0.139 & 0.152 \\ 
        \bottomrule
    \end{tabular}
    \caption{Simulation results for ANCOVA estimators under different cluster-period configurations.  `inform' indicates whether cluster size is informative or not. `bias' and `mse' are the Monte-Carlo bias and mean squared error over 1000 iterations. The `type.I' columns count the proportion of iterations where testing at the true $\lambda_0$ resulted in a rejection using specific variance estimators - `CR0', `CR3', `CR3' comparing to a $t$-distribution, and design-based from \citet{Chen2025}. The `power' columns count the proportion of iterations where the false null $\lambda = 0$ was rejected.}
    \label{tab: ancova results}
\end{sidewaystable}

\begin{sidewaystable}[ht]
    \centering
    \begin{tabular}{lllrrrrrrr}
        \toprule
        IJ & method & inform & mean.bias & mse & type.I.rate & cons.type.I.rate & power.rate & cons.power.rate \\ 
        \midrule
        I11J10 & reg\_adj & TRUE & -0.013 & 0.157 & 0.014 & 0.001 & 0.099 & 0.016 \\ 
        I11J10 & reg\_adj\_full & TRUE & 0.014 & 0.132 & 0.330 & 0.178 & 0.716 & 0.591 \\ 
        I11J10 & unadj & TRUE & -0.075 & 0.650 & 0.305 & 0.003 & 0.125 & 0.005 \\ 
        I12J5 & reg\_adj & TRUE & -0.030 & 0.176 & 0.021 & 0.004 & 0.147 & 0.032 \\ 
        I12J5 & reg\_adj\_full & TRUE & -0.018 & 0.141 & 0.256 & 0.137 & 0.611 & 0.478 \\ 
        I12J5 & unadj & TRUE & -0.244 & 1.487 &  & 0.012 &  & 0.045 \\ 
        I30J5 & reg\_adj & TRUE & -0.003 & 0.067 & 0.011 & 0.001 & 0.300 & 0.071 \\ 
        I30J5 & reg\_adj\_full & TRUE & 0.001 & 0.056 & 0.170 & 0.067 & 0.845 & 0.727 \\ 
        I30J5 & unadj & TRUE & -0.050 & 0.423 &  & 0.016 &  & 0.073 \\ 
        I60J5 & reg\_adj & TRUE & 0.004 & 0.035 & 0.009 & 0.000 & 0.604 & 0.208 \\ 
        I60J5 & reg\_adj\_full & TRUE & 0.002 & 0.030 & 0.159 & 0.057 & 0.977 & 0.919 \\ 
        I60J5 & unadj & TRUE & -0.029 & 0.199 &  & 0.010 &  & 0.118 \\ 
        I90J5 & reg\_adj & TRUE & -0.001 & 0.023 & 0.005 & 0.000 & 0.764 & 0.382 \\ 
        I90J5 & reg\_adj\_full & TRUE & 0.001 & 0.019 & 0.145 & 0.055 & 0.999 & 0.989 \\ 
        I90J5 & unadj & TRUE & -0.019 & 0.123 &  & 0.013 &  & 0.177 \\ 
        I11J10 & reg\_adj & FALSE & -0.021 & 0.147 & 0.010 & 0.000 & 0.031 & 0.002 \\ 
        I11J10 & reg\_adj\_full & FALSE & 0.002 & 0.127 & 0.308 & 0.168 & 0.440 & 0.296 \\ 
        I11J10 & unadj & FALSE & -0.078 & 0.730 & 0.292 & 0.002 & 0.234 & 0.004 \\ 
        I12J5 & reg\_adj & FALSE & -0.033 & 0.163 & 0.024 & 0.001 & 0.062 & 0.012 \\ 
        I12J5 & reg\_adj\_full & FALSE & -0.021 & 0.132 & 0.238 & 0.125 & 0.347 & 0.209 \\ 
        I12J5 & unadj & FALSE & -0.170 & 1.203 & 0.500 & 0.013 & 1.000 & 0.024 \\ 
        I30J5 & reg\_adj & FALSE & -0.006 & 0.062 & 0.013 & 0.003 & 0.070 & 0.012 \\ 
        I30J5 & reg\_adj\_full & FALSE & -0.001 & 0.053 & 0.158 & 0.059 & 0.409 & 0.240 \\ 
        I30J5 & unadj & FALSE & -0.033 & 0.400 &  & 0.014 &  & 0.033 \\ 
        I60J5 & reg\_adj & FALSE & 0.002 & 0.033 & 0.010 & 0.001 & 0.128 & 0.017 \\ 
        I60J5 & reg\_adj\_full & FALSE & -0.000 & 0.029 & 0.155 & 0.057 & 0.561 & 0.384 \\ 
        I60J5 & unadj & FALSE & -0.021 & 0.197 &  & 0.011 &  & 0.044 \\ 
        I90J5 & reg\_adj & FALSE & -0.000 & 0.021 & 0.006 & 0.000 & 0.182 & 0.024 \\ 
        I90J5 & reg\_adj\_full & FALSE & 0.001 & 0.018 & 0.145 & 0.049 & 0.664 & 0.468 \\ 
        I90J5 & unadj & FALSE & -0.015 & 0.125 &  & 0.014 &  & 0.051 \\ 
        \bottomrule
    \end{tabular}
    \caption{Simulation results for Horvitz-Thompson estimators under different cluster-period configurations. Blanks indicate settings where the unadjusted HT estimator produced ill-defined variance estimates over most of the simulation runs. `inform' indicates whether cluster size is informative or not. reg\_adj\_full is the adjusted HT estimator using all data for adjustment rather than only the pre- and post-rollout (reg\_adj). `bias' and `mse' are the Monte-Carlo bias and mean squared error over 1000 iterations. The `type.I' columns count the proportion of iterations where testing at the true $\lambda_0$ resulted in a rejection using the provable conservative and not necessarily conservative variance estimators. The `power' columns count the proportion of iterations where the false null $\lambda = 0$ was rejected.}
    \label{tab: ht results}
\end{sidewaystable}

\begin{sidewaystable}[ht]
    \centering
    \begin{tabular}{lllrrrrrr}
        \toprule
        cluster-period & method & inform & Fuller.mean.bias & Fuller.mse & CLR.type.I & Fuller.type.I & CLR.power & Fuller.power \\ 
        \midrule
        I11J10 & adj\_ivm & TRUE & -0.022 & 0.122 & 0.645 & 0.125 & 0.883 & 0.487 \\ 
        I11J10 & unadj\_ivm & TRUE & -0.019 & 0.339 & 0.737 & 0.122 & 0.868 & 0.288 \\ 
        I12J5 & adj\_ivm & TRUE & -0.037 & 0.139 & 0.515 & 0.136 & 0.800 & 0.480 \\ 
        I12J5 & unadj\_ivm & TRUE & -0.047 & 0.511 & 0.689 & 0.122 & 0.789 & 0.238 \\ 
        I30J5 & adj\_ivm & TRUE & -0.006 & 0.055 & 0.506 & 0.081 & 0.968 & 0.750 \\ 
        I30J5 & unadj\_ivm & TRUE & 0.001 & 0.209 & 0.694 & 0.083 & 0.868 & 0.332 \\ 
        I60J5 & adj\_ivm & TRUE & -0.003 & 0.029 & 0.544 & 0.073 & 0.997 & 0.933 \\ 
        I60J5 & unadj\_ivm & TRUE & -0.002 & 0.104 & 0.721 & 0.072 & 0.926 & 0.512 \\ 
        I90J5 & adj\_ivm & TRUE & -0.001 & 0.019 & 0.547 & 0.063 & 1.000 & 0.989 \\ 
        I90J5 & unadj\_ivm & TRUE & -0.007 & 0.069 & 0.701 & 0.061 & 0.970 & 0.638 \\ 
        I11J10 & adj\_ivm & FALSE & -0.015 & 0.114 & 0.636 & 0.114 & 0.730 & 0.211 \\ 
        I11J10 & unadj\_ivm & FALSE & -0.029 & 0.466 & 0.745 & 0.109 & 0.779 & 0.149 \\ 
        I12J5 & adj\_ivm & FALSE & -0.032 & 0.130 & 0.511 & 0.105 & 0.589 & 0.197 \\ 
        I12J5 & unadj\_ivm & FALSE & -0.057 & 0.661 & 0.693 & 0.107 & 0.754 & 0.133 \\ 
        I30J5 & adj\_ivm & FALSE & -0.004 & 0.052 & 0.482 & 0.069 & 0.718 & 0.258 \\ 
        I30J5 & unadj\_ivm & FALSE & -0.003 & 0.267 & 0.715 & 0.076 & 0.740 & 0.130 \\ 
        I60J5 & adj\_ivm & FALSE & -0.000 & 0.028 & 0.511 & 0.064 & 0.835 & 0.408 \\ 
        I60J5 & unadj\_ivm & FALSE & -0.004 & 0.135 & 0.723 & 0.069 & 0.815 & 0.134 \\ 
        I90J5 & adj\_ivm & FALSE & 0.002 & 0.018 & 0.513 & 0.051 & 0.901 & 0.509 \\ 
        I90J5 & unadj\_ivm & FALSE & -0.009 & 0.089 & 0.707 & 0.070 & 0.800 & 0.167 \\ 
        \bottomrule
    \end{tabular}
    \caption{Simulation results for ivmodel estimators under different cluster-period configurations. `inform' indicates whether cluster size is informative or not. `bias' and `mse' are the Monte-Carlo bias and mean squared error over 1000 iterations. The `type.I' columns count the proportion of iterations where testing at the true $\lambda_0$ resulted in a rejection using Fuller or CLR variance estimators. The `power' columns count the proportion of iterations where the false null $\lambda = 0$ was rejected.}
    \label{tab: iv model results}
\end{sidewaystable}

\section{Auxiliary lemmas}

\allowdisplaybreaks[4]

Most of the lemmas consist of calculations of certain joint probabilities of treatment assignments under the stepped-wedge design. The calculations are not particularly interesting and are quite tedious to write out, so they are omitted for brevity.
\begin{lemma}[Order 2 Joint Probabilities] \label{lemma: order 2 joint}
    Under the stepped-wedge design, for $j < j'$,
    \begin{equation}
        \begin{aligned}
            & \mathbb P(Z_{ij} = 1, Z_{ij'} = 1) = \frac{I_j}{I}, \  \mathbb P(Z_{ij} = 0, Z_{ij'} = 1) = \frac{I_{j'}-I_j}{I}, \\  
            & \mathbb P(Z_{ij} = 0, Z_{ij'} = 0) = 1-\frac{I_{j'}}{I}, \ \mathbb P(Z_{ij} = 1, Z_{ij'} = 0) = 0, \\  
            & \mathbb P(Z_{ij} = 1, Z_{i'j} = 1) = \frac{I_j}{I} \times \frac{I_j - 1}{I - 1},\\ 
            & \mathbb P(Z_{ij} = 1, Z_{i'j} = 0) = \mathbb P(Z_{ij} = 0, Z_{i'j} = 1) = \left(1-\frac{I_j}{I}\right) \times \frac{I_j }{I - 1}, \\  
            & \mathbb P(Z_{ij} = 0, Z_{i'j} = 0) = \frac{I - I_j}{I} \times \frac{I- I_j -1 }{I - 1},\\  
            & \mathbb P(Z_{ij} = 1, Z_{i'j'} = 1) = \frac{I_j}{I} \times \frac{I_{j'}-1}{I-1},\\ 
            & \mathbb P(Z_{ij} = 0, Z_{i'j'} = 1) = \left(\frac{I_{j'}-I_j}{I}\right)\times \frac{I_{j'}-1}{I-1} + \left(1-\frac{I_{j'}}{I}\right)\times \frac{I_{j'}}{I-1}, \\  
            & \mathbb P(Z_{ij} = 0, Z_{i'j'} = 0) = \left(1-\frac{I_{j'}}{I}\right) \times \frac{I - I_{j}-1}{I-1}, \\  
            & \mathbb P(Z_{ij} = 1, Z_{i'j'} = 0) = \frac{I_{j}}{I} \times \frac{I-I_{j'}}{I-1}.
        \end{aligned}
    \end{equation}    
\end{lemma}

\begin{lemma}[First and second moments] \label{lemma: moments}
    \begin{equation}
        \begin{aligned}
            \mathbb E(\widehat{\tau}_{a,\lambda_0,HT}) &= \tau_{a,\lambda_0} \text{ for } a = 0,1. \\ 
            \mathrm{Var}(\widehat{\tau}_{1,\lambda_0,HT}) &=  \frac{1}{N^2}\sum_{i=1}^I \sum_{j=1}^J \frac{e_{ij}(1-e_{ij})}{e_{ij}^2} N_{ij}^2 \overline{R}_{ij}(1)^{2} \\ 
            &\quad+ \frac{1}{N^2}\sum_{i=1}^I \sum_{j=1}^J\sum_{(i',j') \neq (i,j)} \frac{e^{11}_{ij,i'j'} - e_{ij}e_{i'j'}}{e_{ij}e_{i'j'}} N_{ij}N_{i'j'} \overline{R}_{ij}(1)\overline{R}_{i'j'}(1), \\
            \mathrm{Var}(\widehat{\tau}_{0,\lambda_0,HT}) &=  \frac{1}{N^2}\sum_{i=1}^I \sum_{j=1}^J \frac{e_{ij}(1-e_{ij})}{(1-e_{ij})^2} N_{ij}^2 \overline{R}_{ij}(0)^{2} \\ 
            &\quad+ \frac{1}{N^2}\sum_{i=1}^I \sum_{j=1}^J\sum_{(i',j') \neq (i,j)}\frac{e^{00}_{ij,i'j'} - (1-e_{ij})(1-e_{i'j'})}{(1-e_{ij})(1-e_{i'j'})} N_{ij}N_{i'j'} \overline{R}_{ij}(0)\overline{R}_{i'j'}(0), \\
            \mathrm{Cov}(\widehat{\tau}_{1,\lambda_0,HT},\widehat{\tau}_{0,\lambda_0,HT}) &= -\frac{1}{N^2}\sum_{i=1}^I \sum_{j=1}^J N_{ij}^2 \overline{R}_{ij}(0)\overline{R}_{ij}(1)\\ 
            &\quad+ \frac{1}{N^2}\sum_{i=1}^I \sum_{j=1}^J\sum_{(i',j') \neq (i,j)}\frac{e^{10}_{ij,i'j'} - e_{ij}(1-e_{i'j'})}{e_{ij}(1-e_{i'j'})} N_{ij}N_{i'j'} \overline{R}_{ij}(1)\overline{R}_{i'j'}(0).
        \end{aligned}
    \end{equation}
\end{lemma}

\begin{proof}[Proof of Lemma \ref{lemma: moments}]
    The proof follows similar steps as the proof for Proposition 4.2 of \citet{Aronow2017}, and is thus omitted.
\end{proof}

\begin{lemma}[Order 4 treatment probabilities] \label{lemma: order 4 treatment}
    Let $1 \leq a \leq b \leq c \leq d \leq J$ and $i, i', k, k'$ are pairwise distinct.
    \begin{align*}
        &\text{All distinct:} \\ 
        &\mathbb P(Z_{ia} = 1, Z_{i'b} = 1,Z_{kc} = 1, Z_{k'd} = 1) = \frac{I_a}{I} \times \frac{I_b - 1}{I - 1} \times \frac{I_c - 2}{I-2} \times \frac{I_d-3}{I-3} \\
        &\text{3 distinct:} \\ 
        &\mathbb P(Z_{ia} = 1, Z_{ib} = 1,Z_{kc} = 1, Z_{k'd} = 1) = \frac{I_a}{I} \times \frac{I_c - 1}{I-1} \times \frac{I_d-2}{I-2} \\ 
        &\mathbb P(Z_{ia} = 1, Z_{i'b} = 1,Z_{i'c} = 1, Z_{k'd} = 1) = \frac{I_a}{I} \times \frac{I_b - 1}{I-1} \times \frac{I_d-2}{I-2} \\ 
        &\mathbb P(Z_{ia} = 1, Z_{i'b} = 1,Z_{kc} = 1, Z_{kd} = 1) = \frac{I_a}{I} \times \frac{I_b - 1}{I-1} \times \frac{I_c-2}{I-2} \\ 
        &\mathbb P(Z_{ia} = 1, Z_{i'b} = 1,Z_{ic} = 1, Z_{kd} = 1) = \frac{I_a}{I} \times \frac{I_b - 1}{I-1} \times \frac{I_d-2}{I-2} \\ 
        &\mathbb P(Z_{ia} = 1, Z_{i'b} = 1,Z_{kc} = 1, Z_{i'd} = 1) = \frac{I_a}{I} \times \frac{I_b - 1}{I-1} \times \frac{I_c-2}{I-2} \\ 
        &\mathbb P(Z_{ia} = 1, Z_{i'b} = 1,Z_{kc} = 1, Z_{id} = 1) = \frac{I_a}{I} \times \frac{I_b - 1}{I-1} \times \frac{I_c-2}{I-2} \\ 
        &\text{2 distinct} \\  
        &\mathbb P(Z_{ia} = 1, Z_{ib} = 1,Z_{ic} = 1, Z_{k'd} = 1) = \frac{I_a}{I} \times \frac{I_d - 1}{I - 1}  \\ 
        &\mathbb P(Z_{ia} = 1, Z_{i'b} = 1,Z_{i'c} = 1, Z_{i'd} = 1) = \frac{I_a}{I} \times \frac{I_d - 1}{I - 1} \\  
        &\mathbb P(Z_{ia} = 1, Z_{ib} = 1,Z_{kc} = 1, Z_{id} = 1) = \frac{I_a}{I} \times \frac{I_c - 1}{I - 1}  \\ 
        &\mathbb P(Z_{ia} = 1, Z_{i'b} = 1,Z_{ic} = 1, Z_{id} = 1) = \frac{I_a}{I} \times \frac{I_b - 1}{I - 1} \\ 
        &\mathbb P(Z_{ia} = 1, Z_{ib} = 1,Z_{kc} = 1, Z_{kd} = 1) = \frac{I_a}{I} \times \frac{I_c-1}{I-1}\\ 
        &\mathbb P(Z_{ia} = 1, Z_{i'b} = 1,Z_{ic} = 1, Z_{i'd} = 1) = \frac{I_a}{I} \times \frac{I_b-1}{I-1} \\ 
        &\mathbb P(Z_{ia} = 1, Z_{i'b} = 1,Z_{i'c} = 1, Z_{id} = 1) = \frac{I_a}{I} \times \frac{I_b-1}{I-1}\\ 
        &\text{1 distinct} \\ 
        &\mathbb P(Z_{ia} = 1, Z_{ib} = 1,Z_{ic} = 1, Z_{id} = 1) = \frac{I_a}{I} 
    \end{align*}
\end{lemma}

\begin{lemma}[Order 4 control probabilities] \label{lemma: order 4 control}
    Let $1 \leq a \leq b \leq c \leq d \leq J$ and $i, i', k, k'$ are pairwise distinct.
    \begin{align*}
        &\text{All distinct:} \\ 
        &\mathbb P(Z_{ia} = 0, Z_{i'b} = 0,Z_{kc} = 0, Z_{k'd} = 0) = \frac{I-I_a-3}{I-3} \times \frac{I-I_b - 2}{I - 2} \times \frac{I - I_c - 1}{I-1} \times \frac{I-I_d}{I} \\
        &\text{3 distinct:} \\ 
        &\mathbb P(Z_{ia} = 0, Z_{ib} = 0,Z_{kc} = 0, Z_{k'd} = 0) = \frac{I-I_b-2}{I-2} \times \frac{I-I_c - 1}{I-1} \times \frac{I-I_d}{I} \\ 
        &\mathbb P(Z_{ia} = 0, Z_{i'b} = 0,Z_{i'c} = 0, Z_{k'd} = 0) = \frac{I-I_a-2}{I-2} \times \frac{I - I_c - 1}{I-1} \times \frac{I-I_d}{I} \\ 
        &\mathbb P(Z_{ia} = 0, Z_{i'b} = 0,Z_{kc} = 0, Z_{kd} = 0) = \frac{I-I_a-2}{I-2} \times \frac{I-I_b - 1}{I-1} \times \frac{I-I_d}{I} \\ 
        &\mathbb P(Z_{ia} = 0, Z_{i'b} = 0,Z_{ic} = 0, Z_{kd} = 0) = \frac{I-I_b-2}{I-2} \times \frac{I - I_c - 1}{I-1} \times \frac{I-I_d}{I} \\ 
        &\mathbb P(Z_{ia} = 0, Z_{i'b} = 0,Z_{kc} = 0, Z_{i'd} = 0) = \frac{I-I_a-2}{I-2} \times \frac{I-I_c - 1}{I-1} \times \frac{I-I_d}{I} \\ 
        &\mathbb P(Z_{ia} = 0, Z_{i'b} = 0,Z_{kc} = 0, Z_{id} = 0) = \frac{I-I_a-2}{I-2} \times \frac{I-I_c - 1}{I-1} \times \frac{I-I_d}{I} \\ 
        &\text{2 distinct} \\  
        &\mathbb P(Z_{ia} = 0, Z_{ib} = 0,Z_{ic} = 0, Z_{k'd} = 0) = \frac{I-I_c-1}{I-1} \times \frac{I - I_d}{I}  \\ 
        &\mathbb P(Z_{ia} = 0, Z_{i'b} = 0,Z_{i'c} = 0, Z_{i'd} = 0) = \frac{I-I_a-1}{I-1} \times \frac{I - I_d}{I} \\  
        &\mathbb P(Z_{ia} = 0, Z_{ib} = 0,Z_{kc} = 0, Z_{id} = 0) = \frac{I-I_c-1}{I-1} \times \frac{I - I_d}{I}  \\ 
        &\mathbb P(Z_{ia} = 0, Z_{i'b} = 0,Z_{ic} = 0, Z_{id} = 0) = \frac{I-I_b-1}{I-1} \times \frac{I - I_d}{I} \\ 
        &\mathbb P(Z_{ia} = 0, Z_{ib} = 0,Z_{kc} = 0, Z_{kd} = 0) = \frac{I-I_b-1}{I-1} \times \frac{I - I_d}{I} \\ 
        &\mathbb P(Z_{ia} = 0, Z_{i'b} = 0,Z_{ic} = 0, Z_{i'd} = 0) = \frac{I-I_c-1}{I-1} \times \frac{I - I_d}{I} \\ 
        &\mathbb P(Z_{ia} = 0, Z_{i'b} = 0,Z_{i'c} = 0, Z_{id} = 0) = \frac{I-I_c-1}{I-1} \times \frac{I - I_d}{I} \\ 
        &\text{1 distinct} \\ 
        &\mathbb P(Z_{ia} = 0, Z_{ib} = 0,Z_{ic} = 0, Z_{id} = 0) = \frac{I-I_d}{I} 
    \end{align*} 
\end{lemma}

\begin{lemma}[Order 3 treatment probabilities] \label{lemma: order 3 treatment}
    Let $1 \leq a \leq b \leq c \leq J$ and $i, i', k$ are pairwise distinct.
    \begin{align*}
        &\text{All distinct:} \\ 
        &\mathbb P(Z_{ia} = 1, Z_{i'b} = 1,Z_{kc} = 1) = \frac{I_a}{I} \times \frac{I_b - 1}{I - 1} \times \frac{I_c - 2}{I-2}  \\ 
        &\text{2 distinct} \\  
        &\mathbb P(Z_{ia} = 1, Z_{ib} = 1,Z_{kc} = 1) = \frac{I_a}{I} \times \frac{I_c - 1}{I - 1}  \\ 
        &\mathbb P(Z_{ia} = 1, Z_{i'b} = 1,Z_{ic} = 1) = \frac{I_a}{I} \times \frac{I_b - 1}{I - 1}  \\  
        &\mathbb P(Z_{ia} = 1, Z_{i'b} = 1,Z_{i'c} = 1) = \frac{I_a}{I} \times \frac{I_b - 1}{I - 1}  \\ 
        &\text{1 distinct} \\ 
        &\mathbb P(Z_{ia} = 1, Z_{ib} = 1,Z_{ic} = 1) = \frac{I_a}{I} 
    \end{align*}  
\end{lemma}

\begin{lemma}[Order 3 control probabilities] \label{lemma: order 3 control}
    Let $1 \leq a \leq b \leq c \leq J$ and $i, i', k$ are pairwise distinct.
    \begin{align*}
        &\text{All distinct:} \\ 
        &\mathbb P(Z_{ia} = 0, Z_{i'b} = 0,Z_{kc} = 0) = \frac{I-I_a-2}{I-2} \times \frac{I-I_b - 1}{I - 1} \times \frac{I - I_c }{I}\\ 
        &\text{2 distinct} \\  
        &\mathbb P(Z_{ia} = 0, Z_{ib} = 0,Z_{kc} = 0) = \frac{I-I_b-1}{I-1} \times \frac{I - I_c}{I}  \\ 
        &\mathbb P(Z_{ia} = 0, Z_{i'b} = 0,Z_{ic} = 0) = \frac{I-I_b-1}{I-1} \times \frac{I - I_c}{I} \\ 
        &\mathbb P(Z_{ia} = 0, Z_{i'b} = 0,Z_{i'c} = 0) = \frac{I-I_a-1}{I-1} \times \frac{I - I_c}{I} \\ 
        &\text{1 distinct} \\ 
        &\mathbb P(Z_{ia} = 0, Z_{ib} = 0,Z_{ic} = 0) = \frac{I-I_c}{I} 
    \end{align*}
\end{lemma}

\begin{lemma}[Order 3 mixed probabilities] \label{lemma: order 3 mixed}
    Let $1 \leq a \leq b \leq c \leq J$ and $i, i', k$ are pairwise distinct.
    \begin{align*}
        &\textbf{110) } \text{All distinct:}\\ 
        &\mathbb P(Z_{ia} = 1, Z_{i'b} = 1,Z_{kc} = 0) = \frac{I_a}{I} \times \frac{I_b - 1}{I - 1} \times \frac{I-I_c }{I-2}\\ 
        &\text{2 distinct} \\  
        &\mathbb P(Z_{ia} = 1, Z_{ib} = 1,Z_{kc} = 0) = \frac{I_a}{I} \times \frac{I - I_c}{I-1}  \\ 
        &\mathbb P(Z_{ia} = 1, Z_{i'b} = 1,Z_{ic} = 0) = 0 \\ 
        &\mathbb P(Z_{ia} = 1, Z_{i'b} = 1,Z_{i'c} = 0) = 0 \\ 
        &\text{1 distinct} \\ 
        &\mathbb P(Z_{ia} = 1, Z_{ib} = 1,Z_{ic} = 0) = 0 \\ 
        &\textbf{101) } \text{All distinct:}\\ 
        &\mathbb P(Z_{ia} = 1, Z_{i'b} = 0,Z_{kc} = 1) = \frac{I_a}{I} \times \frac{I_c - I_b}{I - 1} \times \frac{I-I_b -1}{I-2} + \frac{I_a}{I} \times \frac{I_b }{I - 1} \times \frac{I-I_b}{I-2}\\ 
        &\text{2 distinct} \\  
        &\mathbb P(Z_{ia} = 1, Z_{ib} = 0,Z_{kc} = 1) = 0  \\ 
        &\mathbb P(Z_{ia} = 1, Z_{i'b} = 0,Z_{ic} = 1) = \frac{I_a}{I} \times \frac{I-I_b}{I-1} \\ 
        &\mathbb P(Z_{ia} = 1, Z_{i'b} = 0,Z_{i'c} = 1) = \frac{I_a}{I}\times \frac{I_c-I_b}{I-1} \\ 
        &\text{1 distinct} \\ 
        &\mathbb P(Z_{ia} = 1, Z_{ib} = 0,Z_{ic} = 1) = 0 \\
        &\textbf{011) } \text{All distinct:}\\ 
        &\mathbb P(Z_{ia} = 0, Z_{i'b} = 1,Z_{kc} = 1) = \frac{I_b-I_a}{I}\left(\frac{I_b - 1}{I - 1} \times \frac{I_c-2 }{I-2}\right) + \frac{I_c-I_b}{I}\left(\frac{I_b }{I - 1} \times \frac{I_c-2 }{I-2}\right)\\ 
        &+\frac{I-I_c}{I}\left(\frac{I_b }{I - 1} \times \frac{I_c-1 }{I-2}\right)\\ 
        &\text{2 distinct} \\  
        &\mathbb P(Z_{ia} = 0, Z_{ib} = 1,Z_{kc} = 1) = \frac{I_b-I_a}{I} \times \frac{I_c - 1}{I-1}  \\
        &\mathbb P(Z_{ia} = 0, Z_{i'b} = 1,Z_{ic} = 1) = \frac{I_b-I_a}{I} \times \left(\frac{I_c-I_a - 1}{I-1}\right) + \frac{I_a}{I} \times \left(\frac{I_c-I_a}{I-1}\right) \\ 
        &\mathbb P(Z_{ia} = 0, Z_{i'b} = 1,Z_{i'c} = 1 ) = \frac{I_b-I_a}{I} \times \frac{I_b-1}{I-1} + \frac{I-I_b}{I} \times \frac{I_b}{I-1} \\ 
        &\text{1 distinct} \\ 
        &\mathbb P(Z_{ia} = 0, Z_{ib} = 1,Z_{ic} = 1) = \frac{I_b-I_a}{I} \\
        &\textbf{100) } \text{All distinct:}\\ 
        &\mathbb P(Z_{ia} = 1, Z_{i'b} = 0,Z_{kc} = 0) = \frac{I_a}{I} \times \frac{I-I_b - 1}{I - 2} \times \frac{I-I_c }{I-1}\\ 
        &\text{2 distinct} \\  
        &\mathbb P(Z_{ia} = 1, Z_{ib} = 0,Z_{kc} = 0) = 0  \\ 
        &\mathbb P(Z_{ia} = 1, Z_{i'b} = 0,Z_{ic} = 0) = 0 \\ 
        &\mathbb P(Z_{ia} = 1, Z_{i'b} = 0,Z_{i'c} = 0) = \frac{I_a}{I} \times \frac{I-I_c}{I-1} \\ 
        &\text{1 distinct} \\ 
        &\mathbb P(Z_{ia} = 1, Z_{ib} = 0,Z_{ic} = 0) = 0 \\ 
        &\textbf{010) } \text{All distinct:}\\ 
        &\mathbb P(Z_{ia} = 0, Z_{i'b} = 1,Z_{kc} = 0) = \frac{I_a}{I}\left(\frac{I-I_c}{I-1}\times\frac{I-I_a-1}{I-2}\right) + \frac{I_b-I_a}{I}\left(\frac{I-I_c}{I-1}\times\frac{I-I_a-2}{I-2}\right)\\ 
        &\text{2 distinct} \\  
        &\mathbb P(Z_{ia} = 0, Z_{ib} = 1,Z_{kc} = 0) = \frac{I_b-I_a}{I} \times \frac{I-I_c}{I-1}  \\ 
        &\mathbb P(Z_{ia} = 0, Z_{i'b} = 1,Z_{ic} = 0) = \frac{I-I_c}{I} \times \frac{I_b}{I-1} \\ 
        &\mathbb P(Z_{ia} = 0, Z_{i'b} = 1,Z_{i'c} = 0) = 0 \\ 
        &\text{1 distinct} \\ 
        &\mathbb P(Z_{ia} = 0, Z_{ib} = 1,Z_{ic} = 0) = 0 \\
        &\textbf{001) } \text{All distinct:}\\ 
        &\mathbb P(Z_{ia} = 0, Z_{i'b} = 0,Z_{kc} = 1) = \frac{I_a}{I}\left(\frac{I-I_b}{I - 1} \times \frac{I-I_a-1}{I-2}\right) + \frac{I_b-I_a}{I}\left(\frac{I-I_b }{I - 1} \times \frac{I-I_a-2 }{I-2}\right)\\ 
        &+\frac{I_c-I_b}{I}\left(\frac{I - I_b - 1}{I - 1} \times \frac{I-I_a-2}{I-2}\right)\\ 
        &\text{2 distinct} \\  
        &\mathbb P(Z_{ia} = 0, Z_{ib} = 0,Z_{kc} = 1) = \frac{I-I_c}{I} \times \frac{I_c }{I-1} + \frac{I_c-I_b}{I} \times \frac{I_c-1}{I-1} \\ 
        &\mathbb P(Z_{ia} = 0, Z_{i'b} = 0,Z_{ic} = 1) = \frac{I_b-I_a}{I} \times \left(\frac{I-I_b }{I-1}\right) + \frac{I_c-I_b}{I} \times \left(\frac{I-I_b-1 }{I-1}\right) \\ 
        &\mathbb P(Z_{ia} = 0, Z_{i'b} = 0,Z_{i'c} = 1 ) = \frac{I_c-I_b}{I} \times \frac{I-I_a-1}{I-1} \\ 
        &\text{1 distinct} \\ 
        &\mathbb P(Z_{ia} = 0, Z_{ib} = 0,Z_{ic} = 1) = \frac{I_c-I_b}{I} 
    \end{align*}
\end{lemma}

\begin{lemma}[Order 4 mixed probabilities] \label{lemma: order 4 mixed}
    Let $1 \leq a \leq b \leq c \leq d \leq J$ and $i, i', k, k'$ are pairwise distinct.
    \begin{align*}
        &\textbf{1100) } \text{All distinct:}\\ 
        &\mathbb P(Z_{ia} = 1, Z_{i'b} = 1, Z_{kc} = 0,  Z_{k'd} = 0) = \frac{I_a}{I} \times \frac{I_b-1}{I-1} \times \frac{I-I_d}{I-2} \times \frac{I-I_c - 1}{I-3} \\
        &\text{3 distinct:} \\  
        &\mathbb P(Z_{ia} = 1, Z_{i'b} = 1, Z_{kc} = 0,  Z_{kd} = 0) = \frac{I_a}{I} \times \frac{I_b-1}{I-1} \times \frac{I-I_d}{I-2}\\ 
        &\mathbb P(Z_{ia} = 1, Z_{i'b} = 1, Z_{i'c} = 0,  Z_{kd} = 0) = 0\\ 
        &\mathbb P(Z_{ia} = 1, Z_{ib} = 1, Z_{i'c} = 0,  Z_{kd} = 0) = \frac{I_a}{I} \times \frac{I-I_d}{I-1} \times \frac{I-I_c - 1}{I-2}\\ 
        &\mathbb P(Z_{ia} = 1, Z_{i'b} = 1, Z_{ic} = 0,  Z_{kd} = 0) = 0\\ 
        &\mathbb P(Z_{ia} = 1, Z_{i'b} = 1, Z_{kc} = 0,  Z_{id} = 0) = 0\\ 
        &\mathbb P(Z_{ia} = 1, Z_{i'b} = 1, Z_{kc} = 0,  Z_{i'd} = 0) = 0\\  
        &\text{2 distinct:} \\ 
        &\mathbb P(Z_{ia} = 1, Z_{ib} = 1, Z_{ic} = 0,  Z_{i'd} = 0) = 0 \\ 
        &\mathbb P(Z_{ia} = 1, Z_{ib} = 1, Z_{i'c} = 0,  Z_{id} = 0) = 0 \\ 
        &\mathbb P(Z_{ia} = 1, Z_{i'b} = 1, Z_{ic} = 0,  Z_{id} = 0) = 0 \\ 
        &\mathbb P(Z_{ia} = 1, Z_{i'b} = 1, Z_{i'c} = 0,  Z_{i'd} = 0) = 0 \\ 
        &\mathbb P(Z_{ia} = 1, Z_{ib} = 1, Z_{i'c} = 0,  Z_{i'd} = 0) = \frac{I_a}{I} \times \frac{I-I_d}{I-1} \\ 
        &\mathbb P(Z_{ia} = 1, Z_{i'b} = 1, Z_{i'c} = 0,  Z_{id} = 0) = 0 \\ 
        &\mathbb P(Z_{ia} = 1, Z_{i'b} = 1, Z_{ic} = 0,  Z_{i'd} = 0) = 0\\  
        &\text{1 distinct:} \\ 
        &\mathbb P(Z_{ia} = 1, Z_{ib} = 1, Z_{ic} = 0,  Z_{id} = 0) = 0. \\ 
        &\textbf{1010) } \text{All distinct:}\\  
        &\mathbb P(Z_{ia} = 1, Z_{i'b} = 0, Z_{kc} = 1,  Z_{k'd} = 0) = \frac{I_a}{I} \times \frac{I - I_d}{I-1}\\ 
        &\times\left(\frac{I-I_d-1}{I-2} \times \frac{I_c - 1}{I-3} + \frac{I_d-I_c}{I-2} \times \frac{I_c - 1}{I-3} + \frac{I_c-I_b}{I-2} \times \frac{I_c - 2}{I-3}\right)\\ 
        &\text{3 distinct:} \\  
        &\mathbb P(Z_{ia} = 1, Z_{i'b} = 0, Z_{kc} = 1,  Z_{kd} = 0) = 0\\ 
        &\mathbb P(Z_{ia} = 1, Z_{i'b} = 0, Z_{i'c} = 1,  Z_{kd} = 0) = \frac{I_a}{I} \times \frac{I_c-I_b}{I-1} \times \frac{I-I_d}{I-2}
        \\ 
        &\mathbb P(Z_{ia} = 1, Z_{ib} = 0, Z_{i'c} = 1,  Z_{kd} = 0) = 0\\ 
        &\mathbb P(Z_{ia} = 1, Z_{i'b} = 0, Z_{ic} = 1,  Z_{kd} = 0) = \frac{I_a}{I} \times \frac{I-I_d}{I-1} \times \frac{I-I_b-1}{I-2}\\ 
        &\mathbb P(Z_{ia} = 1, Z_{i'b} = 0, Z_{kc} = 1,  Z_{id} = 0) = 0\\ 
        &\mathbb P(Z_{ia} = 1, Z_{i'b} = 0, Z_{kc} = 1,  Z_{i'd} = 0) = \frac{I_a}{I} \times \frac{I_c-1}{I-1} \times \frac{I-I_d}{I-2}\\  
        &\text{2 distinct:} \\ 
        &\mathbb P(Z_{ia} = 1, Z_{ib} = 0, Z_{ic} = 1,  Z_{i'd} = 0) = 0 \\ 
        &\mathbb P(Z_{ia} = 1, Z_{ib} = 0, Z_{i'c} = 1,  Z_{id} = 0) = 0 \\ 
        &\mathbb P(Z_{ia} = 1, Z_{i'b} = 0, Z_{ic} = 1,  Z_{id} = 0) = 0\\ 
        &\mathbb P(Z_{ia} = 1, Z_{i'b} = 0, Z_{i'c} = 1,  Z_{i'd} = 0) = 0 \\ 
        &\mathbb P(Z_{ia} = 1, Z_{ib} = 0, Z_{i'c} = 1,  Z_{i'd} = 0) = 0 \\ 
        &\mathbb P(Z_{ia} = 1, Z_{i'b} = 0, Z_{i'c} = 1,  Z_{id} = 0) = 0 \\ 
        &\mathbb P(Z_{ia} = 1, Z_{i'b} = 0, Z_{ic} = 1,  Z_{i'd} = 0) = \frac{I_a}{I}\times\frac{I-I_d}{I-1}\\  &\text{1 distinct:} \\ 
        &\mathbb P(Z_{ia} = 1, Z_{ib} = 0, Z_{ic} = 1,  Z_{id} = 0) = 0. \\
        &\textbf{1001) } \text{All distinct:} \\ 
        &\mathbb P(Z_{ia} = 1, Z_{i'b} = 0, Z_{kc} = 0,  Z_{k'd} = 1) \\ 
        &= \frac{I_a}{I} \times \left(\frac{I - I_d}{I-1} \times \frac{I - I_d - 1}{I-2} \times \frac{I_d - 1}{I-3} + \frac{I_d-I_c}{I-1} \times \frac{I_d - I_b - 1}{I-2} \times \frac{I_d - 3}{I-3}\right)\\ 
        &\text{3 distinct:} \\  
        &\mathbb P(Z_{ia} = 1, Z_{i'b} = 0, Z_{kc} = 0,  Z_{kd} = 1) = \frac{I_a}{I} \times \frac{I_d-I_c}{I-1} \times \frac{I-I_b-1}{I-2}\\ 
        &\mathbb P(Z_{ia} = 1, Z_{i'b} = 0, Z_{i'c} = 0,  Z_{kd} = 1) = \frac{I_a}{I} \times \left(\frac{I-I_d}{I-1} \times \frac{I_d-1}{I-2} + \frac{I_d-I_c}{I-1} \times \frac{I_d-2}{I-2}\right)
        \\ 
        &\mathbb P(Z_{ia} = 1, Z_{ib} = 0, Z_{i'c} = 0,  Z_{kd} = 1) = 0\\ 
        &\mathbb P(Z_{ia} = 1, Z_{i'b} = 0, Z_{ic} = 0,  Z_{kd} = 1) = 0\\ 
        &\mathbb P(Z_{ia} = 1, Z_{i'b} = 0, Z_{kc} = 0,  Z_{id} = 1) = \frac{I_a}{I} \times \frac{I-I_c}{I-1} \times \frac{I-I_b-1}{I-2}\\ 
        &\mathbb P(Z_{ia} = 1, Z_{i'b} = 0, Z_{kc} = 0,  Z_{i'd} = 1) = \frac{I_a}{I} \times \left(\frac{I-I_d}{I-1} \times \frac{I_d-I_b}{I-2} + \frac{I_d-I_c}{I-1} \times \frac{I_d-I_b - 1}{I-2}\right)\\  
        &\text{2 distinct:} \\ 
        &\mathbb P(Z_{ia} = 1, Z_{ib} = 0, Z_{ic} = 0,  Z_{i'd} = 1) = 0 \\ 
        &\mathbb P(Z_{ia} = 1, Z_{ib} = 0, Z_{i'c} = 0,  Z_{id} = 1) = 0 \\ 
        &\mathbb P(Z_{ia} = 1, Z_{i'b} = 0, Z_{ic} = 0,  Z_{id} = 1) = 0\\ 
        &\mathbb P(Z_{ia} = 1, Z_{i'b} = 0, Z_{i'c} = 0,  Z_{i'd} = 1) = \frac{I_a}{I} \times \frac{I_d-I_c}{I-1}\\ 
        &\mathbb P(Z_{ia} = 1, Z_{ib} = 0, Z_{i'c} = 0,  Z_{i'd} = 1) = 0 \\ 
        &\mathbb P(Z_{ia} = 1, Z_{i'b} = 0, Z_{i'c} = 0,  Z_{id} = 1) = \frac{I_a}{I} \times \frac{I-I_c}{I-1} \\ 
        &\mathbb P(Z_{ia} = 1, Z_{i'b} = 0, Z_{ic} = 0,  Z_{i'd} = 1) = 0\\  
        &\text{1 distinct:} \\ 
        &\mathbb P(Z_{ia} = 1, Z_{ib} = 0, Z_{ic} = 0,  Z_{id} = 1) = 0.\\ 
        &\textbf{0101) } \text{All distinct:}\\
        &\mathbb P(Z_{ia} = 0, Z_{i'b} = 1, Z_{kc} = 0,  Z_{k'd} = 1) = \frac{I_b-I_a}{I} \times \frac{I_b-1}{I-1} \times \left(\frac{I-I_d}{I-2}\times \frac{I_d-2}{I-3} + \frac{I_d-I_c}{I-2}\times\frac{I_d-3}{I-3}\right) \\ 
        &+ \frac{I_c-I_b}{I} \times \frac{I_b}{I-1} \times \left(\frac{I-I_d}{I-2}\times \frac{I_d-2}{I-3} + \frac{I_d-I_c}{I-2}\times\frac{I_d-3}{I-3}\right) \\
        &+ \frac{I_d-I_c}{I} \times \frac{I_b}{I-1} \times \left(\frac{I-I_d}{I-2}\times \frac{I_d-2}{I-3} + \frac{I_d-I_c-1}{I-2}\times\frac{I_d-3}{I-3}\right) \\ 
        &+ \frac{I-I_d}{I} \times \frac{I_b}{I-1} \times \left(\frac{I-I_d-1}{I-2}\times \frac{I_d-1}{I-3} + \frac{I_d-I_c}{I-2}\times\frac{I_d-2}{I-3}\right) \\ 
        &\text{3 distinct:} \\  
        &\mathbb P(Z_{ia} = 0, Z_{i'b} = 1, Z_{kc} = 0,  Z_{kd} = 1) =  \frac{I_d-I_c}{I} \times\left(\frac{I_a}{I-1} \times \frac{I-I_a-1}{I-2} + \frac{I_b-I_a}{I-1} \times \frac{I-I_a-2}{I-2}\right)\\ 
        &\mathbb P(Z_{ia} = 0, Z_{i'b} = 1, Z_{i'c} = 0,  Z_{kd} = 1) = 0
        \\ 
        &\mathbb P(Z_{ia} = 0, Z_{ib} = 1, Z_{i'c} = 0,  Z_{kd} = 1) = \frac{I_b-I_a}{I} \times \left(\frac{I-I_d}{I-1} \times \frac{I_d-1}{I-2} + \frac{I_d-I_c}{I-1} \times \frac{I_d-2}{I-2}\right)\\ 
        &\mathbb P(Z_{ia} = 0, Z_{i'b} = 1, Z_{ic} = 0,  Z_{kd} = 1) = \frac{I_b}{I} \times \left(\frac{I-I_d}{I-1} \times \frac{I_d-1}{I-2} + \frac{I_d-I_c}{I-1} \times \frac{I_d-2}{I-2}\right)\\ 
        &\mathbb P(Z_{ia} = 0, Z_{i'b} = 1, Z_{kc} = 0,  Z_{id} = 1) = \frac{I_b-I_a}{I} \times \frac{I_b-1}{I-1} \times \frac{I-I_c}{I-2} + \frac{I_c-I_b}{I} \times \frac{I_b}{I-1} \times \frac{I-I_c}{I-2} \\ 
        &+ \frac{I_d-I_c}{I} \times \frac{I_b}{I-1} \times \frac{I-I_c-1}{I-2}\\ 
        &\mathbb P(Z_{ia} = 0, Z_{i'b} = 1, Z_{kc} = 0,  Z_{i'd} = 1) = \frac{I_a}{I} \times \frac{I-I_c}{I-1} \times \frac{I-I_a-1}{I-2} + \frac{I_b-I_a}{I} \times \frac{I-I_c}{I-1} \times \frac{I-I_a-2}{I-2}\\  &\text{2 distinct:} \\ 
        &\mathbb P(Z_{ia} = 0, Z_{ib} = 1, Z_{ic} = 0,  Z_{i'd} = 1) = 0 \\ 
        &\mathbb P(Z_{ia} = 0, Z_{ib} = 1, Z_{i'c} = 0,  Z_{id} = 1) = \frac{I_b-I_a}{I} \times \frac{I-I_c}{I-1} \\ 
        &\mathbb P(Z_{ia} = 0, Z_{i'b} = 1, Z_{ic} = 0,  Z_{id} = 1) = \frac{I_b}{I}\times\frac{I_d-I_c}{I-1} \\ 
        &\mathbb P(Z_{ia} = 0, Z_{i'b} = 1, Z_{i'c} = 0,  Z_{i'd} = 1) =0\\ 
        &\mathbb P(Z_{ia} = 0, Z_{ib} = 1, Z_{i'c} = 0,  Z_{i'd} = 1) = \frac{I_b-I_a}{I} \times \frac{I_d-I_c}{I-1} \\ 
        &\mathbb P(Z_{ia} = 0, Z_{i'b} = 1, Z_{i'c} = 0,  Z_{id} = 1) = 0 \\ 
        &\mathbb P(Z_{ia} = 0, Z_{i'b} = 1, Z_{ic} = 0,  Z_{i'd} = 1) = \frac{I_b}{I} \times \frac{I-I_c}{I-1}\\  &\text{1 distinct:} \\ 
        &\mathbb P(Z_{ia} = 0, Z_{ib} = 1, Z_{ic} = 0,  Z_{id} = 1) = 0.\\ 
        &\textbf{0110) } \text{All distinct:}\\ 
        &\mathbb P(Z_{ia} = 0, Z_{i'b} = 1, Z_{kc} = 1,  Z_{k'd} = 0) \\ 
        &= \frac{I-I_d}{I}\left(\frac{I_b-I_a}{I-1} \times \frac{I_b-1}{I-2} \times \frac{I_c-2}{I-2}+ \frac{I_c-I_b}{I-1} \times \frac{I_b}{I-2} \times \frac{I_c-2}{I-2} + \frac{I-I_c - 1}{I-1} \times \frac{I_b}{I-2} \times \frac{I_c-1}{I-2}\right) \\ 
        &\text{3 distinct:} \\  
        &\mathbb P(Z_{ia} = 0, Z_{i'b} = 1, Z_{kc} = 1,  Z_{kd} = 0) =  0\\ 
        &\mathbb P(Z_{ia} = 0, Z_{i'b} = 1, Z_{i'c} = 1,  Z_{kd} = 0) = \frac{I-I_d}{I} \times \left(\frac{I_b-I_a}{I-1} \times \frac{I_b-1}{I-2} + \frac{I-I_b-1}{I-1} \times \frac{I_b}{I-2}\right)
        \\ 
        &\mathbb P(Z_{ia} = 0, Z_{ib} = 1, Z_{i'c} = 1,  Z_{kd} = 0) = \frac{I_b-I_a}{I} \times \frac{I_c-1}{I-1} \times \frac{I-I_d}{I-2} \\ 
        &\mathbb P(Z_{ia} = 0, Z_{i'b} = 1, Z_{ic} = 1,  Z_{kd} = 0) = \frac{I_b-I_a}{I} \times \frac{I_b-1}{I-1} \times \frac{I-I_d}{I-2} + \frac{I_c-I_b}{I} \times \frac{I_b}{I-1} \times \frac{I-I_d}{I-2}\\ 
        &\mathbb P(Z_{ia} = 0, Z_{i'b} = 1, Z_{kc} = 1,  Z_{id} = 0) = \frac{I_b}{I} \times \frac{I_c-1}{I-1} \times \frac{I-I_d}{I-2} \\ 
        &\mathbb P(Z_{ia} = 0, Z_{i'b} = 1, Z_{kc} = 1,  Z_{i'd} = 0) = 0\\  
        &\text{2 distinct:} \\ 
        &\mathbb P(Z_{ia} = 0, Z_{ib} = 1, Z_{ic} = 1,  Z_{i'd} = 0) = \frac{I_b-I_a}{I} \times \frac{I-I_d}{I-1} \\ 
        &\mathbb P(Z_{ia} = 0, Z_{ib} = 1, Z_{i'c} = 1,  Z_{id} = 0) = 0 \\ 
        &\mathbb P(Z_{ia} = 0, Z_{i'b} = 1, Z_{ic} = 1,  Z_{id} = 0) =0 \\ 
        &\mathbb P(Z_{ia} = 0, Z_{i'b} = 1, Z_{i'c} = 1,  Z_{i'd} = 0) = 0 \\ 
        &\mathbb P(Z_{ia} = 0, Z_{ib} = 1, Z_{i'c} = 1,  Z_{i'd} = 0) = 0 \\ 
        &\mathbb P(Z_{ia} = 0, Z_{i'b} = 1, Z_{i'c} = 1,  Z_{id} = 0) = \frac{I_b}{I} \times \frac{I-I_d}{I-1} \\ 
        &\mathbb P(Z_{ia} = 0, Z_{i'b} = 1, Z_{ic} = 1,  Z_{i'd} = 0) = 0\\  
        &\text{1 distinct:} \\ 
        &\mathbb P(Z_{ia} = 0, Z_{ib} = 1, Z_{ic} = 1,  Z_{id} = 0) = 0.\\ 
        &\textbf{0011) } \text{All distinct:}\\ 
        &\mathbb P(Z_{ia} = 0, Z_{i'b} = 0, Z_{kc} = 1,  Z_{k'd} = 1) = \\ 
        &\frac{I_c-I_b}{I} \times \left(\frac{I_c-I_a-1}{I-1} \times \frac{I_c-2}{I-2} \times \frac{I_d-3}{I-3} + \frac{I_d-I_c}{I-1} \times \frac{I_c-1}{I-2} \times \frac{I_d-3}{I-3} + \frac{I-I_d}{I-1} \times \frac{I_c-1}{I-2} \times \frac{I_d-2}{I-3}\right) \\ 
        &+\frac{I_d-I_c}{I} \times \left(\frac{I_c-I_a}{I-1} \times \frac{I_c-2}{I-2} \times \frac{I_d-3}{I-3} + \frac{I_d-I_c-1}{I-1} \times \frac{I_c}{I-2} \times \frac{I_d-3}{I-3} + \frac{I-I_d}{I-1} \times \frac{I_c}{I-2} \times \frac{I_d-2}{I-3}\right) \\ 
        &+\frac{I-I_d}{I} \times \left(\frac{I_c-I_a}{I-1} \times \frac{I_c-1}{I-2} \times \frac{I_d-2}{I-3} + \frac{I_d-I_c}{I-1} \times \frac{I_c}{I-2} \times \frac{I_d-2}{I-3} + \frac{I-I_d -1}{I-1} \times \frac{I_c}{I-2} \times \frac{I_d-1}{I-3}\right) \\ 
        &\text{3 distinct:} \\  
        &\mathbb P(Z_{ia} = 0, Z_{i'b} = 0, Z_{kc} = 1,  Z_{kd} = 1) =  \frac{I_a}{I} \times \frac{I-I_b}{I - 1} \times \frac{I-I_a - 1}{I-2} + \frac{I_b-I_a}{I} \times \frac{I-I_b}{I - 1} \times \frac{I-I_a - 2}{I-2} \\ 
        &+ \frac{I_c-I_b}{I} \times \frac{I-I_b-1}{I - 1} \times \frac{I-I_a - 2}{I-2}\\ 
        &\mathbb P(Z_{ia} = 0, Z_{i'b} = 0, Z_{i'c} = 1,  Z_{kd} = 1) = \frac{I-I_d}{I} \times \frac{I_c-I_b}{I - 1} \times \frac{I_d - 1}{I-2} + \frac{I_d-I_c}{I} \times \frac{I_c-I_b}{I - 1} \times \frac{I_d - 2}{I-2} \\ 
        &+ \frac{I_c-I_b}{I} \times \frac{I_c-I_b-1}{I - 1} \times \frac{I_d - 2}{I-2} + \frac{I_b-I_a}{I} \times \frac{I_c-I_b}{I - 1} \times \frac{I_d - 2}{I-2}
        \\ 
        &\mathbb P(Z_{ia} = 0, Z_{ib} = 0, Z_{i'c} = 1,  Z_{kd} = 1) = \frac{I-I_d}{I} \times \frac{I_c}{I - 1} \times \frac{I_d-1}{I-2} \\ 
        &+ \frac{I_d-I_c}{I} \times \frac{I_c}{I - 1} \times \frac{I_d-2}{I-2} \times \frac{I_c - I_a}{I} \times \frac{I_c - 1}{I - 1} \times \frac{I_d-2}{I-2}\\ 
        &\mathbb P(Z_{ia} = 0, Z_{i'b} = 0, Z_{ic} = 1,  Z_{kd} = 1) = \frac{I_b-I_a}{I} \times \left(\frac{I_d- I_b }{I - 1} \times \frac{I_d-2}{I-2} + \frac{I- I_d}{I - 1} \times \frac{I_d-1}{I-2}\right) \\ 
        &+ \frac{I_c-I_b}{I} \times \left(\frac{I_d- I_b - 1 }{I - 1} \times \frac{I_d-2}{I-2} + \frac{I- I_d}{I - 1} \times \frac{I_d-1}{I-2}\right)\\ 
        &\mathbb P(Z_{ia} = 0, Z_{i'b} = 0, Z_{kc} = 1,  Z_{id} = 1) = \frac{I_b-I_a}{I} \times \left(\frac{I_c-I_b}{I-1} \times \frac{I_c-2}{I-2} + \frac{I-I_c}{I-1} \times \frac{I_c-1}{I-2} \right) \\ 
        &+ \frac{I_c-I_b}{I} \times \left(\frac{I_c-I_b-1}{I-1} \times \frac{I_c-2}{I-2} + \frac{I-I_c-1}{I-1} \times \frac{I_c-1}{I-2} \right) \\
        &+ \frac{I_d-I_c}{I} \times \left(\frac{I_c-I_b}{I-1} \times \frac{I_c-1}{I-2} + \frac{I-I_c-1}{I-1} \times \frac{I_c}{I-2} \right) \\ 
        &\mathbb P(Z_{ia} = 0, Z_{i'b} = 0, Z_{kc} = 1,  Z_{i'd} = 1) = \frac{I_a}{I} \times \frac{I_d-I_b}{I-1} \times \frac{I-I_a - 1}{I-2} \\ 
        &+ \frac{I_b-I_a}{I} \times \frac{I_d-I_b}{I-1} \times \frac{I-I_a - 2}{I-2} + \frac{I_c-I_b}{I} \times \frac{I_d-I_b-1}{I-1} \times \frac{I-I_a - 2}{I-2}\\  
        &\text{2 distinct:} \\ 
        &\mathbb P(Z_{ia} = 0, Z_{ib} = 0, Z_{ic} = 1,  Z_{i'd} = 1) = \frac{I_c-I_b}{I} \times \frac{I_d-1}{I-1}\\ 
        &\mathbb P(Z_{ia} = 0, Z_{ib} = 0, Z_{i'c} = 1,  Z_{id} = 1) = \frac{I_c-I_b}{I} \times\frac{I_c-1}{I-1} + \frac{I_d-I_c}{I} \times \frac{I_c}{I-1}\\ 
        &\mathbb P(Z_{ia} = 0, Z_{i'b} = 0, Z_{ic} = 1,  Z_{id} = 1) = \frac{I_c-I_b}{I} \times \frac{I_c-I_a - 1}{I-1} + \frac{I - I_c}{I} \times \frac{I_c-I_a}{I-1} \\ 
        &\mathbb P(Z_{ia} = 0, Z_{i'b} = 0, Z_{i'c} = 1,  Z_{i'd} = 1) = \frac{I_c-I_b}{I} \times \frac{I-I_a - 1}{I - 1} \\ 
        &\mathbb P(Z_{ia} = 0, Z_{ib} = 0, Z_{i'c} = 1,  Z_{i'd} = 1) = \frac{I_c-I_b}{I} \times \frac{I_c-1}{I-1} + \frac{I - I_c}{I} \times \frac{I_c}{I-1}  \\ 
        &\mathbb P(Z_{ia} = 0, Z_{i'b} = 0, Z_{i'c} = 1,  Z_{id} = 1) = \frac{I_c-I_b}{I} \times \frac{I_d-I_a-1}{I-1} \\ 
        &\mathbb P(Z_{ia} = 0, Z_{i'b} = 0, Z_{ic} = 1,  Z_{i'd} = 1) = \frac{I_b-I_a}{I} \times \frac{I_d-I_b}{I-1} + \frac{I_c-I_b}{I} \times \frac{I_d-I_b-1}{I-1}\\  
        &\text{1 distinct:} \\ 
        &\mathbb P(Z_{ia} = 0, Z_{ib} = 0, Z_{ic} = 1,  Z_{id} = 1) = \frac{I_c-I_b}{I}.
    \end{align*}
\end{lemma}

\section{Testing sharp nulls using Fisherian randomization tests}
\label{appendix: sharp null}
\subsection{Null of no effect}
When the number of clusters is small, the $I \to \infty$, fixed $J$ asymptotic regime (as in the REDAPS trial) can be tenuous. It is therefore sensible to test Fisher's sharp null of no effect, at least as a first analysis, for which one can conduct a finite-sample exact $\alpha$ level test. More explicitly, consider the following null hypothesis: 
\begin{equation} \label{eq: fisher sharp null}
    H_0^F: Y_{ijk}(1) = Y_{ijk}(0) \ \forall i,j,k.
\end{equation}
The null $H_0^F$ states that receiving the randomized intervention (a default order for palliative care) has no effect whatsoever for any individual in the study. Under the exclusion restriction of \ref{assumption: iv}, $H_0^F$ can also be interpreted as a test of no effect of receiving palliative care on the outcome. Under the sharp null, $Y_{ijk}(1) = Y_{ijk}(0) = Y_{ijk} \ \forall i,j,k$, and it is straightforward to construct a $p$-value for an arbitrary test statistic $t(\bm{Z}, \bm{Y}, \bm{X})$ as follows:
\begin{equation} \label{eq: frt p-value}
    \mathbb P\{t(\bm{Z}, \bm{Y}, \bm{X}) \geq q | \mathcal{Y}, \mathcal{D}, \mathcal{X}\} \equiv \frac{|\bm{z} \in \Omega: t(\bm{z}, \bm{Y}, \bm{X}) \geq q|}{|\Omega|},
\end{equation}
where $\Omega$ denotes the set of all possible intervention assignments in the stepped-wedge design. Since $\Omega$ can be prohibitively large, the $p$-value can be approximated to the desired precision using Monte-Carlo samples from $\Omega$. \cite{Ji2017RandomizationInsurance} proposes to use the Wald statistic corresponding to the randomized intervention coefficient from fitting an LMM. In principle, any test statistic can be used.

\subsection{Proportional effect model for noncompliance}

Beyond the sharp null of no effect, it is also possible to test the proportional effect null hypothesis, where the effect of the randomized intervention on the outcome is posited to be proportional to the effect of the randomized intervention on the treatment. Mathematically, these can be expressed as follows:
\begin{equation} \label{eq: constant effect null}
    H_0^C: Y_{ijk}(1)-Y_{ijk}(0) = \beta_0 \{D_{ijk}(1)-D_{ijk}(0)\} \ \forall \ i,j,k.
\end{equation}
Note that $H_0^C$ posits a constant effect, and $\beta_0$ matches the effect ratio $\lambda_0$ from the main text. Testing $H_0^C$ is facilitated by the noticing that $Y_{ijk}(1)-\beta_0 D_{ijk}(1) = Y_{ijk}(0) - \beta_0 D_{ijk}(0) = Y_{ijk} - \beta_0 D_{ijk}$, meaning that the vector of values $\bm{Y}-\beta_0\bm{D}$ does not vary with the randomized intervention $\bm{Z}$. Thus, for an arbitrary test statistic $t(\bm{Z},\bm{Y}-\beta_0\bm{D},\bm{X})$ a $p$-value for $H_0^C$ can be constructed as follows:
\begin{equation} \label{constant effect p-value}
    \mathbb P\{t(\bm{Z}, \bm{Y}-\beta_0\bm{D}, \bm{X}) \geq q | \mathcal{Y}, \mathcal{D}, \mathcal{X}\} \equiv \frac{|\bm{z} \in \Omega: t(\bm{z}, \bm{Y}-\beta_0\bm{D}, \bm{X}) \geq q|}{|\Omega|},
\end{equation}
Such hypotheses have been considered by \cite{small2008randomization} in the context of a parallel arm cluster design. Taking $\beta_0 = 0$ recovers the test of the sharp null of no effect.

\section{ANCOVA estimators using cluster-period-level data}
\label{appendix: aggregate estimators}
Recall that
\begin{align} 
    \tau_{\lambda_0} &\equiv \frac{1}{N} \sum_{i=1}^I \sum_{j=1}^J \sum_{k=1}^{N_{ij}} \left[\{Y_{ijk}(1) - \lambda_0D_{ijk}(1)\} - \{Y_{ijk}(0) - \lambda_0D_{ijk}(0)\}\right] \nonumber\displaybreak[0]\\
    &= \frac{1}{N} \sum_{i=1}^I \sum_{j=1}^J \sum_{k=1}^{N_{ij}} \left\{R_{ijk}(1) - R_{ijk}(0)\right\}, \label{eq:R-def-estimand}
\end{align}
where $R_{ijk}(z)=Y_{ijk}(z)-\lambda_0D_{ijk}(z)$ for $z=0,1$. \eqref{eq:R-def-estimand} can be re-expressed as
\begin{align}
    \tau_{\lambda_0} &= \frac{1}{N} \sum_{i=1}^I \sum_{j=1}^J \sum_{k=1}^{N_{ij}} \left\{R_{ijk}(1) - R_{ijk}(0)\right\} = \frac{1}{N} \sum_{j=1}^J \sum_{i=1}^I N_{ij} \cdot \frac{1}{N_{ij}}\sum_{k=1}^{N_{ij}} \left\{R_{ijk}(1) - R_{ijk}(0)\right\} \nonumber \displaybreak[0]\\
    &= \frac{1}{N} \sum_{j=1}^J N_j \cdot \frac{1}{N_j}\sum_{i=1}^I N_{ij} \left\{\overline R_{ij}(1) - \overline R_{ij}(0)\right\} \label{eq:ave-R-def-estimand} \displaybreak[0]\\
    &= \frac{1}{N} \sum_{j=1}^J N_j\cdot \frac{1}{I}\sum_{i=1}^I  \left\{\widetilde R_{ij}(1) - \widetilde R_{ij}(0)\right\} \label{eq:scaled-R-def-estimand},
\end{align}
where $\widetilde R_{ij}(z) \equiv IN_{ij}/N_j \cdot \overline R_{ij}(z)$ is the scaled cluster-period total. \eqref{eq:ave-R-def-estimand} and \eqref{eq:scaled-R-def-estimand} suggest that the ANCOVA estimators can be applied to cluster-period-level data, as discussed in \citet{Su2021Model-AssistedExperiments} for parallel-arm CRTs and \citet{Chen2025} for staggered rollout CRTs with anticipation and without invariance to history. Specifically, \eqref{eq:ave-R-def-estimand} suggests that we can conduct weighted least squares (WLS) regression using cluster-period averages with weights $w_{ij}\equiv N_{ij}$, and \eqref{eq:scaled-R-def-estimand} suggests that we can conduct ordinary least squares (OLS) regression using scaled cluster-period totals.

\subsection{ANCOVA estimators using cluster-period averages}

We first study the ANCOVA estimators using cluster-period averages. The ANCOVA-I model in this case is
\begin{align} \label{eq:ave-ANCOVA-I}
    \overline R_{ij} \sim \beta_j + \theta_j Z_{ij} + \overline X_{ij}^c \eta,
\end{align}
where $\overline X_{ij}^c=N_{ij}^{-1}\sum_{k=1}^{N_{ij}}X_{ijk}-N_j^{-1}\sum_{i=1}^I\sum_{k=1}^{N_{ij}}X_{ijk}=\overline X_{ij} - \overline X_j$ is the centered cluster-period average of covariates. Similarly, the ANCOVA-III model in this case is
\begin{align} \label{eq:ave-ANCOVA-III}
    \overline R_{ij} \sim \beta_j + \theta_j Z_{ij} + \overline X_{ij}^c\gamma + Z_{ij}\overline X_{ij}^c \eta,
\end{align}
which include the interaction between $Z_{ij}$ and $\overline X_{ij}^c$. For ANCOVA-III, \eqref{eq:ave-ANCOVA-III} can be reparameterized to facilitate the derivation of analytical results. That is, we separate the model components by treatment conditions and obtain the following reparameterized ANCOVA-III model:
\begin{equation} \label{eq:ave-ANCOVA-III-reexp}
    \overline R_{ij} \sim (1-Z_{ij})\beta_j +(1-Z_{ij})\overline X_{ij}^c\gamma+ Z_{ij}\theta_j^* + Z_{ij}\overline X_{ij}^c\eta^*,
\end{equation}
where $\theta_j^* = \beta_j+\theta_j$ and $\eta^*=\gamma+\eta$. Similar to the main article, we denote the WLS estimators of $\theta_j$ in \eqref{eq:ave-ANCOVA-I} and $\theta_j = \theta_j^*-\beta_j$ in \eqref{eq:ave-ANCOVA-III-reexp} as $\widehat\theta_j$, and the estimator of $\tau_{\lambda_0}$ is $\widehat\tau_{\lambda_0} = N^{-1}\sum_{j=1}^JN_j\widehat\theta_j$. We now present the finite-population CLT for $\widehat\tau_{\lambda_0}$ using cluster-period averages. 

\begin{proposition} \label{prop:ave-CLT}
    Under Assumptions \ref{assumption: cluster sutva}-\ref{assumption: sw randomization} and analogous regularity conditions in Theorem 1 of \citet{Chen2025}, as $I\to\infty$, $\widehat\tau_{\lambda_0}$ from the WLS fit of \eqref{eq:ave-ANCOVA-I} or \eqref{eq:ave-ANCOVA-III-reexp} is a consistent estimator for $\tau_{\lambda_0}$ and 
    \begin{align*}
        \frac{\widehat\tau_{\lambda_0}-\tau_{\lambda_0}}{\sqrt{\mathrm{Var}(\widehat\tau_{\lambda_0})}} \overset{d}{\to} \mathcal N\left(0,1\right),
    \end{align*}
    where $\mathrm{Var}(\widehat\tau_{\lambda_0}) = \bm\varpi^\top\bm\Sigma_\theta\bm\varpi$ with
    \begin{align} \label{eq:estimator-cov-mat}
        \bm\Sigma_\theta = \sum_{a\in\mathcal A}\frac{1}{I^a}\bm S_{U^*,W}^a-\sum_{a,a'\in\mathcal A}\frac{1}{I}\bm S_{U^*,W}^{a,a'},
    \end{align}
    and $\bm\varpi = (N_1/N,\ldots,N_J/N)^\top$.
\end{proposition}

Similar to Proposition \ref{prop: ancova test validity}, Proposition \ref{prop:ave-CLT} also directly follows from Theorem 1 of \citet{Chen2025} as a special example, with the cluster-period size equal to one and their outcomes $(Y_{ijk})$ replaced by the cluster-period averages of residualized counterparts $\overline R_{ij}$ and the cluster-period-level weight $w_{ij}=N_{ij}$ or $N_{ij}/N_j$. The required conditions are also similar to those of Proposition \ref{prop: ancova test validity}, which are the Lindeberg conditions for the finite-population CLT, no overly dominant cluster-period, and the finiteness and positive definiteness of the limiting values of the covariance matrices such as $\bm S_{U^*,W}^a$ and $\bm S_{U^*,W}^{a,a'}$ in \eqref{eq:estimator-cov-mat}. We omit the formal proof for its length and similarity, and refer the reader to the Web Appendix of \citet{Chen2025} for details.

We now give formal definitions of the covariance matrices in \eqref{eq:estimator-cov-mat}. Recall that we use $A_i=a$, $a\in\mathcal A=\{1,\ldots,J,J+1\}$, to denote the adoption date of the treatment for cluster $i$. Let $G_i^a = \mathbbm 1(A_i=a)$, $\overline R_{ij}^a = N_{ij}^{-1}\sum_{k=1}^{N_{ij}} R_{ijk}^a$, and $\overline R_j^a = N_j^{-1}\sum_{i=1}^I N_{ij}\overline R_{ij}^a$. We write $\overline R_j(1)$ and $\overline R_j(0)$ as a function of $\overline R_j^a$ to highlight the implicit role of the adoption date in the target estimand, where, specifically, we have
\begin{align*}
    \overline R_j(1) = \frac{\sum_{a\in\mathcal A}\mathbbm 1(a\leq j)\overline R_j^a}{\sum_{a\in\mathcal A}\mathbbm 1(a\leq j)} ~~\text{and}~~ \overline R_j(0) = \frac{\sum_{a\in\mathcal A}\mathbbm 1(a>j)\overline R_j^a}{\sum_{a\in\mathcal A}\mathbbm 1(a>j)}.
\end{align*}

Analogous to the definition of $\overline R_{ij}^a$ and $\overline R_j^a$, we define $\overline U_j^a = N_j^{-1}\sum_{i=1}^I N_{ij}\overline U_{ij}^a$, where for ANCOVA-I, $\overline U_{ij}^a = \overline R_{ij}^a - \overline X_{ij}^c\widehat\gamma$; and for ANCOVA-III, $\overline U_{ij}^a = \overline R_{ij}^a - \overline X_{ij}^c\{\mathbbm 1(a\leq j)\widehat\eta^* + \mathbbm 1(a> j)\widehat\gamma\}$. If we further write $\overline u_j^a = \sum_{i=1}^I N_{ij}G_i^a\overline U_{ij}^a/\sum_{i=1}^I N_{ij}G_i^a$ to signify the stochasticity of the treatment assignment, then we have the following re-expressions,
\begin{align*}
    \overline u_j(1) = \frac{\sum_{a\in\mathcal A}\mathbbm 1(a\leq j)(\sum_{i=1}^I N_{ij}G_i^a)\overline u_j^a}{\sum_{a\in\mathcal A}\mathbbm 1(a\leq j)(\sum_{i=1}^I N_{ij}G_i^a)} ~~\text{and}~~
    \overline u_j(0) = \frac{\sum_{a\in\mathcal A}\mathbbm 1(a>j)(\sum_{i=1}^I N_{ij}G_i^a)\overline u_j^a}{\sum_{a\in\mathcal A}\mathbbm 1(a>j)(\sum_{i=1}^I N_{ij}G_i^a)}.
\end{align*}

We proceed to first define the intermediate quantities, $\overline U_{ij}^{a*}$, by assuming that estimates of associated parameter vectors, $\widehat \gamma$ and $\widehat\eta^*$, are substituted by their corresponding known values, $\gamma$ and $\eta^*$, which are WLS coefficient vectors that would be obtained if the full set of potential outcomes is available. That is, for ANCOVA-I, $\overline U_{ij}^{a*} = \overline R_{ij}^a - \overline X_{ij}^c\gamma$; and for ANCOVA-III, $\overline U_{ij}^{a*} = \overline R_{ij}^a - \overline X_{ij}^c\{\mathbbm 1(a\leq j)\eta^* + \mathbbm 1(a> j)\gamma\}$. 
 
Similar to $\overline U_{ij}^{a*}$, we define the intermediate quantities $\overline U_{ij}^*(z)$ by replacing estimates of associated parameter vectors with corresponding known values. Specifically, $\overline U_{ij}^*(z) = \overline R_{ij}(z) - \overline X_{ij}^c\gamma$ for ANCOVA-I; and $\overline U_{ij}^*(1) = \overline R_{ij}(1) - \overline X_{ij}^c\eta^*$ and $\overline U_{ij}^*(0) = \overline R_{ij}(0) - \overline X_{ij}^c\gamma$ for ANCOVA-III.

Then, we can obtain cluster-level intermediate vectors
\begin{align*}
    \overline{\bm U}_i^{a*} = \left\{\overline{\bm U}_i^*(1)-\overline{\bm U}^*(1)\right\}\left\{\otimes_{j=1}^J\mathbbm 1(a\leq j)\right\} + \left\{\overline{\bm U}_i^*(0)-\overline{\bm U}^*(0)\right\}\left\{\otimes_{j=1}^J\mathbbm 1(a> j)\right\},
\end{align*}
where $\overline{\bm U}_i^{a*} = (\overline U_{i1}^{a*}, \ldots, \overline U_{iJ}^{a*})^\top$, $\overline{\bm U}_i^*(z) = (\overline U_{i1}^*(z), \ldots, \overline U_{iJ}^*(z))^\top$, and $\overline{\bm U}^*(z) = (\overline U_1^*(z), \ldots, \overline U_J^*(z))^\top$ with $\overline U_j^*(z) = N_j^{-1}\sum_{i=1}^I N_{ij}\overline U_{ij}^*(z)$; `$\otimes$' is the block diagonal operator, and therefore $\otimes_{j=1}^J\mathbbm 1(a\leq j)$ and $\otimes_{j=1}^J\mathbbm 1(a> j)$ are $J\times J$ diagonal matrices, with the $j$-th diagonal element being $\mathbbm 1(a\leq j)$ and $\mathbbm 1(a> j)$, respectively. In addition, we have $J\times J$ adjusted diagonal weight matrices
\begin{align*}
    \bm W_i^a = \otimes_{j=1}^J\left\{\frac{I^a N_{ij}}{I_j\overline N_j}\mathbbm 1(a\leq j) - \frac{I^a N_{ij}}{(I-I_j)\overline N_j}\mathbbm 1(a>j)\right\},
\end{align*}
where $\overline N_j=I^{-1}\sum_{i=1}^IN_{ij}=I^{-1}N_j$. Finally, we have covariance matrices in \eqref{eq:estimator-cov-mat}, i.e., $\bm S_{U^*,W}^a = (I-1)^{-1}\sum_{i=1}^I\overline{\bm U}_i^{a*}\bm W_i^a\bm W_i^a\overline{\bm U}_i^{a*\top}$ and $\bm S_{U^*,W}^{a,a'} = (I-1)^{-1}\sum_{i=1}^I\overline{\bm U}_i^{a*}\bm W_i^a\bm W_i^{a'}\overline{\bm U}_i^{a'*\top}$.

\subsubsection*{Variance estimation}

The estimation of $\mathrm{Var}(\widehat\tau_{\lambda_0})$ hinges upon that of $\bm\Sigma_\theta$, the covariance matrix of $(\widehat\theta_1,\ldots,\widehat\theta_J)^\top$, since the normalized weight vector, $\bm\varpi$, is known. The second term on the right-hand side (RHS) of \eqref{eq:estimator-cov-mat} is the variance of the average treatment effect estimators across clusters in period $j$, which is generally inestimable because each cluster can only be randomized to a specific treatment adoption time in practice. The first term, the summation of separate covariance matrices of cluster-level average model residuals for different treatment adoption time groups, can be estimated via a consistent design-based (DB) plug-in estimator. Specifically, the covariance matrix component for the group with treatment adoption time $a$, $\bm S_{U^*,W}^a$, can be estimated by
\begin{align} \label{eq:db-var}
    \widehat{\bm S}_{U,W}^a=\frac{1}{I^a-1}\sum_{i:A_i=a}\widehat{\bm U}_i^a\bm W_i^a\bm W_i^a\widehat{\bm U}_i^{a\top},
\end{align}
where $\widehat{\bm U}_i^a = (\widehat U_{i1}^a,\ldots,\widehat U_{iJ}^a)^\top$, and $\widehat U_{ij}^a = \{\overline U_{ij}(1) - \overline u_j(1)\}\{\otimes_{j=1}^J\mathbbm 1(a\leq j)\} + \{\overline U_{ij}(0)-\overline u_j(0)\}\{\otimes_{j=1}^J\mathbbm 1(a> j)\}$. Here, $\overline{U}_{ij}(z) = \overline R_{ij}(z)-\overline X_{ij}^c\widehat\gamma$ and $\overline u_j(z) = \overline r_j(z) - \overline X_j^{c,z}\widehat\gamma$, with $\overline r_j(z) = \sum_{i=1}^I\mathbbm 1(Z_{ij}=z)N_{ij}\overline R_{ij}(z)/\sum_{i=1}^I\mathbbm 1(Z_{ij}=z)N_{ij}$ and $\overline X_j^{c,z} = \sum_{i=1}^I\mathbbm 1(Z_{ij}=z)N_{ij}\overline X_{ij}^c/\sum_{i=1}^I\mathbbm 1(Z_{ij}=z)N_{ij}$, for ANCOVA-I; and $\overline U_{ij}(1) = \overline R_{ij}(1)-\overline X_{ij}^c\widehat\eta^*$, $\overline U_{ij}(0) = \overline R_{ij}(0)-\overline X_{ij}^c \widehat\gamma$, $\overline u_j(1) = \overline r_j(1)-\overline X_j^{c,1}\widehat\eta^*$, and $\overline u_j(0) = \overline r_j(0) - \overline X_j^{c,0}\widehat\gamma$ for ANCOVA-III.

The DB estimator for $\mathrm{Var}(\widehat\tau_{\lambda_0})$ thus is
\begin{align} \label{eq:var-est-db}
    \widehat{\mathrm{Var}}_\mathrm{DB}(\widehat\tau_{\lambda_0}) = \bm\varpi^\top\left(\sum_{a\in\mathcal A}\frac{1}{I^a}\widehat{\bm S}_{U,W}^a\right)\bm\varpi.
\end{align}
Note that since
\begin{align*}
    \bm\varpi^\top\left(\sum_{a,a'\in\mathcal A}\frac{1}{I}\bm S_{U^*,W}^{a,a'}\right)\bm\varpi = \frac{1}{I(I-1)}\sum_{i=1}^I\left(\bm\varpi^\top\sum_{a\in\mathcal A}\overline{\bm U}_i^{a*}\bm W_i^a\right)^2\geq 0,
\end{align*}
and is generally unestimable, the DB estimator in \eqref{eq:var-est-db} is expected to be conservative. Due to the construction of the design-based variance estimator, \eqref{eq:db-var} is undefined when a treatment sequence only contains a single cluster or $I_a=1$ for some $a\in\mathcal A$. To ensure a feasible design-based variance estimator in this special case, \eqref{eq:db-var} can be modified to $\widehat{\bm S}_{U,W}^a=(I^a)^{-1}\sum_{i:A_i=a}\widehat{\bm U}_i^a\bm W_i^a\bm W_i^a\widehat{\bm U}_i^{a^\top}$. The resulting modified variance estimator remains conservative and asymptotically unbiased as $I\to\infty$.

\subsection{ANCOVA estimators using scaled cluster-period totals}

An alternative approach to \eqref{eq:ave-ANCOVA-I} and \eqref{eq:ave-ANCOVA-III-reexp} is to use scaled cluster-period totals in ANCOVA models. Specifically, for ANCOVA-I, we have
\begin{align} \label{eq:total-ANCOVA-I}
    \widetilde R_{ij} \sim \beta_j + \theta_j Z_{ij} + \widetilde X_{ij}^c \eta,
\end{align}
where $\widetilde R_{ij} = IN_{ij}\overline R_{ij}/N_j$ and  $\widetilde X_{ij}^c = IN_{ij}\overline X_{ij}^c/N_j$. Similarly, for ANCOVA-III, we have
\begin{equation} \label{eq:total-ANCOVA-III-reexp}
    \widetilde R_{ij} \sim (1-Z_{ij})\beta_j +(1-Z_{ij})\widetilde X_{ij}^c\gamma+ Z_{ij}\theta_j^* + Z_{ij}\widetilde X_{ij}^c\eta^*,
\end{equation}
where $\theta_j^* = \beta_j+\theta_j$ and $\eta^*=\gamma+\eta$. We denote the OLS estimators of $\theta_j$ in \eqref{eq:total-ANCOVA-I} and $\theta_j = \theta_j^*-\beta_j$ in \eqref{eq:total-ANCOVA-III-reexp} as $\widehat\theta_j$, and the estimator of $\tau_{\lambda_0}$ is $\widehat\tau_{\lambda_0} = N^{-1}\sum_{j=1}^JN_j\widehat\theta_j$. We now present the finite-population Central Limit Theorem (CLT) for $\widehat\tau_{\lambda_0}$ using scaled cluster-period totals. 

\begin{proposition} \label{prop:total-CLT}
    Under Assumptions \ref{assumption: cluster sutva}-\ref{assumption: sw randomization} and analogous regularity conditions in Theorem 1 of \citet{Chen2025}, as $I\to\infty$, $\widehat\tau_{\lambda_0}$ from the OLS fit of \eqref{eq:total-ANCOVA-I} or \eqref{eq:total-ANCOVA-III-reexp} is a consistent estimator for $\tau_{\lambda_0}$ and 
    \begin{align*}
        \frac{\widehat\tau_{\lambda_0}-\tau_{\lambda_0}}{\sqrt{\mathrm{Var}(\widehat\tau_{\lambda_0})}} \overset{d}{\to} \mathcal N\left(0,1\right),
    \end{align*}
    where $\mathrm{Var}(\widehat\tau_{\lambda_0}) = \bm\varpi^\top\bm\Sigma_\theta\bm\varpi$ with
    \begin{align} \label{eq:total-estimator-cov-mat}
        \bm\Sigma_\theta = \sum_{a\in\mathcal A}\frac{1}{I^a}\bm S_{U^*,W}^a-\sum_{a,a'\in\mathcal A}\frac{1}{I}\bm S_{U^*,W}^{a,a'},
    \end{align}
    and $\bm\varpi = (N_1/N,\ldots,N_J/N)^\top$.
\end{proposition}

Similar to Proposition \ref{prop:ave-CLT}, Proposition \ref{prop:total-CLT} is also a special example of Theorem 1 in \cite{Chen2025}, with the cluster-period size equal to one and their outcomes $(Y_{ijk})$ replaced by the scaled cluster-period totals of residualized counterparts $\widetilde R_{ij}$ and the cluster-period-level weight $w_{ij}=1$ for OLS regressions. We omit the formal proof for its length and similarity, and refer the reader to the Web Appendix of \citet{Chen2025} for details.

We now give formal definitions of the covariance matrices in \eqref{eq:total-estimator-cov-mat}. Let $\widetilde R_{ij}^a = IN_{ij}\overline R_{ij}^a/N_j$ and $\widetilde R_j^a = I^{-1}\sum_{i=1}^I \widetilde R_{ij}^a = \overline R_j^a$. We write $\widetilde R_j(1)$ and $\widetilde R_j(0)$ as a function of $\widetilde R_j^a$ to highlight the implicit role of the adoption date in the target estimand, where, specifically, we have
\begin{align*}
    \widetilde R_j(1) = \frac{\sum_{a\in\mathcal A}\mathbbm 1(a\leq j)\widetilde R_j^a}{\sum_{a\in\mathcal A}\mathbbm 1(a\leq j)} ~~\text{and}~~ \widetilde R_j(0) = \frac{\sum_{a\in\mathcal A}\mathbbm 1(a>j)\widetilde R_j^a}{\sum_{a\in\mathcal A}\mathbbm 1(a>j)}.
\end{align*}
Define $\widetilde U_j^a = I^{-1}\sum_{i=1}^I \widetilde U_{ij}^a$, where for ANCOVA-I, $\widetilde U_{ij}^a = \widetilde R_{ij}^a - \widetilde X_{ij}^c\widehat\gamma$; and for ANCOVA-III, $\widetilde U_{ij}^a = \widetilde R_{ij}^a - \widetilde X_{ij}^c\{\mathbbm 1(a\leq j)\widehat\eta^* + \mathbbm 1(a> j)\widehat\gamma\}$. We further write $\widetilde u_j^a = \sum_{i=1}^I G_i^a\widetilde U_{ij}^a/\sum_{i=1}^I G_i^a= \sum_{i=1}^I G_i^a\widetilde U_{ij}^a/I^a$, then, we have,
\begin{align*}
    &\widetilde u_j(1) = \frac{\sum_{a\in\mathcal A}\mathbbm 1(a\leq j)(\sum_{i=1}^I G_i^a)\widetilde u_j^a}{\sum_{a\in\mathcal A}\mathbbm 1(a\leq j)(\sum_{i=1}^IG_i^a)} = \frac{\sum_{a\in\mathcal A}\mathbbm 1(a\leq j)I^a\widetilde u_j^a}{\sum_{a\in\mathcal A}\mathbbm 1(a\leq j)I^a} = \frac{\sum_{a\in\mathcal A}\mathbbm 1(a\leq j)I^a\widetilde u_j^a}{I_j},\displaybreak[0]\\
    &\widetilde u_j(0) = \frac{\sum_{a\in\mathcal A}\mathbbm 1(a>j)(\sum_{i=1}^I G_i^a)\widetilde u_j^a}{\sum_{a\in\mathcal A}\mathbbm 1(a>j)(\sum_{i=1}^I G_i^a)} = \frac{\sum_{a\in\mathcal A}\mathbbm 1(a>j)I^a\widetilde u_j^a}{\sum_{a\in\mathcal A}\mathbbm 1(a>j)I^a} = \frac{\sum_{a\in\mathcal A}\mathbbm 1(a>j)I^a\widetilde u_j^a}{I-I_j}.
\end{align*}
We proceed to define the intermediate quantities, $\widetilde U_{ij}^{a*}$, by assuming that estimates of associated parameter vectors, $\widehat \gamma$ and $\widehat\eta^*$, are substituted by their corresponding known values, $\gamma$ and $\eta^*$, which are OLS coefficient vectors that would be obtained if the full set of potential outcomes is available. That is, for ANCOVA-I, $\widetilde U_{ij}^{a*} = \widetilde R_{ij}^a - \widetilde X_{ij}^c\gamma$; and for ANCOVA-III, $\widetilde U_{ij}^{a*} = \widetilde R_{ij}^a - \widetilde X_{ij}^c\{\mathbbm 1(a\leq j)\eta^* + \mathbbm 1(a> j)\gamma\}$. Similar to $\widetilde U_{ij}^{a*}$, we define the intermediate quantities $\widetilde U_{ij}^*(z)$ by replacing estimates of associated parameter vectors with corresponding known values. Specifically, $\widetilde U_{ij}^*(z) = \widetilde R_{ij}(z) - \widetilde X_{ij}^c\gamma$ for ANCOVA-I; and $\widetilde U_{ij}^*(1) = \widetilde R_{ij}(1) - \widetilde X_{ij}^c\eta^*$ and $\widetilde U_{ij}^*(0) = \widetilde R_{ij}(0) - \widetilde X_{ij}^c\gamma$ for ANCOVA-III. Then, we can obtain cluster-level intermediate vectors
\begin{align*}
    \widetilde{\bm U}_i^{a*} = \left\{\widetilde{\bm U}_i^*(1)-\widetilde{\bm U}^*(1)\right\}\left\{\otimes_{j=1}^J\mathbbm 1(a\leq j)\right\} + \left\{\widetilde{\bm U}_i^*(0)-\widetilde{\bm U}^*(0)\right\}\left\{\otimes_{j=1}^J\mathbbm 1(a> j)\right\},
\end{align*}
where $\widetilde{\bm U}_i^{a*} = (\widetilde U_{i1}^{a*}, \ldots, \widetilde U_{iJ}^{a*})^\top$, $\widetilde{\bm U}_i^*(z) = (\widetilde U_{i1}^*(z), \ldots, \widetilde U_{iJ}^*(z))^\top$, and $\widetilde{\bm U}^*(z) = (\widetilde U_1^*(z), \ldots, \widetilde U_J^*(z))^\top$ with $\widetilde U_j^*(z) = I^{-1}\sum_{i=1}^I \widetilde U_{ij}^*(z)$. In addition, we have $J\times J$ adjusted diagonal weight matrices
\begin{align*}
    \bm W_i^a = \otimes_{j=1}^J\left\{\frac{I^a }{I_j}\mathbbm 1(a\leq j) - \frac{I^a}{I-I_j}\mathbbm 1(a>j)\right\}.
\end{align*}
Finally, we have covariance matrices in \eqref{eq:estimator-cov-mat}, i.e., $\bm S_{U^*,W}^a = (I-1)^{-1}\sum_{i=1}^I\widetilde{\bm U}_i^{a*}\bm W_i^a\bm W_i^a\widetilde{\bm U}_i^{a*\top}$ and $\bm S_{U^*,W}^{a,a'} = (I-1)^{-1}\sum_{i=1}^I\widetilde{\bm U}_i^{a*}\bm W_i^a\bm W_i^{a'}\widetilde{\bm U}_i^{a'*\top}$.

\subsubsection*{Variance estimation}

For variance estimation, the covariance matrix component for the group with treatment adoption time $a$, $\bm S_{U^*,W}^a$ in \eqref{eq:total-estimator-cov-mat}, can be estimated by
\begin{align} \label{eq:total-db-var}
    \widehat{\bm S}_{U,W}^a=\frac{1}{I^a-1}\sum_{i:A_i=a}\widehat{\bm U}_i^a\bm W_i^a\bm W_i^a\widehat{\bm U}_i^{a\top},
\end{align}
where $\widehat{\bm U}_i^a = (\widehat U_{i1}^a,\ldots,\widehat U_{iJ}^a)^\top$, and $\widehat U_{ij}^a = \{\widetilde U_{ij}(1) - \widetilde u_j(1)\}\{\otimes_{j=1}^J\mathbbm 1(a\leq j)\} + \{\widetilde U_{ij}(0)-\widetilde u_j(0)\}\{\otimes_{j=1}^J\mathbbm 1(a> j)\}$. Here, $\widetilde{U}_{ij}(z) = \widetilde R_{ij}(z)-\widetilde X_{ij}^c\widehat\gamma$ and $\widetilde u_j(z) = \widetilde r_j(z) - \widetilde X_j^{c,z}\widehat\gamma$, with $\widetilde r_j(z) = \sum_{i=1}^I\mathbbm 1(Z_{ij}=z)\widetilde R_{ij}(z)/\sum_{i=1}^I\mathbbm 1(Z_{ij}=z)$ and $\widetilde X_j^{c,z} = \sum_{i=1}^I\mathbbm 1(Z_{ij}=z)\widetilde X_{ij}^c/\sum_{i=1}^I\mathbbm 1(Z_{ij}=z)$, for ANCOVA-I; and $\widetilde U_{ij}(1) = \widetilde R_{ij}(1)-\widetilde X_{ij}^c\widehat\eta^*$, $\widetilde U_{ij}(0) = \widetilde R_{ij}(0)-\widetilde X_{ij}^c \widehat\gamma$, $\widetilde u_j(1) = \widetilde r_j(1)-\widetilde X_j^{c,1}\widehat\eta^*$, and $\widetilde u_j(0) = \widetilde r_j(0) - \widetilde X_j^{c,0}\widehat\gamma$ for ANCOVA-III.

The DB estimator for $\mathrm{Var}(\widehat\tau_{\lambda_0})$ thus is
\begin{align} \label{eq:total-var-est-db}
    \widehat{\mathrm{Var}}_\mathrm{DB}(\widehat\tau_{\lambda_0}) = \bm\varpi^\top\left(\sum_{a\in\mathcal A}\frac{1}{I^a}\widehat{\bm S}_{U,W}^a\right)\bm\varpi.
\end{align}
Similarly, the DB estimator in \eqref{eq:total-var-est-db} is expected to be conservative. To ensure a feasible design-based variance estimator in this special case, \eqref{eq:total-db-var} can be modified to $\widehat{\bm S}_{U,W}^a=(I^a)^{-1}\sum_{i:A_i=a}\widehat{\bm U}_i^a\bm W_i^a\bm W_i^a\widehat{\bm U}_i^{a^\top}$. The resulting modified variance estimator remains conservative and asymptotically unbiased as $I\to\infty$.

\section{Cluster-period-level simulation results}
\label{appendix: aggregate sim results}
This section collects the simulation results from using an identical simulation setup as in Section \ref{section: sims} but instead using the average and scaled versions of the ANCOVA estimators. Figures \ref{fig: avg ancova results} and \ref{fig: scaled ancova results} collect the results for the average and scaled ANCOVA estimators, respectively. Figure \ref{fig: all ancova summary results} compares the performance of the best individual and aggregate ANCOVA estimators. The CR0 and CR3 variance estimates (the latter with $t$ correction) are denoted accordingly, and the $S(\lambda)$ variance estimator is denoted by DB. Tables \ref{tab: avg ancova results} and \ref{tab: scaled ancova results} include the comprehensive results for average and scaled ANCOVA. Notably, Figure \ref{fig: all ancova summary results} shows that individual level ANCOVA estimators perform very similarly to the average ANCOVA estimators, which have a somewhat sizable advantage over the scaled ANCOVA estimators in these simulation settings.

\begin{figure}
    \centering
    \includegraphics[width=\linewidth]{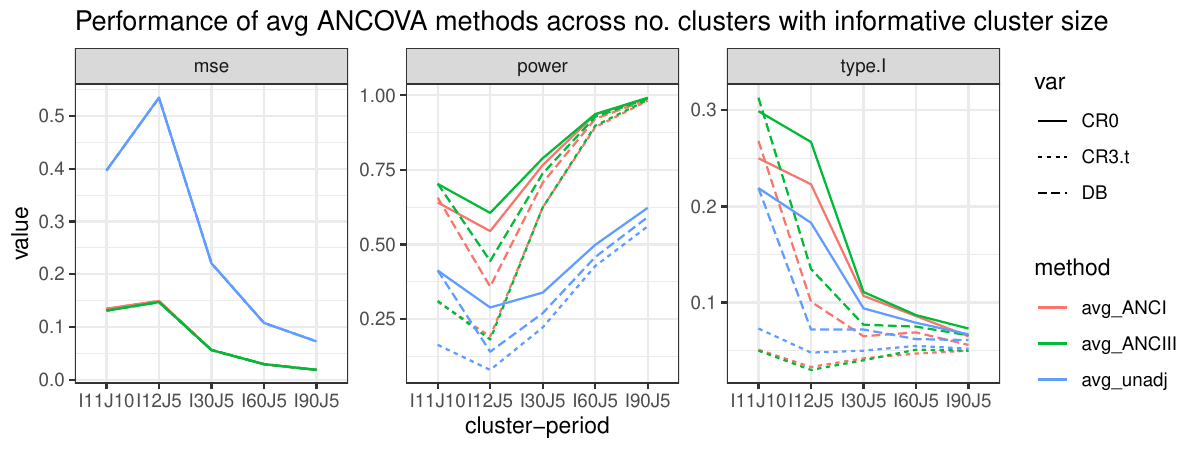}
    \caption{Simulation results for average ANCOVA estimators under different cluster-period configurations and informative cluster size. `mse' is the mean squared error over 1000 iterations. `type.I' counts the proportion of iterations where testing at the true $\lambda_0$ resulted in a rejection. `var' refers to variance/reference distribution - `CR0', `CR3' comparing to a $t$-distribution, and design-based from \cite{Chen2025}. `power' counts the proportion of iterations where the false null $\lambda = 0$ was rejected.}
    \label{fig: avg ancova results}
\end{figure}

\begin{figure}
    \centering
    \includegraphics[width=\linewidth]{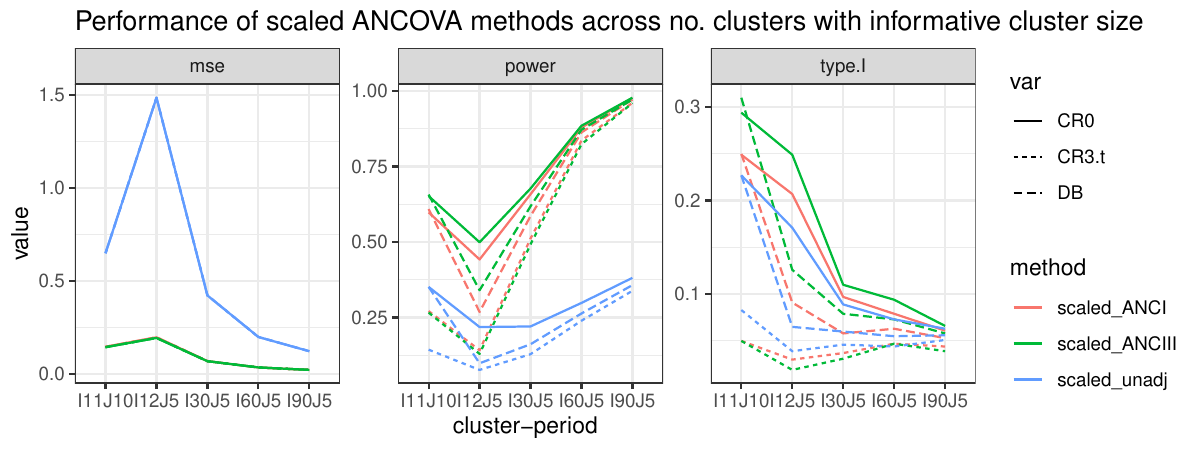}
    \caption{Simulation results for scaled ANCOVA estimators under different cluster-period configurations and informative cluster size. `mse' is the mean squared error over 1000 iterations. `type.I' counts the proportion of iterations where testing at the true $\lambda_0$ resulted in a rejection. `var' refers to variance/reference distribution - `CR0', `CR3' comparing to a $t$-distribution, and design-based from \cite{Chen2025}. `power' counts the proportion of iterations where the false null $\lambda = 0$ was rejected.}
    \label{fig: scaled ancova results}
\end{figure}

\begin{figure}
    \centering
    \includegraphics[width=\linewidth]{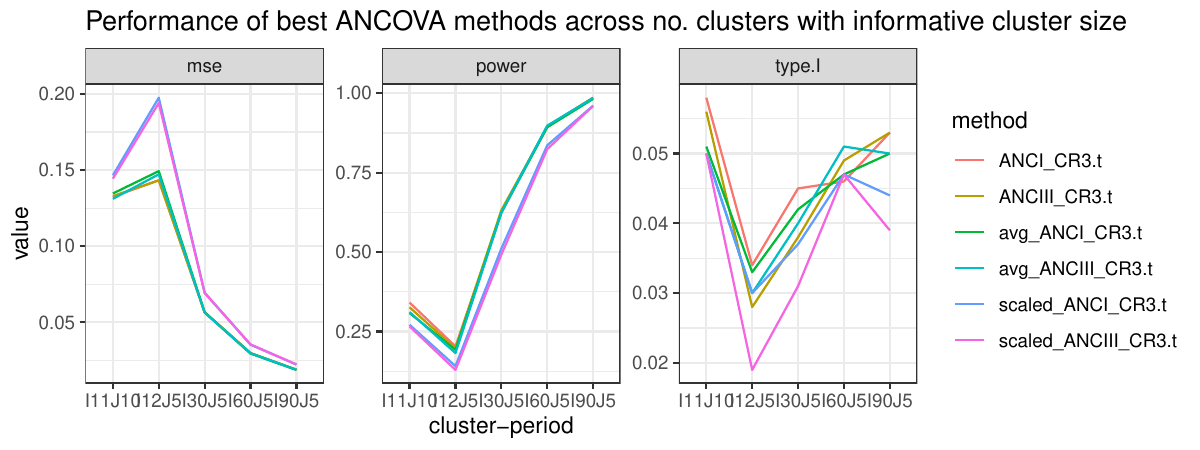}
    \caption{Simulation results for the best ANCOVA estimators under different cluster-period configurations and informative cluster size. `mse' is the mean squared error over 1000 iterations. `type.I' counts the proportion of iterations where testing at the true $\lambda_0$ resulted in a rejection. `var' refers to variance/reference distribution - `CR0', `CR3' comparing to a $t$-distribution, and design-based from \cite{Chen2025}. `power' counts the proportion of iterations where the false null $\lambda = 0$ was rejected.}
    \label{fig: all ancova summary results}
\end{figure}

\begin{sidewaystable}[ht]
    \centering
    \footnotesize
    \begin{tabular}{lllrrrrrrrrrr}
  \hline
IJ & method & inform & bias & mse & CR0.type.I & CR3.type.I & CR3.type.I.t & DB.type.I & CR0.power & CR3.power & CR3.power.t & DB.power \\ 
  \hline
I11J10 & avg\_ANCI & TRUE & -0.014 & 0.135 & 0.250 & 0.078 & 0.051 & 0.268 & 0.641 & 0.413 & 0.308 & 0.657 \\ 
  I11J10 & avg\_ANCIII & TRUE & 0.005 & 0.131 & 0.299 & 0.073 & 0.050 & 0.313 & 0.703 & 0.415 & 0.311 & 0.704 \\ 
  I11J10 & avg\_unadj & TRUE & -0.014 & 0.397 & 0.219 & 0.111 & 0.073 & 0.219 & 0.413 & 0.236 & 0.164 & 0.413 \\ 
  I12J5 & avg\_ANCI & TRUE & -0.044 & 0.149 & 0.223 & 0.060 & 0.033 & 0.101 & 0.545 & 0.249 & 0.192 & 0.359 \\ 
  I12J5 & avg\_ANCIII & TRUE & -0.029 & 0.147 & 0.267 & 0.047 & 0.030 & 0.135 & 0.606 & 0.247 & 0.182 & 0.445 \\ 
  I12J5 & avg\_unadj & TRUE & -0.048 & 0.535 & 0.183 & 0.065 & 0.048 & 0.072 & 0.289 & 0.109 & 0.080 & 0.141 \\ 
  I30J5 & avg\_ANCI & TRUE & -0.004 & 0.056 & 0.107 & 0.048 & 0.042 & 0.065 & 0.766 & 0.650 & 0.624 & 0.707 \\ 
  I30J5 & avg\_ANCIII & TRUE & 0.001 & 0.056 & 0.111 & 0.045 & 0.040 & 0.077 & 0.789 & 0.653 & 0.624 & 0.738 \\ 
  I30J5 & avg\_unadj & TRUE & 0.005 & 0.221 & 0.094 & 0.062 & 0.050 & 0.072 & 0.339 & 0.245 & 0.222 & 0.270 \\ 
  I60J5 & avg\_ANCI & TRUE & -0.002 & 0.029 & 0.086 & 0.053 & 0.047 & 0.069 & 0.936 & 0.899 & 0.893 & 0.921 \\ 
  I60J5 & avg\_ANCIII & TRUE & 0.001 & 0.030 & 0.087 & 0.053 & 0.051 & 0.075 & 0.937 & 0.902 & 0.898 & 0.928 \\ 
  I60J5 & avg\_unadj & TRUE & 0.001 & 0.108 & 0.079 & 0.057 & 0.055 & 0.062 & 0.499 & 0.441 & 0.428 & 0.458 \\ 
  I90J5 & avg\_ANCI & TRUE & -0.001 & 0.019 & 0.065 & 0.052 & 0.050 & 0.056 & 0.990 & 0.984 & 0.982 & 0.986 \\ 
  I90J5 & avg\_ANCIII & TRUE & 0.001 & 0.019 & 0.073 & 0.053 & 0.050 & 0.066 & 0.992 & 0.986 & 0.985 & 0.989 \\ 
  I90J5 & avg\_unadj & TRUE & -0.005 & 0.073 & 0.067 & 0.057 & 0.052 & 0.061 & 0.623 & 0.576 & 0.561 & 0.592 \\ 
  I11J10 & avg\_ANCI & FALSE & -0.012 & 0.125 & 0.212 & 0.062 & 0.034 & 0.225 & 0.330 & 0.130 & 0.080 & 0.342 \\ 
  I11J10 & avg\_ANCIII & FALSE & -0.003 & 0.121 & 0.299 & 0.081 & 0.049 & 0.296 & 0.438 & 0.165 & 0.116 & 0.445 \\ 
  I11J10 & avg\_unadj & FALSE & -0.037 & 0.538 & 0.214 & 0.101 & 0.067 & 0.214 & 0.237 & 0.118 & 0.089 & 0.237 \\ 
  I12J5 & avg\_ANCI & FALSE & -0.037 & 0.137 & 0.191 & 0.039 & 0.022 & 0.076 & 0.262 & 0.085 & 0.056 & 0.139 \\ 
  I12J5 & avg\_ANCIII & FALSE & -0.032 & 0.133 & 0.259 & 0.055 & 0.034 & 0.117 & 0.329 & 0.098 & 0.063 & 0.194 \\ 
  I12J5 & avg\_unadj & FALSE & -0.056 & 0.685 & 0.162 & 0.059 & 0.037 & 0.066 & 0.183 & 0.064 & 0.039 & 0.064 \\ 
  I30J5 & avg\_ANCI & FALSE & -0.003 & 0.053 & 0.080 & 0.034 & 0.028 & 0.052 & 0.281 & 0.185 & 0.165 & 0.228 \\ 
  I30J5 & avg\_ANCIII & FALSE & -0.002 & 0.051 & 0.103 & 0.043 & 0.038 & 0.069 & 0.331 & 0.202 & 0.174 & 0.267 \\ 
  I30J5 & avg\_unadj & FALSE & -0.001 & 0.279 & 0.089 & 0.058 & 0.050 & 0.066 & 0.136 & 0.078 & 0.067 & 0.094 \\ 
  I60J5 & avg\_ANCI & FALSE & -0.001 & 0.027 & 0.067 & 0.039 & 0.032 & 0.048 & 0.399 & 0.326 & 0.314 & 0.357 \\ 
  I60J5 & avg\_ANCIII & FALSE & -0.000 & 0.027 & 0.089 & 0.051 & 0.046 & 0.074 & 0.461 & 0.363 & 0.345 & 0.412 \\ 
  I60J5 & avg\_unadj & FALSE & -0.003 & 0.138 & 0.079 & 0.058 & 0.056 & 0.067 & 0.131 & 0.100 & 0.093 & 0.109 \\ 
  I90J5 & avg\_ANCI & FALSE & -0.001 & 0.017 & 0.051 & 0.040 & 0.036 & 0.047 & 0.497 & 0.458 & 0.445 & 0.473 \\ 
  I90J5 & avg\_ANCIII & FALSE & 0.000 & 0.017 & 0.069 & 0.051 & 0.046 & 0.059 & 0.544 & 0.485 & 0.478 & 0.520 \\ 
  I90J5 & avg\_unadj & FALSE & -0.007 & 0.093 & 0.070 & 0.055 & 0.053 & 0.063 & 0.166 & 0.148 & 0.139 & 0.152 \\ 
   \hline
\end{tabular}
    \caption{Simulation results for averaged cluster-period ANCOVA estimators under different cluster-period configurations. `inform' indicates whether cluster size is informative or not. `bias' and `mse' are the Monte-Carlo bias and mean squared error over 1000 iterations. The `type.I' columns count the proportion of iterations where testing at the true $\lambda_0$ resulted in a rejection using specific variance estimators - `CR0', `CR3', `CR3' comparing to a $t$-distribution, and design-based from \cite{Chen2025}. The `power' columns count the proportion of iterations where the false null $\lambda = 0$ was rejected.}
    \label{tab: avg ancova results}
\end{sidewaystable}

\begin{sidewaystable}[ht]
    \centering
    \footnotesize
    \begin{tabular}{lllrrrrrrrrrr}
  \hline
IJ & method & inform & bias & mse & CR0.type.I & CR3.type.I& CR3.type.I.t & DB.type.I & CR0.power & CR3.power& CR3.power.t & DB.power \\ 
  \hline
I11J10 & scaled\_ANCI & TRUE & -0.014 & 0.147 & 0.249 & 0.084 & 0.050 & 0.250 & 0.598 & 0.352 & 0.271 & 0.609 \\ 
  I11J10 & scaled\_ANCIII & TRUE & -0.009 & 0.144 & 0.294 & 0.069 & 0.050 & 0.310 & 0.652 & 0.333 & 0.266 & 0.657 \\ 
  I11J10 & scaled\_unadj & TRUE & -0.075 & 0.650 & 0.227 & 0.115 & 0.083 & 0.227 & 0.351 & 0.196 & 0.143 & 0.351 \\ 
  I12J5 & scaled\_ANCI & TRUE & -0.055 & 0.198 & 0.207 & 0.051 & 0.030 & 0.091 & 0.442 & 0.199 & 0.141 & 0.268 \\ 
  I12J5 & scaled\_ANCIII & TRUE & -0.045 & 0.194 & 0.249 & 0.032 & 0.019 & 0.126 & 0.499 & 0.174 & 0.129 & 0.340 \\ 
  I12J5 & scaled\_unadj & TRUE & -0.244 & 1.487 & 0.171 & 0.069 & 0.039 & 0.065 & 0.218 & 0.102 & 0.076 & 0.098 \\ 
  I30J5 & scaled\_ANCI & TRUE & -0.008 & 0.069 & 0.097 & 0.049 & 0.037 & 0.058 & 0.658 & 0.535 & 0.511 & 0.586 \\ 
  I30J5 & scaled\_ANCIII & TRUE & -0.004 & 0.069 & 0.110 & 0.046 & 0.031 & 0.079 & 0.677 & 0.523 & 0.492 & 0.618 \\ 
  I30J5 & scaled\_unadj & TRUE & -0.050 & 0.423 & 0.089 & 0.056 & 0.046 & 0.060 & 0.220 & 0.145 & 0.128 & 0.161 \\ 
  I60J5 & scaled\_ANCI & TRUE & -0.008 & 0.035 & 0.079 & 0.055 & 0.047 & 0.063 & 0.879 & 0.844 & 0.836 & 0.863 \\ 
  I60J5 & scaled\_ANCIII & TRUE & -0.006 & 0.035 & 0.094 & 0.054 & 0.047 & 0.073 & 0.885 & 0.838 & 0.824 & 0.872 \\ 
  I60J5 & scaled\_unadj & TRUE & -0.029 & 0.199 & 0.073 & 0.050 & 0.044 & 0.055 & 0.298 & 0.253 & 0.239 & 0.263 \\ 
  I90J5 & scaled\_ANCI & TRUE & -0.002 & 0.022 & 0.061 & 0.046 & 0.044 & 0.053 & 0.977 & 0.964 & 0.961 & 0.972 \\ 
  I90J5 & scaled\_ANCIII & TRUE & -0.000 & 0.022 & 0.066 & 0.041 & 0.039 & 0.058 & 0.978 & 0.963 & 0.960 & 0.972 \\ 
  I90J5 & scaled\_unadj & TRUE & -0.019 & 0.123 & 0.063 & 0.053 & 0.051 & 0.056 & 0.381 & 0.349 & 0.338 & 0.357 \\ 
  I11J10 & scaled\_ANCI & FALSE & -0.052 & 0.245 & 0.201 & 0.065 & 0.046 & 0.222 & 0.256 & 0.102 & 0.056 & 0.273 \\ 
  I11J10 & scaled\_ANCIII & FALSE & -0.046 & 0.242 & 0.240 & 0.064 & 0.040 & 0.274 & 0.284 & 0.094 & 0.068 & 0.341 \\ 
  I11J10 & scaled\_unadj & FALSE & -0.078 & 0.730 & 0.211 & 0.112 & 0.084 & 0.211 & 0.235 & 0.132 & 0.086 & 0.235 \\ 
  I12J5 & scaled\_ANCI & FALSE & -0.153 & 0.600 & 0.180 & 0.049 & 0.034 & 0.084 & 0.207 & 0.061 & 0.039 & 0.100 \\ 
  I12J5 & scaled\_ANCIII & FALSE & -0.147 & 0.586 & 0.206 & 0.042 & 0.025 & 0.115 & 0.219 & 0.045 & 0.035 & 0.121 \\ 
  I12J5 & scaled\_unadj & FALSE & -0.170 & 1.203 & 0.167 & 0.070 & 0.036 & 0.056 & 0.168 & 0.067 & 0.049 & 0.065 \\ 
  I30J5 & scaled\_ANCI & FALSE & -0.035 & 0.167 & 0.097 & 0.049 & 0.042 & 0.058 & 0.146 & 0.091 & 0.078 & 0.110 \\ 
  I30J5 & scaled\_ANCIII & FALSE & -0.035 & 0.168 & 0.103 & 0.042 & 0.036 & 0.078 & 0.155 & 0.083 & 0.070 & 0.128 \\ 
  I30J5 & scaled\_unadj & FALSE & -0.033 & 0.400 & 0.096 & 0.057 & 0.048 & 0.065 & 0.104 & 0.072 & 0.064 & 0.077 \\ 
  I60J5 & scaled\_ANCI & FALSE & -0.019 & 0.085 & 0.086 & 0.063 & 0.057 & 0.076 & 0.189 & 0.150 & 0.142 & 0.168 \\ 
  I60J5 & scaled\_ANCIII & FALSE & -0.018 & 0.085 & 0.091 & 0.066 & 0.063 & 0.083 & 0.203 & 0.152 & 0.140 & 0.180 \\ 
  I60J5 & scaled\_unadj & FALSE & -0.021 & 0.197 & 0.072 & 0.052 & 0.048 & 0.054 & 0.115 & 0.096 & 0.087 & 0.098 \\ 
  I90J5 & scaled\_ANCI & FALSE & -0.008 & 0.049 & 0.058 & 0.045 & 0.044 & 0.048 & 0.218 & 0.192 & 0.187 & 0.208 \\ 
  I90J5 & scaled\_ANCIII & FALSE & -0.008 & 0.049 & 0.064 & 0.048 & 0.045 & 0.057 & 0.227 & 0.190 & 0.183 & 0.218 \\ 
  I90J5 & scaled\_unadj & FALSE & -0.015 & 0.125 & 0.062 & 0.053 & 0.051 & 0.055 & 0.127 & 0.109 & 0.104 & 0.113 \\ 
   \hline
\end{tabular}
    \caption{Simulation results for scaled cluster-period ANCOVA estimators under different cluster-period configurations. `inform' indicates whether cluster size is informative or not. `bias' and `mse' are the Monte-Carlo bias and mean squared error over 1000 iterations. The `type.I' columns count the proportion of iterations where testing at the true $\lambda_0$ resulted in a rejection using specific variance estimators - `CR0', `CR3', `CR3' comparing to a $t$-distribution, and design-based from \cite{Chen2025}. The `power' columns count the proportion of iterations where the false null $\lambda = 0$ was rejected.}
    \label{tab: scaled ancova results}
\end{sidewaystable}

\section{Deferred results regarding practical recommendations and the data analysis}
\subsection{Practical recommendations flowchart}
Figure \ref{fig: recommendation flowchart} illustrates our practical recommendations for the analysis of stepped-wedge designs with individual-level noncompliance. 
\begin{figure}
    \centering    \includegraphics[width=0.8\linewidth]{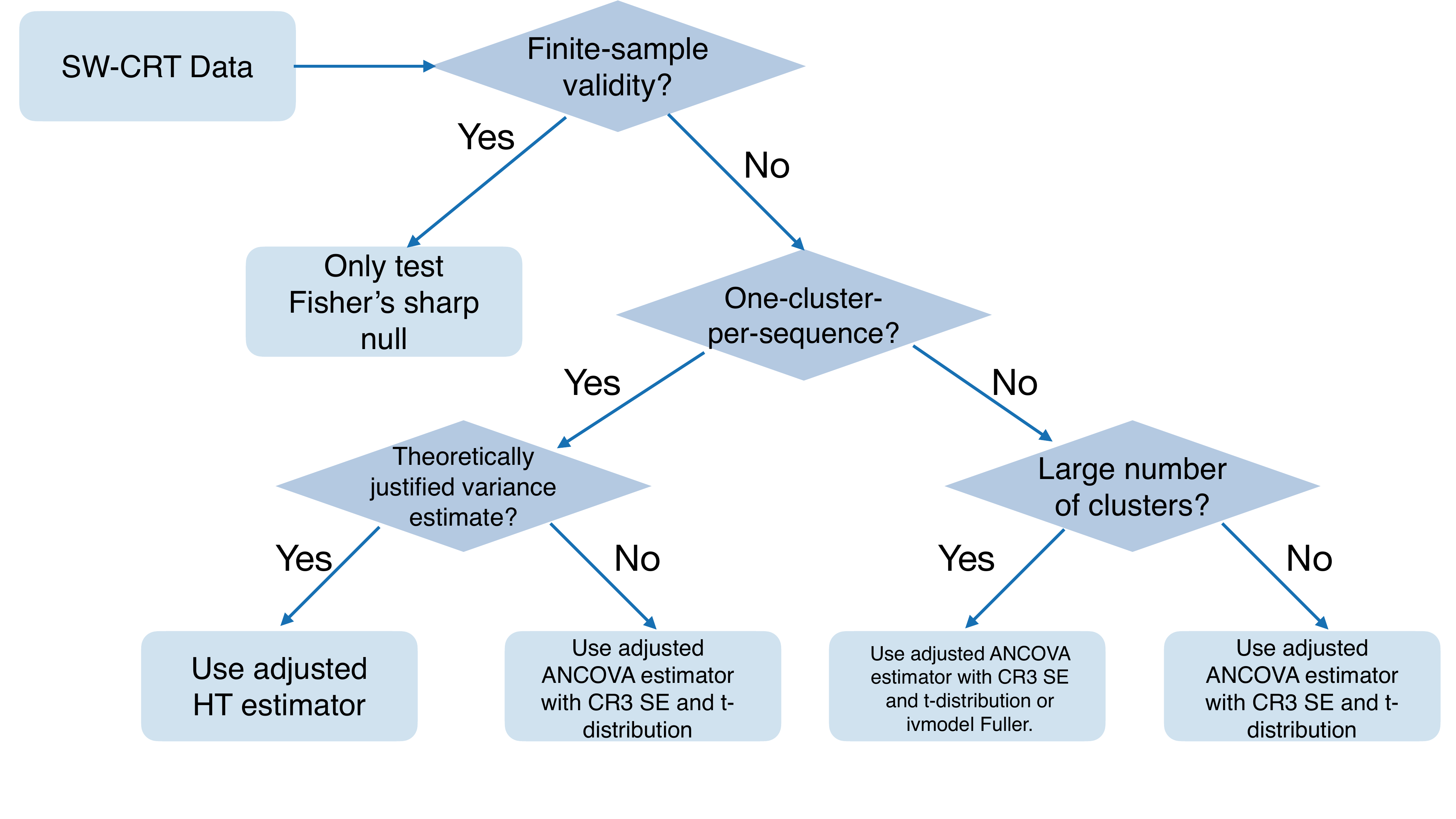}
    \caption{A flowchart detailing estimator and variance estimate recommendations for stepped-wedge designs with individual-level noncompliance.}
    \label{fig: recommendation flowchart}
\end{figure}
\subsection{Heuristic check for treatment duration irrelevance}
Table \ref{tab: test duration irrelevance} reports of the heuristic check coding time on intervention/control as continuous.
\begin{table}[ht]
    \centering
    \caption{Results of tests of the association between time on intervention/control and $D$/$Y$ using logistic/linear regression adjusted for baseline covariates with CR3 cluster-robust standard errors, treating time on intervention/control as a continuous variable.}
    \label{tab: test duration irrelevance}
    \begin{tabular}{ccccccc}
        \toprule
        $D/Y$ & intervention/control & $\beta$ & SE & $t$-statistic & $p$-value \\ 
        \midrule
        $D$ & intervention &  0.085 & 0.087 & 0.983 &  0.505  \\ 
        $Y$ & intervention &  0.001 & 0.004 &  0.175 & 0.900 \\ 
        $D$ & control  & 0.006 & 0.033 &  0.171 & 0.873  \\ 
        $Y$ & control & 0.001 & 0.007 & 0.092 &  0.932 \\ 
        \bottomrule
    \end{tabular}
\end{table}
\subsection{Testing sharp nulls using cluster-period ANCOVA deviates}
\begin{table}[htbp]
    \centering
    \small
    \begin{tabular}{lcc|cc}
        \toprule
        & \multicolumn{2}{c}{Ranked length of stay} & \multicolumn{2}{c}{Readmission within 30 days count} \\
        \cmidrule(lr){2-3} \cmidrule(lr){4-5}
        & 1-sided $p$-value & 2-sided $p$-value & 1-sided $p$-value & 2-sided $p$-value  \\
        \midrule
        avg unadj. ANCOVA dev. & 0.124 & 0.228 &  0.107 & 0.159 \\ 
        avg ANCOVA-I dev. & 0.064 & 0.135 & 0.118 &  0.251  \\ avg ANCOVA-III dev. & 0.160 & 0.317 & 0.011 & 0.024 \\ \midrule
        scaled unadj. ANCOVA dev. & 0.189 & 0.372 &  0.160 & 0.317 \\ 
        scaled ANCOVA-I dev. & 0.103 & 0.213 & 0.162 &  0.332  \\ scaled ANCOVA-III dev. & 0.150 & 0.300 & 0.210 & 0.420\\
        \bottomrule
    \end{tabular}
    \caption{Testing the sharp null in the REDAPS trial using average and scaled ANCOVA deviates with designed-based variance estimators: primary ranked length of stay and secondary readmission count within 30 days outcomes.}
    \label{tab: redaps results aggregate sharp null}
\end{table}

\section{Heavy-tailed simulation results}
\label{appendix: heavy tail sims}
This section collects simulation results under a heavy-tailed setting. The settings considered are strict subset of the simulation settings considered in Section \ref{section: sims}, with the independent $N(0,1)$ noise variables replaced with independent noise drawn from a $t$-distribution with 4 degrees of freedom. Tables \ref{tab: ancova results heavy tail}-\ref{tab: scaled ancova results heavy tail} collect the comprehensive results for the ANCOVA, Horvitz-Thompson, average ANCOVA, and scaled ANCOVA approaches.
\begin{sidewaystable}[ht]
    \centering
    \footnotesize
    \begin{tabular}{lllrrrrrrrrrr}
        \toprule
        IJ & method & inform & bias & mse & CR0.type.I & CR3.type.I & CR3.type.I.t & DB.type.I & CR0.power & CR3.power & CR3.power.t & DB.power \\ 
        \hline
I11J10 & ANCI & TRUE & 0.034 & 0.169 & 0.273 & 0.122 & 0.095 & 0.290 & 0.630 & 0.425 & 0.335 & 0.645 \\ 
  I11J10 & ANCIII & TRUE & 0.032 & 0.133 & 0.223 & 0.075 & 0.040 & 0.258 & 0.660 & 0.407 & 0.292 & 0.682 \\ 
  I11J10 & unadj & TRUE & -0.050 & 0.498 & 0.268 & 0.133 & 0.102 & 0.268 & 0.417 & 0.263 & 0.177 & 0.417 \\ 
  I12J5 & ANCI & TRUE & -0.018 & 0.152 & 0.160 & 0.048 & 0.020 & 0.072 & 0.490 & 0.260 & 0.205 & 0.338 \\ 
  I12J5 & ANCIII & TRUE & -0.002 & 0.164 & 0.180 & 0.033 & 0.022 & 0.107 & 0.547 & 0.235 & 0.182 & 0.380 \\ 
  I12J5 & unadj & TRUE & -0.044 & 0.567 & 0.200 & 0.068 & 0.043 & 0.072 & 0.292 & 0.125 & 0.077 & 0.138 \\ 
  I30J5 & ANCI & TRUE & -0.018 & 0.064 & 0.100 & 0.043 & 0.037 & 0.068 & 0.685 & 0.557 & 0.520 & 0.603 \\ 
  I30J5 & ANCIII & TRUE & 0.011 & 0.064 & 0.100 & 0.055 & 0.052 & 0.072 & 0.770 & 0.623 & 0.590 & 0.700 \\ 
  I30J5 & unadj & TRUE & 0.004 & 0.231 & 0.102 & 0.043 & 0.037 & 0.052 & 0.345 & 0.245 & 0.220 & 0.275 \\ 
  I60J5 & ANCI & TRUE & -0.003 & 0.033 & 0.077 & 0.048 & 0.040 & 0.062 & 0.910 & 0.897 & 0.887 & 0.902 \\ 
  I60J5 & ANCIII & TRUE & 0.002 & 0.031 & 0.062 & 0.037 & 0.030 & 0.050 & 0.915 & 0.895 & 0.885 & 0.907 \\ 
  I60J5 & unadj & TRUE & -0.018 & 0.104 & 0.048 & 0.040 & 0.035 & 0.040 & 0.505 & 0.422 & 0.400 & 0.445 \\ 
   \hline
\end{tabular}
    \caption{Simulation results for ANCOVA estimators under different cluster-period configurations with heavy-tailed outcomes and informative cluster size. `bias' and `mse' are the Monte-Carlo bias and mean squared error over 400 iterations. The `type.I' columns count the proportion of iterations where testing at the true $\lambda_0$ resulted in a rejection using specific variance estimators - `CR0', `CR3', `CR3' comparing to a $t$-distribution, and design-based from \cite{Chen2025}. The `power' columns count the proportion of iterations where the false null $\lambda = 0$ was rejected.}
    \label{tab: ancova results heavy tail}
\end{sidewaystable}

\begin{sidewaystable}[ht]
    \centering
    \begin{tabular}{lllrrrrrr}
  \hline
IJ & method & inform & bias & mse & type.I & cons.type.I & power & cons.power\\ 
  \hline
I11J10 & reg\_adj\_HT & TRUE & -0.015 & 0.181 & 0.022 & 0.000 & 0.128 & 0.025 \\ 
  I11J10 & reg\_adj\_full\_HT & TRUE & 0.002 & 0.136 & 0.297 & 0.160 & 0.700 & 0.530 \\ 
  I11J10 & unadj\_HT & TRUE & -0.045 & 0.553 & 0.167 & 0.000 & 0.429 & 0.003 \\ 
  I12J5 & reg\_adj\_HT & TRUE & -0.023 & 0.186 & 0.015 & 0.005 & 0.133 & 0.040 \\ 
  I12J5 & reg\_adj\_full\_HT & TRUE & -0.006 & 0.147 & 0.215 & 0.102 & 0.655 & 0.482 \\ 
  I12J5 & unadj\_HT & TRUE & -0.202 & 1.508 & 0.500 & 0.013 &  & 0.043 \\ 
  I30J5 & reg\_adj\_HT & TRUE & -0.012 & 0.074 & 0.013 & 0.003 & 0.278 & 0.077 \\ 
  I30J5 & reg\_adj\_full\_HT & TRUE & -0.007 & 0.064 & 0.193 & 0.060 & 0.787 & 0.652 \\ 
  I30J5 & unadj\_HT & TRUE & -0.084 & 0.426 &  & 0.020 &  & 0.080 \\ 
  I60J5 & reg\_adj\_HT & TRUE & -0.002 & 0.035 & 0.007 & 0.000 & 0.588 & 0.163 \\ 
  I60J5 & reg\_adj\_full\_HT & TRUE & -0.011 & 0.031 & 0.135 & 0.050 & 0.955 & 0.890 \\ 
  I60J5 & unadj\_HT & TRUE & -0.009 & 0.173 &  & 0.005 &  & 0.130 \\ 
   \hline
\end{tabular}
    \caption{Simulation results for Horvitz-Thompson estimators under different cluster-period configurations with heavy-tailed outcomes and informative cluster size. Blanks indicate settings where the unadjusted HT estimator produced ill-defined variance estimates over most of the simulation runs. reg\_adj\_full is the adjusted HT estimator using all data for adjustment rather than only the pre and post rollout (reg\_adj). `bias' and `mse' are the Monte-Carlo bias and mean squared error over 400 iterations. The `type.I' columns count the proportion of iterations where testing at the true $\lambda_0$ resulted in a rejection using the provable conservative and not necessarily conservative variance estimators. The `power' columns count the proportion of iterations where the false null $\lambda = 0$ was rejected.}
    \label{tab: ht results heavy tail}
\end{sidewaystable}

\begin{sidewaystable}[ht]
    \centering
    \footnotesize
    \begin{tabular}{lllrrrrrrrrrr}
        \toprule
        IJ & method & inform & bias & mse & CR0.type.I & CR3.type.I & CR3.type.I.t & DB.type.I & CR0.power & CR3.power & CR3.power.t & DB.power \\ 
        \hline
I11J10 & avg\_ANCI & TRUE & -0.020 & 0.178 & 0.285 & 0.117 & 0.077 & 0.325 & 0.615 & 0.410 & 0.335 & 0.632 \\ 
  I11J10 & avg\_ANCIII & TRUE & -0.008 & 0.161 & 0.320 & 0.098 & 0.052 & 0.350 & 0.645 & 0.380 & 0.278 & 0.660 \\ 
  I11J10 & avg\_unadj & TRUE & -0.042 & 0.355 & 0.212 & 0.085 & 0.052 & 0.212 & 0.407 & 0.198 & 0.140 & 0.407 \\ 
  I12J5 & avg\_ANCI & TRUE & -0.006 & 0.172 & 0.200 & 0.055 & 0.040 & 0.090 & 0.525 & 0.223 & 0.177 & 0.328 \\ 
  I12J5 & avg\_ANCIII & TRUE & 0.000 & 0.161 & 0.233 & 0.043 & 0.025 & 0.125 & 0.562 & 0.207 & 0.152 & 0.412 \\ 
  I12J5 & avg\_unadj & TRUE & -0.072 & 0.487 & 0.147 & 0.077 & 0.050 & 0.087 & 0.260 & 0.125 & 0.080 & 0.135 \\ 
  I30J5 & avg\_ANCI & TRUE & 0.014 & 0.066 & 0.110 & 0.058 & 0.050 & 0.083 & 0.730 & 0.562 & 0.535 & 0.662 \\ 
  I30J5 & avg\_ANCIII & TRUE & -0.006 & 0.068 & 0.117 & 0.048 & 0.037 & 0.090 & 0.752 & 0.600 & 0.550 & 0.680 \\ 
  I30J5 & avg\_unadj & TRUE & 0.007 & 0.199 & 0.083 & 0.043 & 0.037 & 0.048 & 0.333 & 0.235 & 0.212 & 0.270 \\ 
  I60J5 & avg\_ANCI & TRUE & 0.018 & 0.030 & 0.080 & 0.058 & 0.052 & 0.068 & 0.938 & 0.920 & 0.915 & 0.930 \\ 
  I60J5 & avg\_ANCIII & TRUE & -0.003 & 0.033 & 0.080 & 0.052 & 0.052 & 0.062 & 0.925 & 0.880 & 0.863 & 0.912 \\ 
  I60J5 & avg\_unadj & TRUE & 0.003 & 0.110 & 0.080 & 0.055 & 0.050 & 0.060 & 0.475 & 0.422 & 0.415 & 0.438 \\ 
   \hline
\end{tabular}
    \caption{Simulation results for average ANCOVA estimators under different cluster-period configurations with heavy-tailed outcomes and informative cluster size. `bias' and `mse' are the Monte-Carlo bias and mean squared error over 400 iterations. The `type.I' columns count the proportion of iterations where testing at the true $\lambda_0$ resulted in a rejection using specific variance estimators - `CR0', `CR3', `CR3' comparing to a $t$-distribution, and design-based from \cite{Chen2025}. The `power' columns count the proportion of iterations where the false null $\lambda = 0$ was rejected.}
    \label{tab: avg ancova results heavy tail}
\end{sidewaystable}

\begin{sidewaystable}[ht]
    \centering
    \footnotesize
    \begin{tabular}{lllrrrrrrrrrr}
        \toprule
        IJ & method & inform & bias & mse & CR0.type.I & CR3.type.I & CR3.type.I.t & DB.type.I & CR0.power & CR3.power & CR3.power.t & DB.power \\ 
        \hline
I11J10 & scaled\_ANCI & TRUE & -0.015 & 0.160 & 0.230 & 0.090 & 0.052 & 0.242 & 0.627 & 0.343 & 0.270 & 0.625 \\ 
  I11J10 & scaled\_ANCIII & TRUE & -0.014 & 0.161 & 0.300 & 0.072 & 0.040 & 0.297 & 0.603 & 0.292 & 0.225 & 0.620 \\ 
  I11J10 & scaled\_unadj & TRUE & -0.013 & 0.549 & 0.205 & 0.102 & 0.068 & 0.205 & 0.330 & 0.182 & 0.130 & 0.330 \\ 
  I12J5 & scaled\_ANCI & TRUE & -0.029 & 0.195 & 0.140 & 0.045 & 0.033 & 0.072 & 0.427 & 0.185 & 0.140 & 0.255 \\ 
  I12J5 & scaled\_ANCIII & TRUE & -0.004 & 0.224 & 0.242 & 0.060 & 0.030 & 0.105 & 0.492 & 0.185 & 0.130 & 0.365 \\ 
  I12J5 & scaled\_unadj& TRUE & -0.186 & 1.671 & 0.152 & 0.060 & 0.050 & 0.068 & 0.240 & 0.105 & 0.065 & 0.098 \\ 
  I30J5 & scaled\_ANCI & TRUE & 0.021 & 0.078 & 0.095 & 0.048 & 0.037 & 0.068 & 0.645 & 0.502 & 0.477 & 0.580 \\ 
  I30J5 & scaled\_ANCIII & TRUE & -0.036 & 0.075 & 0.087 & 0.028 & 0.022 & 0.052 & 0.588 & 0.453 & 0.420 & 0.545 \\ 
  I30J5 & scaled\_unadj & TRUE & -0.098 & 0.457 & 0.090 & 0.045 & 0.043 & 0.050 & 0.198 & 0.158 & 0.133 & 0.163 \\ 
  I60J5 & scaled\_ANCI & TRUE & -0.003 & 0.035 & 0.058 & 0.033 & 0.030 & 0.040 & 0.868 & 0.830 & 0.810 & 0.848 \\ 
  I60J5 & scaled\_ANCIII & TRUE & 0.013 & 0.041 & 0.100 & 0.075 & 0.068 & 0.085 & 0.880 & 0.835 & 0.825 & 0.863 \\ 
  I60J5 & scaled\_unadj & TRUE & -0.024 & 0.189 & 0.068 & 0.037 & 0.030 & 0.040 & 0.292 & 0.245 & 0.230 & 0.250 \\ 
   \hline
\end{tabular}
    \caption{Simulation results for scaled ANCOVA estimators under different cluster-period configurations with heavy-tailed outcomes and informative cluster size. `bias' and `mse' are the Monte-Carlo bias and mean squared error over 400 iterations. The `type.I' columns count the proportion of iterations where testing at the true $\lambda_0$ resulted in a rejection using specific variance estimators - `CR0', `CR3', `CR3' comparing to a $t$-distribution, and design-based from \cite{Chen2025}. The `power' columns count the proportion of iterations where the false null $\lambda = 0$ was rejected.}
    \label{tab: scaled ancova results heavy tail}
\end{sidewaystable}

\section{Details on the sensitivity analysis}
\label{appendix: sensitivity}
We first introduce some additional notation. Let $\pi_C$ and $\pi_D$ denote the proportion of compliers and defiers in the study population, respectively. That is,
\begin{equation*}
\begin{aligned}
    \pi_C &= \frac{1}{N} \sum_{i=1}^I \sum_{j=1}^J \sum_{k=1}^K \mathbbm{1}(D_{ijk}(1) = 1,  D_{ijk}(0) = 0), \\ \pi_D &= \frac{1}{N} \sum_{i=1}^I \sum_{j=1}^J \sum_{k=1}^K \mathbbm{1}(D_{ijk}(0) = 1,  D_{ijk}(1) = 0).
\end{aligned}
\end{equation*}
It is clear to see that $\tau_D = \pi_C - \pi_D$. Next, we define sample average treatment effects among the complier and defier populations:
\begin{equation*}
\begin{aligned}
    \text{SATE}_C &= \frac{1}{N \pi_C} \sum_{i=1}^I \sum_{j=1}^J \sum_{k=1}^K \mathbbm{1}(D_{ijk}(1) = 1,  D_{ijk}(0) = 0)[Y_{ijk}(1)-Y_{ijk}(0)], \\ \text{SATE}_D &= \frac{1}{N \pi_D} \sum_{i=1}^I \sum_{j=1}^J \sum_{k=1}^K \mathbbm{1}(D_{ijk}(0) = 1,  D_{ijk}(1) = 0)[Y_{ijk}(1)-Y_{ijk}(0)].
\end{aligned}
\end{equation*}
Then by the exclusion restriction from Assumption \ref{assumption: iv}, we see that $\tau_Y = \pi_C \text{SATE}_C  + \pi_D \text{SATE}_D $. Moreover, recall we define the sensitivity parameter to the be the proportion of defiers to compliers, i.e. $\gamma = \pi_D / \pi_C$. We wish to bound $\text{SATE}_C$ with the sensitivity parameter and estimable quantities. Simple rearrangement and utilizing the identities $\pi_D/\pi_C = \gamma$ and $\pi_C-\pi_D = \tau_D$ yields
\begin{equation*}
\begin{aligned}
\text{SATE}_C &= \frac{\tau_Y - \pi_D \text{SATE}_D}{\pi_C} \\ &= \frac{\tau_Y - \gamma / (1-\gamma) \tau_D \text{SATE}_D}{\tau_D/(1-\gamma)}  \\ &= \frac{(1-\gamma)\tau_Y - \gamma \tau_D \text{SATE}_D}{\tau_D} \\ &= \frac{(1-\gamma)\tau_Y}{\tau_D} - \gamma \text{SATE}_D \\ &= (1-\gamma)\lambda - \gamma \text{SATE}_D.
\end{aligned}
\end{equation*}

\section{Miscellaneous remarks}
\label{appendix: misc}
\subsection{Stepped-wedge design graphic}
\begin{figure}[htbp]
    \centering
    \includegraphics[width=\linewidth]{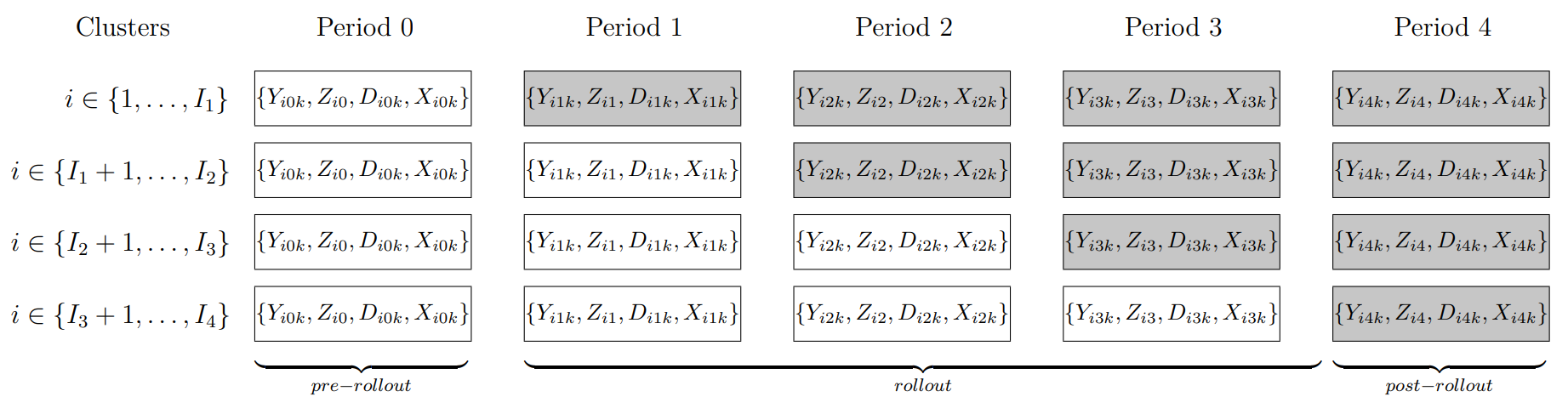}
    \caption{An example schematic and data structure of a standard stepped-wedge design with $I=I_4$ clusters and $J=3$ rollout periods (plus one pre-rollout and one post-rollout). There are a total of $4$ possible treatment adoption time points, and a subset of all clusters, represented by each row, has a unique treatment adoption time. We include in each cell an example observation consisting of outcome, assignment, treatment receipt, and baseline covariates for an individual $k\in \{1,\ldots,N_{ij}\}$. A white cell indicates the control condition, and a shaded cell indicates the treatment condition. }
    \label{fig:stepped}
\end{figure}
\subsection{Simpson's paradox}
\cite{Zhang2022BridgingRatio} demonstrate that the effect ratio (their pooled effect ratio) can exhibit a form of Simpson's paradox in the sense that it is not guaranteed to preserve trends in cluster-level effect ratios. A similar phenomenon cannot be ruled out in the stepped-wedge setting, but under Assumption \ref{assumption: iv}, the effect ratio can be safely interpreted as the sample average treatment effect among the compliers.

\subsection{Adapting the approach of \citet{ren2024model}}

\citet{ren2024model} proposed a computationally efficient way to estimate the variance of Wald-style estimators in the designed-based framework. The method proposed by \citet{ren2024model} is designed for parallel, individually randomized trials and for specific estimators. It appears nontrivial to extend the approach to our setting with clusters, multiple time periods, and vastly different estimators, though such an endeavor is a promising direction for future research.

\end{document}